\pdfoutput=1
\documentclass[acmsmall]{acmart}
\acmDOI{10.1145/3720500}
\acmYear{2025}
\acmJournal{PACMPL}
\acmVolume{9}
\acmNumber{OOPSLA1}
\acmArticle{135}
\acmMonth{4}
\received{2024-10-15}
\received[accepted]{2025-02-18}
\setcopyright{cc}
\setcctype{by}

\usepackage{booktabs}   %% For formal tables:
\usepackage{subcaption} %% For complex figures with subfigures/subcaptions
\usepackage{meu, proof}

\newcommand{\appcite}[1]{the appendix}

\begin{document}

\pgfplotstableread[col sep = comma]{data/hmm-data.csv}\hmmdata
\pgfplotstableread[col sep = comma]{data/ladder-long.csv}\ladderlongdata
\pgfplotstableread[col sep = comma]{data/ladder4.csv}\ladderfour
\pgfplotstableread[col sep = comma]{nested-pineappl/pineappl-nested.csv}\Nested
\pgfplotstableread[col sep = comma]{nested-pineappl/pineappl-nested-fit.csv}\Nestedfit

%% Title information
\title{Scaling Optimization over Uncertainty via Compilation}
\title{Scaling Optimization over Uncertainty via Compilation}

\author{Minsung Cho}
\orcid{0009-0006-6170-6033}
\affiliation{%
  \institution{Northeastern University}
  \city{Boston}
  \country{USA}
}
\email{minsung@ccs.neu.edu}

\author{John Gouwar}
\orcid{0000-0003-0494-7245}
\affiliation{%
  \institution{Northeastern University}
  \city{Boston}
  \country{USA}
}
\email{gouwar.j@northeastern.edu}

\author{Steven Holtzen}
\orcid{0000-0002-8190-5412}
\affiliation{%
  \institution{Northeastern University}
  \city{Boston}
  \country{USA}
}
\email{s.holtzen@northeastern.edu}
%% Abstract
%% Note: \begin{abstract}...\end{abstract} environment must come
%% before \maketitle command

%% 2012 ACM Computing Classification System (CSS) concepts
%% Generate at 'http://dl.acm.org/ccs/ccs.cfm'.

%% End of generated code

%% Keywords
%% comma separated list
  %% \keywords are mandatory in final camera-ready submission

%% \maketitle
%% Note: \maketitle command must come after title commands, author
%% commands, abstract environment, Computing Classification System
%% environment and commands, ana keywords command.

\begin{abstract}
Probabilistic inference is fundamentally hard, yet many tasks require
optimization on top of inference, which is even harder.  We present a new
\textit{optimization-via-compilation} strategy to scalably solve a certain
class of such problems.  In particular, we introduce a new intermediate
representation (IR), binary decision diagrams weighted by a novel notion of
\textit{branch-and-bound semiring}, that enables a scalable branch-and-bound
based optimization procedure. This IR automatically \textit{factorizes}
problems through program structure and \textit{prunes} suboptimal values via a
straightforward branch-and-bound style algorithm to find optima.
Additionally, the IR is naturally amenable to \textit{staged compilation},
allowing the programmer to query for optima mid-compilation to inform further
executions of the program.  We showcase the effectiveness and flexibility of
the IR by implementing two performant languages that both compile to it:
$\dappl$ and $\pineappl$.  $\dappl$ is a functional language that solves
maximum expected utility problems with first-class support for rewards,
decision making, and conditioning.  $\textsc{pineappl}$ is an imperative
language that performs exact probabilistic inference with support for nested
marginal maximum a posteriori (MMAP) optimization via staging.
\end{abstract}

\maketitle

\section{Introduction}\label{sec:introduction}

The Achilles' heel of probabilistic programming languages (PPLs) is \textit{scalability}.
The primary task of probabilistic programs,
probabilistic inference,
is \#P-hard~\citep{roth1996hardness} even when restricted to only Boolean random variables,
which amounts to counting accepting inputs
for an NP-complete problem.
Intuitively, this complexity stems from a \textit{state-space explosion}:
there are exponentially many probabilistic outcomes in the number of random variables,
and one must add up the probability of an
arbitrarily large subset of these outcomes to perform inference.

Monumental strides have been taken to make PPLs scalable.
One such stride is the development of the \textit{reasoning-via-compilation} scheme,
which is currently the
state-of-the-art approach for exact inference for many kinds of
probabilistic programs and graphical models~\citep{holtzen2020scaling,fierens2015inference,li2023scallop}.
The essence of reasoning-via-compilation is to identify \emph{tractable target
languages} that (1) support efficient reasoning, and (2) exploit
program structure to scale.
Tractable target languages capture a class of tractable problem
instances:
for example, in probabilistic
inference, knowledge compilation data-structures like binary decision diagrams (BDDs),
despite their inexpressiveness as a language~\citep{darwiche2002knowledge},
have proven very successful in practice as tractable targets because
they scale by exploiting conditional independence, a property that is abundant in many
real-world probabilistic programs~\citep{holtzen2020scaling}.
% One such stride was the \textit{factorization} of probabilistic inference, as
% defined as a probabilistic program,
% via program structure~\citep{holtzen2020scaling,de2007problog,saad2021sppl}.
% By compiling a (discrete) probabilistic program into a probability-weighted
% Boolean formula, one can take advantage of tractable, factored representations of
% the Boolean formula to factorize inference.
% \citet{holtzen2020scaling} demonstrated that binary decision diagrams (BDDs),
% a certain tractable representation of Boolean formulae,
% naturally exploit repeated sub-structure in the program as conditional independence
% in order to scale. Compiling to a BDD can be expensive, but once it is compiled,
% probabilistic inference becomes \textit{efficient in the size of the BDD},
% which allows inference to scale.

But, many practical real-world problems require additional reasoning \textit{on top of} inference,
increasing the complexity of an already hard problem.
In this paper, we focus on the additional task of \textit{optimization over inference},
in which an objective function over probabilistic inference must be optimized.
Such problems are ubiquitous and have been studied through the lenses of
game theory~\citep{osborne2004introduction},
probabilistic graphical models~\citep{howard2002comments,puterman1978modified,koller2009probabilistic},
reinforcement learning~\citep{busoniu2008comprehensive},
and beyond.
Often, such tasks require \textit{meta-reasoning}, or nested reasoning,
in which computed optimal values inform the next step of inference, which
serves to increase the complexity of reasoning about such problems~\citep{lew2023probabilistic,rainforth2018nesting,zhang2022reasoning}.
Despite their inherent difficulty,
optimization problems over inference have had broad applicability in
medical diagnosis~\citep{heckerman1992toward,lee2014applying},
image segmentation~\citep{bioucas2016bayesian},
and AI planning~\citep{kiselev2014policy}.

The high complexity of optimization over inference has two root causes.
The first is the state-space explosion as described before.
The second is \textit{search-space explosion}: to find the optimal value,
in the worst case one must traverse and compare all possible values the objective function can take,
which often causes significant blowup.
Indeed, the complexity of the two optimization problems we will study in this paper,
maximum expected utility (MEU) and marginal maximum a posteriori (MMAP),
is $NP^{PP}$-hard, so it is still $NP$-hard even with a probabilistic polynomial-time ($PP$) oracle
to perform fast inference~\citep{maua2016equivalences}.

MEU and MMAP are examples of \emph{discrete finite-horizon decision-making
problems with deterministic policies}.
Such decision-making problems are quite common in diagnosis and
planning, and have typically been represented using decision-theoretic Bayesian networks~\citep[Ch. 16]{russell2016artificial}
and influence
diagrams~\citep{howard2005influence,sanner2010relational}.
Despite their intractability, they are remarkably simple, lacking features such as loops and continuous random variables,
differentiating them from related decision-making problems under uncertainty such as
Markov decision processes (MDPs)~\citep{sutton2018reinforcement} or
optimal value-of-information problems where the goal
is to decide what kinds of events to observe~\citep[\S 16.6]{russell2016artificial}.

\textit{What is an effective target language to express problems such as MEU and MMAP}?
Generalizing the reasoning-via-compilation perspective,
we present \textit{optimization-via-compilation}, a compilation scheme
supporting efficient probabilistic inference and, additionally, \textit{efficient pruning} of non-optimal values,
at the cost of no builtin loops or continuous random variables \'a la BDDs.
Our new tractable target language, which we call the \emph{branch-and-bound
intermediate representation} (BBIR),
\textit{factorizes} the state space of a probabilistic program to manage state-space explosion
and \textit{prunes} the search space
via a branch-and-bound approach to manage search-space explosion.
The compilation in BBIR can also be \textit{staged}, in which a partially compiled BBIR can
be queried for optimal values to be used further along in compilation, allowing for
\textit{meta-optimization}.
This culminates in $\dappl$ and $\pineappl$,
simple discrete probabilistic languages with bounded loops
expressing MEU and MMAP problems,
that demonstrate the performance and generality of optimization-via-compilation, as laid out by~\cref{fig:bb-overview}.
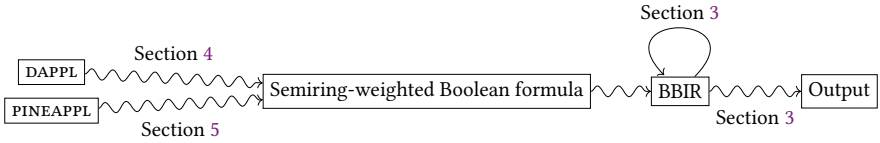
\begin{figure}
  \centering
  \scalebox{0.8}{
  \begin{tikzpicture}[decoration=snake]
    \node[draw, above=3pt] (dappl) {\dappl{}};
    \node[draw, below=3pt] (pineappl) {\pineappl{}};
    \node[draw, right = 100pt] (wbf) {Semiring-weighted Boolean formula};
    \node[draw, right = of wbf] (bbir) {BBIR};
    \node[draw, right = of bbir, xshift=15pt] (out) {Output};
    \draw[->, decorate] (dappl) -- (wbf) node[midway, above=5pt] {\cref{sec:dappl}};
    \draw[->, decorate] (pineappl) -- (wbf) node[midway, below=5pt] {\cref{sec:pineappl}};
    \draw[->, decorate] (wbf) -- (bbir) node[midway, above] {};
    \draw[->] (bbir) edge [loop] (bbir) node[right, above=30pt] {\cref{sec:bbir}};
    \draw[->, decorate] (bbir) -- (out) node[midway, below=5pt] {\cref{sec:bbir}};
  \end{tikzpicture}
  }
  \caption{Overview of the optimization-via-compilation scheme and associated sections of the paper. }
  \label{fig:bb-overview}
\end{figure}
In sum, we make the following contributions:
\begin{itemize}[leftmargin=*]
  \item (\cref{sec:bbir}): We identify a new intermediate
  representation for solving max-over-sum problems called the
  \emph{branch-and-bound intermediate representation} (BBIR).  The key feature
  of BBIR is that it supports efficient (i.e., polynomial-time) computation of
  upper-bounds of partially computed values of the objective function,
  at the cost of lacking support for dynamically bounded loops, almost surely-terminating loops and continuous random variables.
  We show how the BBIR can represent optimization problems over inference, and how BBIR admits an algorithm that uses
  its efficient upper bounds to find optima via pruning.
  \item (\cref{sec:dappl}): We develop \dappl{}, a discrete-valued
  functional decision-theoretic probabilistic programming language with Bayesian
  conditioning. We give a semantics-preserving
  compilation scheme from \dappl{} to BBIR, and prove it correct.
  \item (\cref{sec:pineappl}): We develop $\pineappl$, a discrete-valued
  imperative probabilistic programming language in which MMAP queries, a meta-optimization query,
  are a first-class primitive
  on top of inference. This mid-program optimization is performed using
  staged compilation~\citep{chambers2002staged,rompf2010lightweight,devito2013terra} and
  querying of partially compiled BBIR, which we again prove sound with respect to
  the semantics of $\pineappl$.
  \item (\cref{sec:eval}): We empirically validate the effectiveness of
  our optimization-via-compilation strategy and show that it
  outperforms existing approaches to solving MEU and MMAP in discrete probabilistic programs
  while simultaneously supporting the novel feature of meta-optimization.
\end{itemize}

\section{Overview}\label{sec:overview}

First, we will define formally the MEU and MMAP problems
as well as demonstrate the core ideas behind BBIR via two illustrative examples.
The first example in $\dappl$ (\cref{subsec:dappl-overview}) will show
the generalization of the \textit{reasoning-via-compilation} scheme
to lattice semirings, BBIR's theoretical foundation.
The second example in $\pineappl$ (\cref{subsec:pineappl-overview}) will illustrate
how we can model mid-program optimization through the BBIR via staging.

\subsection{The Maximum Expected Utility Problem}\label{subsec:dappl-overview}
In this section, we first introduce the maximum expected utility (MEU) problem
through example (\cref{subsubsec:meu-example}). Then we describe our
approach to solving MEU via compilation (\cref{subsubsec:eu-via-compilation}
and \cref{subsubsec:optimization-via-compilation}).

\subsubsection{Defining MEU}\label{subsubsec:meu-example}

\begin{wrapfigure}{r}{0.3\linewidth}
  \begin{dapplcodeblock}[basicstyle=\tiny\ttfamily]
rainy <- flip 0.1;
// observe rainy ;
choose [Umb, No_umb]
| Umb -> if rainy then
    reward 10 else reward -5
| No_umb -> if rainy then
    reward -100 else ()
\end{dapplcodeblock}
\caption{Example $\dappl$ program.}
\label{fig:motivation-dappl}
\end{wrapfigure}
Consider the following simple decision-making scenario that we model as a $\dappl$ program
in~\cref{fig:motivation-dappl}.

\begin{quote}
\textit{``Today there is a 10\% chance of rain.
  If it rains and you have your umbrella, you are dry and happy.
  If it rains and you do not have your
  umbrella, you are very unhappy. However, you prefer not to carry your umbrella,
  so you are mildly annoyed if it does not rain and you brought your umbrella.
  Should you bring your umbrella?''}
\end{quote}

\Cref{fig:motivation-dappl} shows how we encode this scenario in \dappl{}.
On Line 1 (indicated on the right of~\cref{fig:motivation-dappl}),
we model the fact that there is a $10\%$ chance of rain via the syntax
\dapplcode{flip 0.1}, which outputs $\tt$ with probability $0.1$ and $\ff$
otherwise; in \dappl{}, all random variables are finite and discrete. The syntax
\dapplcode{choose [Umb, No\_umb]} on Line 3 denotes
a non-deterministic \emph{choice} about whether
or not to bring an umbrella; similar to random variables, all choices must be
finite and discrete. On Lines 5 and 7, we assign rewards to specific outcomes with
the \dapplcode{reward} keyword, which is an effectful operation that
accumulates a reward when it is executed: in this case, the outcome of ``it is raining and I brought
my umbrella'' is assigned a reward of $10$ and the outcome of
``it is not raining and I brought my umbrella'' is assigned a reward of $-5$.

The goal of a \dappl{} program is to compute the assignment to all choices --
i.e., the \emph{policy} -- that maximizes the expected accumulated reward.
In \dappl{}, all policies are deterministic.
For
the program in~\cref{fig:motivation-dappl}, there are two possible policies: $\pi_1
= \texttt{Umb}$, where the umbrella is taken, and $\pi_2 =
\texttt{No\_umb}$ where it is not.  Given a policy $\pi$, we let
$e|{_\pi}$ denote the \dappl{} program that results from substituting all
choices in the program $e$ for their corresponding policies $\pi$.  Then, we can
define an evaluation function $\EU(e|_{\pi}) = v$ for a \dappl{} program $e$
under a fixed specific policy $\pi$ that yields the expected utility $v$ of that
policy; we make this precise in Section~\ref{subsec:util}.
With this in mind, we can compute the maximum expected utility of our example,
call it $ex$, by comparing the expected utility of the two policies:

{\footnotesize
\begin{equation}\label{eqn:example}
  \MEUfn {ex} =
    \max \bigg\{k_1, k_2 :\begin{gathered}
    \EU(ex|_{\pi_1})= k_1, \\
    \EU(ex|_{\pi_2})= k_2
  \end{gathered}\bigg\} =
    \max\bigg\{\begin{gathered}
    \overbrace{0.1 \times 10}^{\texttt{rainy}=\tt}  + \overbrace{0.9 \times (-5)}^{\texttt{rainy}=\ff} \\
    \underbrace{0.1 \times (-100)}_{\texttt{rainy}=\tt} + \underbrace{0.9 \times 0}_{\texttt{rainy}=\ff}
    \end{gathered}\bigg\} =
  \max \{-3.5,-10\} =-3.5.
\end{equation}
}

However, if we choose to uncomment Line 2 of~\cref{fig:motivation-dappl},
we add to our scenario that we \textit{observe} that it is raining today.
Thus, to compute MEU we must compute the \textit{conditional expected utility}
of each of our policies given that it is raining.
If we say $ex_{\texttt{obs}}$ as our motivating example with the \dapplcode{observe},
then we can now compute the MEU \textit{conditional on the fact that it is raining},
which yields a different answer than that of~\cref{eqn:example}:

{\footnotesize
\begin{equation}\label{eqn:example-with-observe}
  \MEUfn {ex_{\texttt{obs}}} =
    \max \bigg\{k_1, k_2 :\begin{gathered}
      \EU(ex_{\texttt{obs}}|_{\pi_1})= k_1, \\
      \EU(ex_{\texttt{obs}}|_{\pi_2})= k_2
  \end{gathered}\bigg\} =
    \max\bigg\{\begin{gathered}
    \overbrace{1 \times 10}^{\texttt{rainy}=\tt}  + \overbrace{0 \times (-5)}^{\texttt{rainy}=\ff} \\
    \underbrace{1 \times (-100)}_{\texttt{rainy}=\tt} + \underbrace{0 \times 0}_{\texttt{rainy}=\ff}
    \end{gathered}\bigg\} =
  \max \{10,-100\} =10.
\end{equation}
}

\subsubsection{Expected Utility of Boolean Formulae}\label{subsubsec:eu-via-compilation}
Now we begin working towards our new approach to scaling MEU for \dappl{}
programs. The core of our approach is to compile a \dappl{} program into a data
structure for which computing upper-bounds on the expected utility of partial
policies is \emph{efficient in the size of the compiled representation}.  Our
approach is a generalization to the recent approaches to performing
probabilistic inference via knowledge compilation, which is currently the
state-of-the-art approach for performing exact discrete probabilistic
inference~\citep{holtzen2020scaling,fierens2015inference}.  The idea with
inference via knowledge compilation is to reduce the problem of inference to
performing a weighted model count of a Boolean formula, for which there exist
specialized scalable solutions. Formally, a \emph{weighted Boolean formula} is a
pair $(\varphi, w)$ where $\varphi$ is a logical formula and $w$ is a function
that maps literals (assignments to variables in $\varphi$) to real-valued
weights.  A \emph{model} $m$ is a total  assignment to variables in $\varphi$
that satisfies the formula.  The weight of a model $m$ is the product of the
weights of each literal.  Then, the \emph{weighted model count}
$\texttt{WMC}(\varphi, w)$ is defined to be the sum of weights of each model of
$\varphi$, i.e. \texttt{WMC}$(\varphi, w) \triangleq \sum_{m \models \varphi} w(m)$.

\citet{holtzen2020scaling} showed how to reduce probabilistic inference for a
small language similar to \dappl{} (but without decisions or rewards) to weighted model counting. However,
our problem is MEU, not probabilistic inference; to connect these ideas,
we leverage a well-known generalization of WMC that allows one to instead
perform weighted model counts where the weights come from an
arbitrary \emph{semiring}~\citep{kimmig2017algebraic,kimmig2011algebraic}:

\begin{definition}[Semiring]\label{def:semiring}
  A semiring is a tuple $\mathcal R = (R, \oplus, \otimes, \mathbf{1}, \mathbf{0})$
  where $R$ is a set, $\oplus$ is a commutative monoid on $R$ with
  unit $\mathbf{0}$, $\otimes$ is a monoid on $R$ with unit $\mathbf{1}$,
  $\mathbf{0}$ annihilates $R$ under $\otimes$, and $\otimes$ distributes over $\oplus$.
\end{definition}

This invites a natural definition of an algebraic model count where
literals are permitted to be weighted by elements of a semiring instead of the real numbers,
similar to \emph{weighted programming}~\citep{batz2022weighted}:
\begin{definition}[Algebraic model counting~\citep{kimmig2011algebraic,kimmig2017algebraic}]
\label{def:amc}
  Let $\varphi$ be a propositional formula, $\vars(\varphi)$ be the variables in $\varphi$,
  and $\lits(\varphi)$ denote the set of literals for variables in $\varphi$.
  Let $w : \lits(\varphi) \rightarrow \mathcal R$ be a weight function that
  maps literals to a weight in semiring $\mathcal R$.
  Then, the \emph{weight of a model of $\varphi$} is the product of the weights
  of the literals in that model: i.e., for some model $m$ of $\varphi$,
  we define $w(m) = \bigotimes_{\ell \in m} w(l)$.
  Then, the \emph{algebraic model count} is the weighted sum of models of $\varphi$:
  \begin{align}\label{eq:amc}
    \AMC(\varphi, w) \triangleq \bigoplus_{m \models \varphi} w(m).
  \end{align}
\end{definition}

Now we illustrate how we reduce computing the MEU of the \dappl{}
program in~\Cref{fig:motivation-dappl} to performing an algebraic
model count of a particular formula.
We
construct formulae with two kinds of Boolean variables: \emph{probabilistic variables}
and \emph{reward variables} that indicate
whether or not the agent receives a reward.
In~\cref{fig:motivation-dappl}, we have a single probabilistic variable $r$ that is true if and only if it is \texttt{rainy},
and three reward variables $R_v$ that are
true if and only if the agent receives a reward of $v$. Then, we can give a
Boolean formula $\varphi_{u}$ and $\varphi_{\overline{u}}$ for the two policies
of bringing and not bringing an umbrella respectively:\footnote{We write the
negation of a variable using an overline.}
\begin{align}
  \varphi_u &= (r \land R_{10} \land \overline{R_{-5}} \land \overline{R_{-100}}) \lor (\overline{r} \land \overline{R_{10}} \land {R_{-5}} \land \overline{R_{-100}}) \label{eq:formula-umbrella} \\
  \varphi_{\overline{u}} &= (r \land \overline{R_{10}} \land \overline{R_{-5}} \land {R_{-100}}) \lor (\overline{r} \land \overline{R_{10}} \land \overline{R_{-5}} \land \overline{R_{-100}})
  \label{eq:formula-no-umbrella}
\end{align}

% \textit{How can one efficiently compute the expected utility from a Boolean formula?}
% To answer this question,
% let us recall that in knowledge compilation~\citep{darwiche2002knowledge,holtzen2020scaling,de2007problog},
% probabilistic inference over a formula is done by assigning a probability weight to each literal in the formula,
% now called a \textit{weighted Boolean formula}, at which point inference specializes to the task of
% \textit{weighted model counting}.
% As now there are rewards involved, we need to weigh literals by more than just probabilities:
% we need to weight them by a suitable semiring.

Continuing with our reduction,
we can now encode expected utility computations as an algebraic model count
over a particular kind of semiring, the expectation semiring:

\begin{definition}[Expectation semiring~\citep{eisner2002parameter}]
\label{def:expectation semiring}
  The expectation semiring
  $\mathcal S$
  is a semiring on a base set $S = \R^{\geq 0} \times \R$, where the first component
  is a probability and the second represents expected utility.
  Addition is defined component-wise $(p,u) \oplus (q,v) \triangleq (p+q, u+v)$,
  multiplication defined as $(p,u) \otimes (q,v) \triangleq (pq, pv+qu)$,
  the multiplicative unit is $\mathbf{1} \triangleq (1,0)$,
  and the additive unit is $\mathbf{0} \triangleq (0,0)$.
\end{definition}

To continue the reduction, we want to design an algebraic model count for
$\varphi_u$ that computes the expected utility of the policy for bringing an
umbrella. To do this, we give weights to each literal:
\begin{align*}
  w(r) = (0.1, 0) \quad& w(R_{10}) = (1, 10) \quad& w(R_{100}) = (1, 100) \quad& w(R_{-5}) = (1, -5)\\
  w(\overline{r}) = (0.9, 0) \quad& w(\overline{R_{10}}) = (1, 0) \quad& w(\overline{R_{100}}) = (1, 0) \quad& w(\overline{R_{-5}}) = (1, 0)
\end{align*}
Intuitively, since $r$ represents the outcome of \texttt{flip 0.1} being true, it has a
probability component of $0.1$ and a reward component of 0.
These weights are carefully designed so that the algebraic model count computes
the expected utility of the policy:
\begin{align}
  \AMC(\varphi_u, w)
  &= \Big(\underbrace{(0.1,0)  \otimes (1,10) \otimes (1, 0) \otimes (1, 0)}_{r,\ R_{10},\ \overline{R_{-5}},\ \overline{R_{-100}}} \Big)
  \oplus \Big(\underbrace{ (0.9,0) \otimes  (1, 0) \otimes(1,-5) \otimes (1, 0)}_{\overline{r},\ \overline{R_{10}},\ R_{-5},\ \overline{R_{-100}}}\Big) \nonumber \\
  & = (0.1, 1) \oplus (0.9, -4.5) = (1, -3.5).
  \label{eq:ex-amc}
\end{align}
% We capture in~\cref{appendix:amc invariant}
% precisely how the algebraic model count aligns with the traditional
% expectation-based definition of
% expected utility shown in~\cref{subsubsec:meu-example}.
% Repeating the same process for the policy \texttt{not\_umbrella} and comparing the expected utility as
% given by the AMC would yield the MEU. We can make this precise:

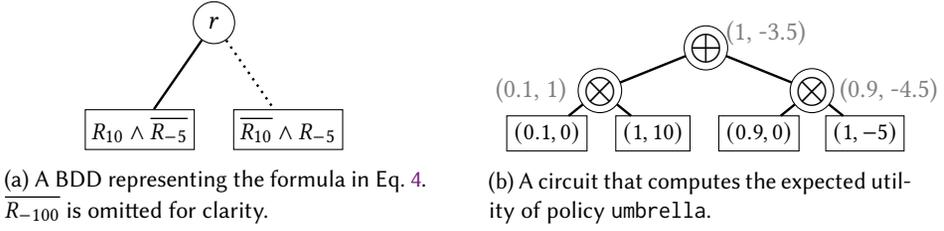
\begin{figure}

  % \hfill
  \begin{subfigure}{0.4\linewidth}
    \centering
    \scalebox{1}{
      \begin{tikzpicture}
        \def\lvl{20pt}
      \node (rainy) at (0, 0) [bddnode] {$r$};

      \node (r1) at ($(rainy) + (-28bp, -\lvl-20)$) [bddformula] {$R_{10} \land \overline{R_{-5}}$};

      \node (r21) at ($(rainy) + (28bp, -\lvl-20)$) [bddformula] {$\overline{R_{10}} \land R_{-5}$};

      \begin{scope}[on background layer]
        \draw [highedge] (rainy) -- (r1);
        \draw [lowedge] (rainy) -- (r21);
      \end{scope}
      \end{tikzpicture}
    }
    \caption{A BDD representing the formula in Eq.~\ref{eq:formula-umbrella}. $\overline{R_{-100}}$ is omitted for clarity.}
    \label{fig:bdd umbrella}
  \end{subfigure}
  \hspace{2em}
  \begin{subfigure}{0.4\linewidth}
    \centering
    \scalebox{1}{
    \begin{tikzpicture}[every node/.style={inner sep=0,outer sep=0}]
        \def\lvl{16pt}

      \node (root) at (0, 0) [bddnode] {$\bigoplus$} node[xshift=0.8cm, yshift=0.2cm] {{\color{gray} (1, -3.5)}};

      % \node (times1)

      \node[bddnode] (mul1) at ($(root) + (-40bp, -\lvl)$) {$\bigotimes$} ;
      \node at ($(mul1) + (-26bp, 0)$) {{\color{gray} (0.1, 1)}};
      \node[bddnode] (mul2) at ($(root) + (40bp, -\lvl)$) {$\bigotimes$};
      \node at ($(mul2) + (29bp, 0)$) {{\color{gray} (0.9, -4.5)}};

      \node[bddformula] (r1) at ($(mul1) + (-20bp, -\lvl)$) {$(0.1, 0)$};
      \node[bddformula] (r2) at ($(mul1) + (20bp, -\lvl)$) {$(1, 10)$};

      \node[bddformula] (r3) at ($(mul2) + (-20bp, -\lvl)$) {$(0.9, 0)$};
      \node[bddformula] (r4) at ($(mul2) + (20bp, -\lvl)$) {$(1, -5)$};

      \begin{scope}[on background layer]
        \draw [highedge] (root) --  (mul1);
        \draw [highedge] (root) -- (mul2);
        \draw [highedge] (mul1) -- (r1);
        \draw [highedge] (mul1) -- (r2);
        \draw [highedge] (mul2) -- (r3);
        \draw [highedge] (mul2) -- (r4);
      \end{scope}
      \end{tikzpicture}
    }
    \caption{A circuit that computes the expected utility of policy \texttt{umbrella}.}
    \label{fig:policy umbrella}
  \end{subfigure}
  \caption{Compiled Boolean circuit representations for $\varphi_u$.}
  \label{fig:motiv-a circuit 1}
\end{figure}
At this point in the reduction we are left with an arbitrary AMC, which in general is \#P-hard~\cite{kimmig2017algebraic}; it seems like we have not yet
made progress. This is where knowledge compilation comes into
play~\cite{darwiche2002knowledge,chavira2008probabilistic, sang2005performing}.
The key idea of knowledge compilation is to compile Boolean formulae into representations
that support particular queries: for instance, \dice{} compiles Boolean formulae into
binary decision diagrams (BDDs), which support linear-time weighted model counting,
in order to perform inference.
This compilation is expensive, but once performed, inference is efficient in the
size of the result; this amortization benefit will be crucial for our subsequent search strategy.
This process scales well because
BDDs naturally exploit repeated
sub-structure in the program such as conditional independence.
\citet{kimmig2017algebraic} showed that an analogous knowledge compilation strategy
can also be used to solve algebraic model counts.
This is visualized in Figure~\ref{fig:bdd umbrella}, which shows a
compiled representation of \cref{eq:formula-umbrella} (where we
have elided the negated reward variables for space).
Fig.~\ref{fig:policy umbrella} shows how to interpret the BDD in
Fig.~\ref{fig:bdd umbrella} as a circuit compactly representing $\AMC$.
The leaves of the circuit are elements of the
expectation semiring $\mathcal S$, and nodes are semiring operations $\oplus$
and $\otimes$, instead of the real-valued operations $+$ and $\times$.
% To construct this circuit from a BDD, we associate each internal variable
% node with a sum-over-products that multiplies the weight of the variable with
% the circuit representing its children. For instance, the left path from the root
% $\oplus$ node in Figure~\ref{fig:policy umbrella} corresponds to the BDD
% path where $r$ is true in Figure~\ref{fig:bdd umbrella}.
The algebraic model
computation is shown in gray, and only requires a linear-time
bottom-up pass of the graph, mirroring the weighted model count.
In Section~\ref{sec:eval}, we will show that we
can compile very large \dappl{} programs into surprisingly compact circuits
due to the opportunities for structure sharing.

\subsubsection{Optimization-via-Compilation}\label{subsubsec:optimization-via-compilation}
At this point, we know how to use algebraic model counting to compute the expected utility
of a particular policy, but we do not yet know how to efficiently \emph{search for an optimal policy}.
We now return to our task of finding the \emph{optimal policy}
for a $\dappl$ program, which is our key new novelty.
A na\"{i}ve approach can be to
associate a Boolean formula to every policy as in~\cref{subsubsec:eu-via-compilation},
compute the expected utilities via $\AMC$, then find the maximum over this collection.
However, this approach is clearly exponential in the number of decisions and wasteful:
it unnecessarily recompiles the same sub-program into a BDD
numerous times, even if it is shared across the different policies.
What we
desire is a \emph{single compilation pass} on which to do repeated
efficient evaluation of different policies for \dappl{} programs.

% \begin{definition}[MEU of a Boolean Formula]
% \label{def:meu boolean formula}
%   Let $\varphi$ be a Boolean formula that consists of reward variables
%   $\{R_i\}$, probabilistic variables $\{P_j\}$,
%   decision variables $\{D_k\}$, utility map $U$, and probability map $\Pr$.
%   Then, the maximum expected utility of $\varphi$ is:
%   \begin{align}
%     M\EU(\varphi) \triangleq \max_{\{d_k\}} \EU[\varphi \mid \{d_k\}]
%   \end{align}
%   where $\varphi \mid \{d_k\}$ denotes setting $D_k = d_k$ in $\varphi$.
%   Applying Lemma~\ref{thm:amc invariant}, we can give an alternative definition
%   \begin{align}
%     M\EU(\varphi) \triangleq \max_{\{d_k\}} \AMC(\varphi | \{d_k\},w)_{\EU}.
%   \end{align}
%   We call each $d_k$ a \emph{policy}, with the witness for $\MEU(\varphi)$ being the
%   \emph{optimal policy.}
%   \label{def:meu on wbf}
% \end{definition}

% In Section~\ref{subsubsec:meu}, we extend Definition~\ref{def:meu boolean formula}
% to handle evidence.
% This reduction enables us to write down a single Boolean formula that represents
% the entire state-space of our example program, including decisions and
% probabilistic outcomes. Figure~\ref{fig:motiv-a-bdd} shows a BDD representing
% this formula for the motivating example in Figure~\ref{fig:motiv-a}: this BDD contains
% a Boolean variable $u$ that is true if and only if an umbrella is brought.

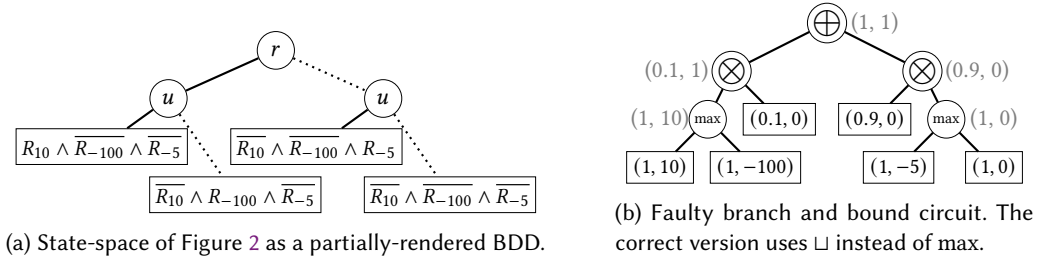
\begin{figure}
  \begin{subfigure}{0.55\linewidth}
    \centering
    \scalebox{0.9}{
    \begin{tikzpicture}
        \def\lvl{20pt}
      \node (rainy) at (0, 0) [bddnode] {$r$};

      \node (u2) at ($(rainy) + (45bp, -\lvl)$) [bddnode] {$u$};
      \node (u) at ($(rainy) + (-45bp, -\lvl)$) [bddnode] {$u$};

      \node (r1) at ($(u) + (-28bp, -\lvl)$) [bddformula] {$R_{10} \land \overline{R_{-100}}\land \overline{R_{-5}}$};
      \node (r2) at ($(u) + (28bp, -\lvl-20)$) [bddformula] {$\overline{R_{10}} \land R_{-100}\land \overline{R_{-5}}$};

      \node (r21) at ($(u2) + (-28bp, -\lvl)$) [bddformula] {$\overline{R_{10}} \land \overline{R_{-100}}\land R_{-5}$};

      \node (r22) at ($(u2) + (28bp, -\lvl-20)$) [bddformula] {$\overline{R_{10}} \land \overline{R_{-100}} \land \overline{R_{-5}}$};
      % \node (r22) at ($(u2) + (28bp, -\lvl)$) [bddterminal] {$\false$};

      \begin{scope}[on background layer]
        \draw [highedge] (rainy) -- (u);
        \draw [lowedge] (rainy) -- (u2);
        \draw [highedge] (u) -- (r1);
        \draw [lowedge] (u) -- (r2);
        \draw [highedge] (u2) -- (r21);
        \draw [lowedge] (u2) -- (r22);
      \end{scope}
      \end{tikzpicture}
    }
    \caption{State-space of~\cref{fig:motivation-dappl} as a partially-rendered BDD.}
    \label{fig:motiv-a-bdd}
    \end{subfigure}
\hfill
\begin{subfigure}{0.4\linewidth}
  \centering
  \scalebox{0.9}{
  \begin{tikzpicture}[every node/.style={inner sep=0,outer sep=0}]
      \def\lvl{20pt}

    \node (root) at (0, 0) [bddnode] {$\bigoplus$};
    \node at ($(root) + (20bp, 0)$) {{\color{gray} (1, 1)}};

    \node (mul1) at ($(root) + (-40bp, -\lvl)$) [bddnode] {$\bigotimes$};
    \node at ($(mul1) + (-23bp, 0)$) {{\color{gray} (0.1, 1)}};
    \node (mul2) at ($(root) + (40bp, -\lvl)$) [bddnode] {$\bigotimes$};
    \node at ($(mul2) + (23bp, 0)$) {{\color{gray} (0.9, 0)}};

    \node[bddformula] (term5) at ($(mul1) + (20bp, -\lvl)$) {$(0.1, 0)$};
    \node[bddformula] (term6) at ($(mul2) + (-20bp, -\lvl)$) {$(0.9, 0)$};

    \node (cup) at ($(mul1) + (-10bp, -\lvl)$) [bddnode] {\tiny{$\max$}};
    \node at ($(cup) + (-20bp, 0)$) {{\color{gray} (1, 10)}};
    \node (cup2) at ($(mul2) + (10bp, -\lvl)$) [bddnode] {\tiny{$\max$}};
    \node at ($(cup2) + (20bp, 0)$) {{\color{gray} (1, 0)}};

    \node[bddformula] (term1) at ($(cup) + (-20bp, -\lvl)$) {$(1, 10)$};
    \node[bddformula] (term2) at ($(cup) + (20bp, -\lvl)$) {$(1, -100)$};

    \node[bddformula] (term3) at ($(cup2) + (-20bp, -\lvl)$) {$(1, -5)$};
    \node[bddformula] (term4) at ($(cup2) + (20bp, -\lvl)$) {$(1, 0)$};

    \begin{scope}[on background layer]
      \draw [highedge] (root) -- (mul1) ;
      \draw [highedge] (root) -- (mul2);
      \draw [highedge] (mul1) -- (term5);
      \draw [highedge] (mul2) -- (term6);
      \draw [highedge] (mul1) -- (cup);
      \draw [highedge] (mul2) -- (cup2);
      \draw [highedge] (cup) -- (term1);
      \draw [highedge] (cup) -- (term2);
      \draw [highedge] (cup2) -- (term3);
      \draw [highedge] (cup2) -- (term4);
    \end{scope}
    \end{tikzpicture}
  }
  \caption{Faulty branch and bound circuit. The correct version uses $\sqcup$ instead of $\max$.}
  \label{fig:bb circuit example}
  \end{subfigure}
\caption{Branch-and-bound intermediate representation for the example program in~\cref{fig:motivation-dappl}.}
\label{fig:bbir-for-motivation-dappl}
\end{figure}

This leads to one of our main contributions: a new intermediate representation we call
the branch-and-bound intermediate representaion (BBIR). The example in
Figure~\ref{fig:motiv-a-bdd} already solves the problem of unnecessarily
repeatedly recompiling sub-programs: we can perform policy search directly on
the BDD by exhaustively enumerating all possible assignments to decision
variables and computing an expected utility using the method outlined in~\cref{subsubsec:eu-via-compilation}.
However, this enumeration strategy still suffers from
search-space explosion, and is exponential in the number of decision variables.
To avert this and scale to \dappl{} programs with a large number of choices, we
leverage the compiled BDD in Figure~\ref{fig:motiv-a-bdd} to efficiently compute
upper-bounds on the expected utility of \emph{partial policies},
defined formally in~\ref{def:partial policy}.
This lets us
design a branch-and-bound algorithm in Section~\ref{sec:bbir} to prune
policies during search.

Let us illustrate why a branch-and-bound algorithm is necessary and
a single bottom-up pass, such as the one in~\cref{subsubsec:eu-via-compilation},
is not sufficient.
Consider the circuit description of Figure~\ref{fig:motiv-a-bdd} that efficiently encodes a
solution to our decision scenario. A
straightforward approach to find MEU may be to associate every decision node in
the BDD with a $\max$ operation,
where $\max$ selects the higher utility node.
This circuit is visualized in Figure~\ref{fig:bb circuit example}.

However, there is a problem with the circuit in Figure~\ref{fig:bb circuit example}!
Recall the computations in~\cref{eqn:example}. The
maximum is the \emph{very last} operation
performed in the computation of MEU, performed over all decision variables.
In the bottom-up computation of the circuit in Figure~\ref{fig:bb circuit example},
the maximum is the \emph{very first} operation.
Thus this circuit will compute the wrong answer, as it is
\emph{not} generally the case that $\max_x \sum_y f(x,y) = \sum_y \max_x f(x,y)$
for an ordered semiring-valued function $f$, even in the real setting.
To solve this problem, we can force all
decision variables occur first in the top-down variable order
of the BDD, forcing maximums the final operations taken.
This is the approach
taken by the \emph{two-level algebraic model counting (2AMC)} approach of
\citet{derkinderen2020algebraic}.
As we will show in Section~\ref{sec:eval},
this order constraint can be catastrophic for performance, as
the size of a BDD is very sensitive to the variable order, and hence compiling
to order-constrained BDDs scales very poorly compared to compiling to BDDs where
the variables can be optimally ordered.

Our main contribution, in Section~\ref{sec:bbir}, gives a circuit representation
for upper-bounding the utility of a partially assigned policy
without constraining the variable order during BDD
compilation.
Our approach relies on the following intuition: for a \emph{real-valued} function $f$,
while it not generally the case that $\max_x \sum_y f(x,y) = \sum_y \max_x f(x,y)$, it
\emph{is the case} that $\max_x \sum_y f(x,y) \le \sum_y \max_x f(x,y)$:
commuting sums and maxes yields upper bounds for the real semiring $\R$.
This powerful \emph{commuting bound} holds for the reals,
and more broadly semirings with
a join-semilattice structure:
we verify this intuition via a lemma in~\cref{appendix:commute join}.

% A similar observation is leveraged by \citet{huang2006solving} to
% design a highly effective branch-and-bound strategy for solving the marginal-map problem
% for Bayesian networks; see Section~\ref{subsubsec:mmap} for a discussion.
% Suppose we have two partial-policies $\pi_1$ and $\pi_2$

\begin{definition}[Lattice semiring]\label{def:latticed semiring}
  A \textbf{lattice semiring} is a semiring $\mathcal S = (S, \oplus, \otimes,
  \mathbf{0}, \mathbf{1})$ equipped with a partial order $\sqsubseteq$ on
  $S$ respecting $\oplus$ -- i.e., if $a \plt b$ and $c \plt d$, then
  $a \oplus c \plt b \oplus d$ -- that admits both meets (greatest lower bounds,
  denoted $\sqcap$) and joins (least upper bounds, denoted $\sqcup$).
\end{definition}

If $f$ is a lattice-semiring-valued function, $\bigsqcup_x \bigoplus_y
 f(x,y) \le  \bigoplus_y \bigsqcup_x f(x,y)$.  For the expectation
semiring, we define the partial order pointwise: $(a, b) \plt (c, d)$ if and
only if $a \le c$ and $b \le d$. This implies that $(a, b) \sqcup (c, d) =
(\max(a, c), \max(b, d))$, and similarly so for meets.

Returning to our goal of using BDDs to efficiently compute upper-bounds on
utilities, we can interpret decision variables as joins: this computation is visualized in
Figure~\ref{fig:bb circuit example}. The computed upper-bound is visualized in gray;
the final computed upper-bound $(1,1)$ is indeed an upper-bound
(with respect to $\sqsubseteq$)
on the expected
utility of the optimal policy, which we expect to be $(1, -3.5)$.
Ultimately, this insight allows us to give a branch-and-bound procedure
to solve both a general class of optimization problems over probabilistic
inference. Next, we will show another instantiation of this framework
for solving maximum marginal a-posteriori (MMAP) problems.

\subsection{First-Class Marginal Maximum A Posteriori}\label{subsec:pineappl-overview}
% \cref{subsec:dappl-overview} demonstrated how the BBIR--and more generally, the extension
% of reasoning-via-compilation techniques to semirings--allows us to
% express and solve optimization problems over inference.
% However, $\dappl$ uses BBIR to solve MEU as the \textit{result} of a program.
% What if we want to use the optimal values in intermediate computation?
% How is this compatible with the BBIR?
% To answer this question, we introduce $\pineappl$
% and the marginal maximum a posteriori (MMAP) problem (\cref{subsubsec:mmap-pineappl})
% and how the BBIR is used to solve intermediate optimization queries
% (\cref{subsubsec:staging}), completing our motivation and overview for BBIR.

% \subsubsection{Using MMAP as a first-class reasoning primitive}
% \label{subsubsec:mmap-pineappl}

One of the key benefits of our algebraic approach to solving MEU in the previous
section is that it can be generalized to different semirings, and therefore applied
to a diverse set of reasoning problems.
Another powerful and common form of optimization-over-inference that is useful in
probabilistic reasoning is the \emph{marginal maximum a-posteriori problem} (MMAP),\footnote{The formal definition of MMAP is given
in~\cref{subsubsec:mmap}.}
which has historically had broad applications in diagnosis.
Consider the $\pineappl$ program
in \cref{fig:motiv-pineappl} presenting the following scenario:

\begin{quote}
\textit{``You are a doctor attempting to diagnose a patient. There is a 50\% chance that any given person has the disease.  If someone has the
disease, there is a 70\% percent chance that they have a headache. If
they do not have disease, there is still a 10\% chance that they have
a headache. You make the most likely diagnosis based on observing the
patient has a headache. There are consequences for misdiagnosis, either
diagnosing the patient when they do not have the disease or failing to
diagnose the patient when they do. What is the probability of complications
arising in a patient observing a headache?''}
\end{quote}

\begin{wrapfigure}{r}{0.44\linewidth}
\begin{pineapplcodeblock}[basicstyle=\tiny\ttfamily]
disease = flip 0.5;
if disease { headache=flip 0.7; }
  else { headache=flip 0.1; }
diagnosis = mmap(disease) with { headache }
if diagnosis && disease { complications=ff; }
  else if diagnosis && !disease { complications=flip 0.4; }
  else if !diagnosis && disease { complications=flip 0.9; }
  else { complications=ff; }
pr(complications)
\end{pineapplcodeblock}
\caption{Example $\pineappl$ program.}
\label{fig:motiv-pineappl}
\end{wrapfigure}

The key new element in this scenario is \emph{first-class optimization}: within
this example, a doctor wants to know the most likely symptom given a disease,
and then take some further action based on the outcome of that query.
\Cref{fig:motiv-pineappl} shows how this is encoded as a program.
On Line 1, we define our prior on whether a member of the population will have
the disease. Lines 2--3 model the conditional probability of a member of the
population having a headache based on whether they have the disease.
Then, on Line 4,
we bind \pineapplcode{diagnosis} to the \textit{most likely state
of} \pineapplcode{disease}, given the observation that \pineapplcode{headache}
is true. Lines 5--8 model the conditional probability of complications based on
the state of \pineapplcode{disease} and \pineapplcode{diagnosis}. Finally, on
Line 9 we calculate the probability of complications given the previous model.

The goal of a $\pineappl$ program is to perform probabilistic inference, much
like standard PPLs, but with the added complexity that random variables can
depend on the most likely state of previously defined variables.  For example,
the most likely state, or MMAP,  of \pineapplcode{disease} when
observing \pineapplcode{headache} is $\tt$. We can derive this by
computing the probability of \pineapplcode{disease} conditioned
on the observation of \pineapplcode{headache}:

{\footnotesize
  \begin{align}
    \Pr[\texttt{diagnosis} = \tt]
    &= \Pr[\texttt{disease} = \tt \mid \texttt{headache} = \tt] \nonumber
    = \frac{\Pr[\texttt{headache} = \tt \mid \texttt{disease} = \tt] \times \Pr[\texttt{disease} = \tt]}{\Pr[\texttt{headache} = \tt]} \nonumber \\
    &= \frac{0.7 \times 0.5}{(0.7 \times 0.5) + (0.3 \times 0.1)}
                                  = 0.92
                                    &\label{eq:mmap-motiv-pineappl}
  \end{align}
}
So, when computing \pineapplcode{Pr(complications)}, we need only consider
where \pineapplcode{diagnosis} is $\tt$:

{\footnotesize
\begin{equation}
  \Pr[\texttt{complications}] = \Pr [\texttt{disease} = \tt] \times 0 +
  \Pr[\texttt{disease} = \ff] \times 0.4 = 0 + 0.5 \times 0.4 = 0.2
\end{equation}}

% To give an answer to the program in \cref{fig:motiv-pineappl} we must first
% define the MMAP for some set of variables in a boolean formula and the
% probability of a satisfying a boolean formula, which we cast as a problem of
% weighted model counting (WMC), similar to Dice~\citep{holtzen2020scaling}.

% \begin{definition}[MMAP of a Boolean Formula]
%   Let $\varphi$ be a boolean formula consisting of variables $V$, $M \subseteq
%   V$, and $N = V \setminus M$. Assuming the existence of a map $\Pr: inst(V) \to
%   \mathbb{R}$, we define the $\mathrm{MMAP}$ of $\varphi$ with respect to $M$ as
%   follows:
%   \[
%     MMAP(\varphi, M) = \argmax_{m \in inst(M)} \left( \sum_{n \in N}
%     \Pr[{m \cup n \mid m \cup n \models \varphi}]\right)
%   \]
% \end{definition}

% \begin{definition}[Weighted Model Counting]
%   Let $\varphi$ be a boolean formula consisting of variables $V$ and $M$ be the
%   models of $\varphi$. Assuming the existence of a map $\Pr: inst(V) \to
%   \mathbb{R}$, we define the weighted model count a boolean formula as
%   follows:
%   \[
%     WMC(\varphi) = \sum_{m \in M} \Pr(m)
%   \]
% \end{definition}

% We give formal definitions for these queries with respect to both an operational
% semantics for $\pineappl$ in \cref{sec:pineappl-sem} and a boolean compilation
% strategy in \cref{sec:pineappl-compl}.

\subsubsection{Staging BBIR Compilation for Meta-Optimization}
\label{subsubsec:staging}

In this section, we demonstrate that the BBIR's unrestricted variable order,
as addressed in~\cref{subsubsec:optimization-via-compilation},
paves the way for a \textit{compositional, staged approach} to
efficiently compiling programs with meta-optimization such as MMAP.

Attempting to emulate the methodology in~\cref{subsubsec:eu-via-compilation}
quickly leads to blowup.
We would have to create two Boolean formulae and compare their $\AMC$ over $\R$:
one for when \texttt{diagnosis} is $\tt$ and one for when \texttt{diagnosis} is $\ff$.
Then, the two formulae will have a duplicated subformula--the subformula
that declares the variables \texttt{disease} and \texttt{headache}.
As the number of MMAP queries increase, this quickly becomes intractable --
the number of times needed to recompile subformulae grows exponentially
with respect to the number of input variables.

To combat this blowup, we apply the idea of \textit{staged compilation}.
In traditional staging, the idea is to accelerate expensive and/or
repeated computation
by precompiling it into an optimized representation (see \citet{taha1999multistage}).
Such computations must be identified in the code and compiled first,
reaping performance benefits by avoiding repeated compilation.

In the case of MMAP, the expensive computation is determined
entirely by the input variables to a \pineapplcode{mmap} query.
By compiling the subformula representing the input variables into BBIR
first, we can then use the
branch-and-bound over BBIR to find the most likely state
(in this case, \texttt{diagnosis} being $\tt$)
and continue compilation of the program with that assignment in mind.

Let us see this idea in action.
Drawing another analogy to staging~\citep{devito2013terra},
we can pre-compile the first three lines of our program:

\begin{equation}\label{eq:partial-pineappl}
  \underbrace{\texttt{disease} \leftrightarrow f_{0.5}}_{\text{Line 1}}
  \land
  \underbrace{\texttt{headache} \leftrightarrow (\texttt{disease} \land f_{0.7} \lor \neg\texttt{disease} \land f_{0.1})}_{\text{Lines 2-3}}.
\end{equation}

We insert auxiliary variables \texttt{disease} and \texttt{headache} to maintain
sequentiality of the program: if we were to simply say
$f_{0.5} \land \left(f_{0.5} \land f_{0.7} \lor \neg f_{0.5} \land f_{0.1}\right)$,
then the program will return $\bot$ once any of the sampled values returned $\ff$,
which is incorrect.

BBIR allows \cref{eq:partial-pineappl} to be compiled as a branch-and-bound
circuit that can be efficiently queried, in a manner
similar to $\dappl$. Thus, we can deduce that the most likely assignment to
\texttt{diagnosis} is $\tt$, and then extend~\cref{eq:partial-pineappl} as such:
\begin{equation}
  \texttt{disease} \leftrightarrow f_{0.5} \land
  \texttt{headache} \leftrightarrow (\texttt{disease} \land f_{0.7} \lor \neg\texttt{disease} \land f_{0.1}) \land
  \texttt{diagnosis} \leftrightarrow T
\end{equation}
\noindent
at which point the compilation of the program can resume
by reusing the precompiled BBIR for~\cref{eq:partial-pineappl} and
\textit{without having computed a separate formula for when \texttt{diagnosis} is $\ff$.}

This is achieved without significant blowup because the variable order within the BDD
has no restrictions.
If the variable order were to be restricted
as per the approach of~\citet{derkinderen2020algebraic},
then at every call to \pineapplcode{mmap}
we would need to sift the queried variables to the top,
which is known to be expensive and can blow up
the size of the BDD~\citep{holtzen2020scaling,holtzen2021model}.

To summarize, we have demonstrated the key insights that make BBIR an ideal
compilation target for PPLs performing optimization:
\begin{enumerate}
  \item BBIR generalizes knowledge compilation beyond the real numbers,
  allowing for more general
  optimization problems over inference such as MEU to be expressed.
  \item BBIR does not enforce any variable order, which allows us to
  express probabilistic programs with meta-optimization queries
  through staged compilation of the BBIR.
\end{enumerate}
The next section will delve into technical details of how we achieve both objectives.
%\input{sections/pineappl_bbir.tex}

%%% Local Variables:
%%% mode: LaTeX
%%% TeX-master: "../oopsla-appendix.tex"
%%% End:

%%% Local Variables:
%%% mode: latex
%%% TeX-master: "../oopsla-appendix"
%%% End:

\section{Optimization-via-Compilation}\label{sec:bbir}

In this section, we give a formal account of the intuitions
reflected in Section~\ref{subsubsec:optimization-via-compilation}
and~\cref{subsubsec:staging}.
We will describe the
branch-and-bound semiring (Section~\ref{subsec:bb semiring}),
a class of lattice semirings (recall Definition~\ref{def:latticed semiring})
equipped with an additional total order that is \textit{compatible} with the existing lattice.
Afterwards, we introduce the BBIR
how it represents MEU and MMAP
(Section~\ref{subsec:msp jsp}),
and how it admits a polynomial time upper- and lower-bound
algorithm (Section~\ref{subsec:bounds}),
lending itself well to a branch-and-bound
approach (Section~\ref{subsec:meu with evidence}).

\subsection{The Branch-and-Bound Semiring}
\label{subsec:bb semiring}

In Section~\ref{subsubsec:optimization-via-compilation},
we introduced the definition of a lattice semiring
(Definition~\ref{def:latticed semiring}) and how it generalizes the
interchange law between max and sum
($\max_{x} \sum_y f(x,y) \leq \sum_y \max_x f(x,y)$) in the reals.
However, in a lattice semiring,
$\sqsubseteq$ is a partial order, so in general elements may not be able to be compared:
for example, in the expectation semiring, we cannot compare the values $(0.5, 1)$ and $(1,0)$
as $0.5 \leq 1$ but $1 \not\leq 0$.
However, if we were to compare the values $(0.5, 1)$ and $(1,0)$ as
values of
$\AMC{}$ corresponding to total (as opposed to partial) policies,
then the comparison is
obvious: we select $(0.5, 1)$ as it has the higher utility.
To reflect this intuition, we enrich lattice semirings with a total order,
which gives the definition of a branch-and-bound semiring:

\begin{definition}[Branch-and-Bound Semiring]\label{def:branch-and-bound semiring}
  A branch-and-bound semiring is a lattice semiring
  $(\mathcal R, \oplus, \otimes, \mathbf 0, \mathbf 1, \sqsubseteq)$
  equipped with an additional total order
  $\leq$ such that for all $a,b \in \mathcal R$,
  $a \sqsubseteq b$ implies $a \leq b$,
  which we henceforth call \emph{compatibility}.
\end{definition}

The real semiring $\R$ is a branch-and-bound semiring
in which the two orders are identical: the usual total order on the reals.
However, the intuition above is reflected most prominently in the expectation semiring:

\begin{proposition}\label{prop:S is bb semiring}
  The expectation semiring $\mathcal S$,
  as seen in Definition~\ref{def:expectation semiring},
  forms a branch-and-bound semiring with:
  \begin{enumerate*}
    \item $(p,u) \sqsubseteq (q,v)$ iff $p \leq q$ and $u \leq v$,
    with join $\bigsqcup$ being a coordinatewise max and
    meet $\bigsqcap$ being a coordinatewise min, and
    \item $(p,u) \leq (q,v)$ iff $u < v$ or $u=v$ and $p \leq q$.
  \end{enumerate*}
\end{proposition}

\begin{proof}
  Let $(p,u), (q,v) \in \mathcal S$ such that
  $(p,u) \sqsubseteq (q,v)$. Then $u \leq v$;. if $u <v$ we are done.
  If $u=v$ then $p \leq q$ and we are done.
\end{proof}

The distinction between $\sqsubseteq$ and $\leq$ is required when
comparing partial and total policies in~\cref{subsec:bounds}.
Compatibility will be required when we know, for $p \sqsubseteq q$,
that $p$ and $q$ are associated with total policies as opposed to partial.

% Continuing the intuition outlined
% for Definition~\ref{def:branch-and-bound semiring},
% recall that the AMC over $\mathcal S$ computes the
% expected utility, particularly that the expected
% utility of a policy.
% If we know two policies, then we can simply evaluate
% the expected utility of both and pick the greater; this
% is reflected in $\leq$
% in Proposition~\ref{prop:S is bb semiring}.
% On the other hand, $\sqsubseteq$ allows us to compare
% \textit{partial policies}--that is, a deterministic assignment
% to only \textit{some} components of the action space.
% How we reduce a partially evaluated expected utility calculation
% of a partial policy into a value in $\mathcal S$ is done through
% the join; we will see this in Section~\ref{subsec:bounds}.
% \sh{After reading
% this again, I do feel like either we should move a bunch of this to section 2,
% or try to move some of section 2 here; I'm not sure which is best, and
% we may need to wait until after the deadline to make these adjustments unless you
% see a quick way to do it. }

\subsection{The Branch-and-Bound IR}\label{subsec:msp jsp}
Now that we have defined the branch-and-bound semiring, we are ready
to reconstruct the branch-and-bound circuits in the
motivating examples in Section~\ref{sec:overview}.
% We had a motivating example (Figure~\ref{fig:bb-overview})
% that was then compiled to a binary decision diagram
% (Figure~\ref{fig:motiv-a-bdd}).
What additional information should the BDD in Figure~\ref{fig:motiv-a-bdd} have to fully represent
a decision scenario?
Of course we should specify which variables
to optimize over and which to not, and weights for all variables
present. But additionally we need to incorporate potential \textit{evidence}
showing the events to condition on as we evaluate the program.
We represent exactly this set of information in the BBIR.

\begin{definition}[Branch-and-bound IR]\label{def:bbir}
  A branch-and-bound intermediate representation (BBIR)
  over a branch-and-bound semiring $\mathcal B$
  is a
  tuple $\BBIR$ in which:
  \begin{itemize}
    \item $\{\varphi_i\}$ are propositional formulae
    in the factorized representation of a
    multi-rooted BDD~\citep{darwiche2002knowledge,clarke2018handbook},
    \item $X \subseteq \bigcup_i vars(\varphi_i)$ a selection of variables
    on which to branch over,
    \item $w : \bigcup_i lits(\varphi_i) \to \mathcal B$ a weight function.
  \end{itemize}
\end{definition}

% Over a BBIR we can define optimization problems over Boolean formulae.

% \begin{definition}[Max-Sum Problem]\label{def:msp over bbir}
%   Let $\BBIR$ be a BBIR over a branch-and-bound semiring
%   $\mathcal B$ and $\psi \in \{\varphi\}$.
%   Let $inst(X)$ be the set of all assignments of variables in $X$.
%   Then the max-sum problem (MSP) of $\psi$ over the BBIR is
%   \begin{equation}
%     MSP(\psi) = \max_{\pi \in inst(X)} \bigoplus_{m \models \psi|_{\pi}} f(m, \pi).
%   \end{equation}
% \end{definition}

% We demonstrate how this optimization problem aptly generalizes
% the MEU problem, along with several others
% present in the literature to showcase its generality.

We demonstrate below the definition of MEU and MMAP over BBIR below.

\subsubsection{The MEU Problem with Evidence}\label{subsubsec:meu}

Here, we give a formulation of the
MEU problem with evidence, a generalization of the MEU problem
addressed in~\cref{subsec:dappl-overview}
which allows us to eventually handle
\dapplcode{observe} statements in \dappl{}.
In particular we introduce an additional $\AMC$
in the denominator of the optimization function.
This additional model count can be handled by efficient
computation of bounds;
see Section~\ref{subsec:meu with evidence}
for full detail.

We represent this problem as a BBIR
$(\{\varphi \land \gamma_{\pi} : \pi \in \mathcal A\}, A, w)$, in which:

\begin{enumerate}[leftmargin=*]
  \item $\varphi$ is the Boolean formula detailing the control and data flow
  of the decision making model,
  \item $\gamma_{\pi}$ represent witnessed evidence for
  each policy $\pi \in \mathcal A$, where $\mathcal A$ is the collection of
  all possible policies (i.e., complete instantiations of choices)
  \item $A$ is the collection of variables representing choices.
  % on which we define instantiations (assignments) of $A$ as $inst(A) = \mathcal A$, and
  \item $w$ a weight function to denote rewards.
\end{enumerate}
On which the MEU problem reduces to the following optimization problem:
{\footnotesize\begin{equation}\label{eq:bbir-meu}
  \MEUfn{(\{\varphi \land \gamma_{\pi} : \pi \in \mathcal A\}, A, w)}
    \triangleq \max_{\pi \in \mathcal A} \frac{\AMC(\varphi|_{\pi} \land \gamma_{\pi},w)_{\EU}}{\AMC(\gamma_{\pi}, w)_{\Pr}},
\end{equation}}
where division is the normal division in $\R$  with the additional property that
division by $0$ is defined as $-\infty$.
The subscript $\EU$ and $\Pr$ denote the first and second projections over the expectation semiring,
referring to the $\AMC$ invariant proven in~\cref{appendix:amc invariant}.

To give a concrete example of this optimization problem, consider the example of~\cref{fig:motivation-dappl},
with the \dapplcode{observe} statement uncommented.
We can define
\begin{equation}
  \varphi = (u \land \varphi_u) \lor (\overline u \land \varphi_{\overline u}),
  \qquad
  \gamma_u = \gamma_{\overline u} = r,
  \qquad
  A = \{u\},
\end{equation}
where $\varphi_u$ and $\varphi_{\overline u}$ are
defined in~\cref{eq:formula-umbrella,eq:formula-no-umbrella} and $w$ are the weights
as defined in~\cref{sub@fig:bb circuit example}. Then we observe that
{\footnotesize
\begin{equation}
  \MEUfn{\{\varphi \land \gamma_i \mid i \in \{u, \overline u\}\}, A, w}
  = \max \left\{
  \frac{\AMC(\varphi_u \land r)_{\EU}}{\AMC(r)_{\Pr}},
  \frac{\AMC(\varphi_{\overline u} \land r)_{\EU}}{\AMC(r)_{\Pr}}
  \right\}
  = 10,
\end{equation}
}
validating the computations in~\cref{eqn:example-with-observe}.
\subsubsection{The Marginal Maximum A Posteriori (MMAP) Problem}\label{subsubsec:mmap}

We conclude with a formulation of the MMAP problem in full generality over a BBIR.
$\pineappl$ supports a limited form of conditioning, where observations can only occur
with a call to MMAP or a query (see~\cref{subsec:pineappl-sem} for details),
but we present a formulation of the MMAP problem which supports global conditioning.
We do so by defining the
BBIR $(\{\varphi, \gamma\}, M, w)$ where:

\begin{enumerate}[leftmargin=*]
  \item $M$ are our \emph{MAP variables} to compute the most likely state of,
  a subset of the variables of $\varphi$,
  \item $\varphi$ is our probabilistic model and $\gamma$ is our evidence to condition on,
  with $vars(\varphi) = M \cup V \cup E$ disjoint sets of variables where $E$
  is some set of variables representing priors and $V$ are probabilistic variables, and
  \item $w$ is a weight function with codomain in the real branch-and-bound semiring
  $\R$ where $\sqsubseteq, \leq$ are the usual total order.
\end{enumerate}

Then we can solve the following optimization problem for some priors $e \in inst(E)$,
where $inst$ denotes the set of all instantiations to a set of variables and $\varphi|_m$
denotes the formula derived by applying the literals of $m$ to $\varphi$:
{\footnotesize
\begin{align}\label{eq:bbir-mmap}
  \mathrm{MMAP}{(\{\varphi, \gamma\}, M, w, e)}
  &= \argmax_{m \in inst(M)} \sum_{\substack{v \in inst(V), \\ m \cup v \cup e \models \varphi}}
  \Pr[m \cup v \cup e \mid \gamma|_e]
  = \argmax_{m \in inst(M)} \frac{\AMC_\R(\varphi|_{m,e} \land \gamma|_e)}{\AMC(\gamma|_e)}.
\end{align}
}

When there are no priors, we elide $e$ in the arguments.
To give a concrete example of this problem, consider the example
given in the first four lines of~\cref{fig:motiv-pineappl}.
We can define:

{\footnotesize
\begin{equation*}
  \varphi =  \texttt{disease} \leftrightarrow f_{0.5} \land
  \texttt{headache} \leftrightarrow (\texttt{disease} \land f_{0.7} \lor \overline{\texttt{disease}} \land f_{0.1}).
  \qquad
  \gamma = \texttt{headache},
  \qquad
  M = \texttt{disease},
\end{equation*}
}
where $w(f_{n}) = 1 - w(\overline{f_{n}})$,
and the weight is 1 for all other literals.
\noindent Then we observe that with $V = \{f_{0.5}, f_{0.7}, f_{0.1}\}$,

{\footnotesize
\begin{align*}
  \mathrm{MMAP}{(\{\varphi, \gamma\}, \{\texttt{disease}\},w)}
  &=\max \left\{
    \sum_{\substack{v \in inst(V), \\
    v \cup \texttt{disease} \models \varphi}} \Pr\left[\texttt{disease} \cup v \mid \gamma\right],
    \sum_{\substack{v \in inst(V), \\
    v \cup \overline{\texttt{disease}} \models \varphi}} \Pr\left[\overline{\texttt{disease}} \cup v \mid \gamma\right],
  \right\} \\
  &=\max\{0.92, 0.08\} = 0.92,
\end{align*}
}
validating the computations in~\cref{eq:mmap-motiv-pineappl}.

Prior work, such as that of~\citet{huang2006solving} and~\citet{lee2016exact},
have leveraged
techniques in knowledge compilation to solve the MMAP problem via
a branch-and-bound algorithm.
Our method, to the best of our knowledge,
is the first method to generalize this approach beyond MMAP.

\subsection{Efficiently Upper-Bounding Algebraic Model Counts on BBIR}
\label{subsec:bounds}

We have demonstrated how the BBIR can represent important optimization problems
over probabilistic inference, as promised in~\cref{fig:bb-overview}.
However, a new problem representation
is moot without
gains in efficiency. Where does that happen?

Recall from Definition~\ref{def:bbir}
that the BBIR is over a branch-and-bound semiring, on which the
partial order $\sqsubseteq$ allowed the comparison of partially computed
algebraic model counts.
This is where the BBIR comes into play: it allows us to give an
upper- or lower-bound of partially computed algebraic model counts
on \textit{any} formula defined within the BBIR.
This is efficient--in particular, polynomial in the size of BBIR, more specifically
the BDD within. Thus, we can fully take advantage of the factorization of the
BDD while maintaining a way to compare partially computed values of AMC:

% To make this precise, we first need a definition.

\begin{definition}[Partial policies and completions]\label{def:partial policy}
  Let $\BBIR$ be a BBIR.
  Then, we can define $P$ a \emph{partial policy} of
  $X$ as instantiation of a subset of variables in $X$.
  A \textit{completion} $T$ of $P$ is an instantiation of variables of $X$
  such that $P \subseteq T$.
\end{definition}

\begin{figure}
  \begin{subfigure}[t]{0.37\linewidth}
    \begin{mdframed}{\footnotesize\begin{algorithmic}[1]
      \Procedure{$ub$}{$\BBIR, \varphi, P$}
      \State $pm \gets \bigotimes_{p \in P} w(p)$
      \State $acc \gets h(\varphi|_P, X,w)$
      \State \textbf{return} $pm \otimes acc$
      \EndProcedure
    \end{algorithmic}}
    \end{mdframed}
    \caption{The upper bound algorithm $ub$ takes in a BBIR, $\varphi \in \{\varphi\}$,
    and a $P$ a partial policy of $X$
    to find an upper bound of $\AMC(\varphi|_{T},w)$ for any completion $T$ of $P$.}
    \label{algorithm:ub}
  \end{subfigure}
  \hfill
  \begin{subfigure}[t]{0.6\textwidth}
    \begin{mdframed}{\footnotesize\begin{algorithmic}[1]
      \Procedure{$h$}{$\varphi, X, w$}
      \If {$\varphi = \top$} \textbf{return 1}
      \ElsIf {$\varphi = \bot$} \textbf{return 0}
      \Else{ \textbf{let }$v \gets \mathrm{root}(\varphi)$}
        \If{$v \in X$}  \textbf{return}
        $w(v) \otimes h(\varphi|_v) \sqcup w(\overline v) \otimes h(\varphi|_{\overline v}$)
        \Else{ \textbf{return}
        $w(v)\otimes h(\varphi|_v) \oplus w(\overline v)\otimes h(\varphi|_{\overline v}$)}
        \EndIf
      \EndIf
      \EndProcedure
    \end{algorithmic}}
    \end{mdframed}
    \caption{The helper function $h$ as seen on Line 3 in Fig.~\ref{algorithm:ub}. }
    \label{algorithm:h}
  \end{subfigure}
  \caption{A single top-down pass upper-bound function. The function $\mathrm{root}$
  returns the topmost variable in the BDD. In order to scale efficiently, these
  procedures must be memoized; we omit these details.}
  \label{fig:h}
\end{figure}

With this definition in mind, we can give the pseudocode for our
upper bound algorithm in \Cref{fig:h}.
Algorithm~\ref{algorithm:h} runs in polynomial-time
in the size of the BBIR, as it is known
conditioning takes polynomial time in a
binary decision diagram~\citep{darwiche2002knowledge}.
However, it is not clear what \Cref{algorithm:ub} is
upper-bounding.
The key is observing that, at any choice variable, taking the join $\sqcup$
greedily chooses the best possible value,
without caring about whether it is associated to a policy or not.
This allows us to upper-bound all completions $T$ of $P$,
as we demonstrate in the following theorem, proven in~\cref{appendix:proof ub correctness}.

\begin{theorem}\label{thm:ub correctness}
  Let $\BBIR$ be a BBIR and let $\varphi \in \{\varphi_i\}$. Let $P$ be a partial policy of $X$.
  Then for all completions $T$ of $P$ we have
  \begin{equation}\label{eq:ub correctness}
    ub(\BBIR, \varphi, P)
      \sqsupseteq \bigoplus_{m \models \varphi|_T} \bigotimes_{\ell \in m \cup T} w(\ell)
      = \AMC(\varphi|_T) \bigotimes_{\ell \in T} w(\ell).
  \end{equation}
\end{theorem}

Importantly, we can define a dual \textit{lower bound} algorithm $lb$
by taking Algorithm~\ref{algorithm:h}
and changing the join $\sqcup$ in line 5 to a meet $\sqcap$.
This proves vital when achieving full generality of the branch-and-bound,
as a simultaneous
lower and upper bound is required to maintain sound pruning in the presence of evidence.
We also state an important Lemma
that holds for both upper-and lower-bounds,
whose proof amounts to observing that for
total policies, no join is ever done when bounding,
leading to an exact value.

\begin{lemma}\label{lemma:ub on total policy is amc}
  For any BBIR $\BBIR$ and $\varphi \in \{\varphi\}$, for any total policy $T$
  of $X$, we have
  \begin{equation}
    ub(\BBIR, \varphi, T) =lb(\BBIR, \varphi, T)= \AMC(\varphi|_T, w).
  \end{equation}
\end{lemma}

\subsection{Upper Bounds in Action: a General Branch-and-Bound Algorithm}
\label{subsec:meu with evidence}

We have, so far, demonstrated some of the theory and intuition that leads into
the BBIR, and the efficient upper- and lower-bound operation it supports.
Now,
we can use it to our full advantage to implement a general branch-and-bound style
algorithm to solve optimization problems expressed over BBIR.
This subsumes a previous algorithm for MMAP by~\citet{huang2006solving}
and generalizes it to MEU and to any other branch and bound semiring.

The algorithm is given in Algorithm~\ref{algorithm:bb}. It finds the maximum of an
objective function $f$
(for example, the problems of~\cref{eq:bbir-meu,eq:bbir-mmap})
given an upper-bound $\mathsf{UB}_f$ for $f$ over partial policies,
which we describe for MEU and MMAP in~\cref{appendix:ub_f}.
$\mathsf{UB}_f$ for MEU and MMAP take full advantage of
Algorithm~\ref{algorithm:ub}, and are completed in polynomial time.

% For notational simplicity, instead of using the BBIR
% $(\{\varphi \land \gamma_{\pi} : \pi \in \mathcal A\}, A, w)$,
% we will use
% the tuple $(\{\varphi, \gamma\}, A, w)$
% in which $\gamma$ is the formula in which
% for all $\pi \in \mathcal A$, $\gamma|_{\pi} = \gamma_{\pi}$.
% For a given BBIR for MEU, and a policy $\pi \in \mathcal A$,
% we write $
%   \mathrm{EU}(\pi) =
%   \frac{\AMC((\varphi \land \gamma)|_{\pi}, w)_{\EU}}{\AMC(\gamma|_{\pi})_{\Pr}}
% $, with $w$ implicit.
% Note that this can be solved via two calls to $ub$,
% on $(\varphi \land \gamma)|_{\pi}$
% and $\gamma|_{\pi}$ respectively;
% this is an application of Lemma~\ref{lemma:ub on total policy is amc}.
% More generally, for any value $(a,b) \in \mathcal S$ and $r \in \R$, we define
% scalar division for $\mathcal S$:
% \begin{equation}
%   \frac{(a,b)}{r} = \begin{cases}
%     \paren{\frac a r , \frac b r} & r \neq 0, \\
%     (0, -\infty) & r = 0.
%   \end{cases}
% \end{equation}

\begin{figure}
  \begin{mdframed}{\footnotesize\begin{algorithmic}[1]
    \Procedure{$bb$}{$\BBIR,
    R,
    b,
    P_{curr}$}
    \If {$R = \eset$}
      \State $n = f(P_{curr})$
      \Comment{$P_{curr}$ will be a total policy of $X$}
      \State \textbf{return } $\max(n, b)$\Comment{max uses the total order.}
    \Else
       \State $r = pop(R)$
       \For {$\ell \in \{r, \overline r\}$}
        \State $m = \mathsf{UB}_f((\{\varphi|_{\ell}\}, X, w), P_{curr} \cup \{\ell\})$\label{line:join}
        \If {$m \not\sqsubseteq b$} \label{line:prune}
          \State $n = bb(\{\varphi|_{\ell}, \gamma|_{\ell}\}, R, b, P_{curr} \cup \{\ell\})$
          \Comment{$n$ will always be from a policy}
          \State $b = \max(n,b)$
        \EndIf
       \EndFor
       \State \textbf{return } $b$
    \EndIf
  \EndProcedure
  \end{algorithmic}}\end{mdframed}
  \caption{The branch-and-bound style algorithm
  calculating the optimum of a function $f$ admitting an upper-bound function
  $\mathsf{UB}_f$ for every partial policy.
  The tuple $\BBIR$ is a BBIR,
  $R$ is the remaining search space (initialized to $X$),
  $b$ is a lower-bound,
  and $P_{curr}$ is the current partial policy (initialized to $\emptyset$).}
  \label{algorithm:bb}
\end{figure}

We give a quick walkthrough of \Cref{algorithm:bb}. If $R = \eset$,
we hit a base case, in which our accumulated policy, $P_{curr}$ is
a total policy. We evaluate the expected utility and update our upper bound
accordingly. If $R \neq \eset$, then we let $r$ be some variable in $R$ and
$\ell \in \{r, \overline r\}$ a literal.
Then we consider the extension of partial policy $P_{curr}$ with $\{\ell\}$,
which is still a partial policy. We compute an upper bound for the BBIR
conditioned on this partial policy to form $m$.

The pruning is at Line~\ref{line:prune}; if $m \not\sqsubseteq b$, then
there is no recursion, pruning any policies containing
$P_{curr} \cup \{\ell\}$. This pruning is sound, as shown by the
following theorem, proven in~\cref{appendix:soundness of bb proof}.

\begin{theorem}\label{thm:soundness of bb}
  Algorithm~\ref{algorithm:bb}
  solves the MEU and MMAP problems of~\cref{subsubsec:mmap,subsubsec:meu}.
  % More generally, given that $\mathsf{UB}_f$ is sound, Algorithm~\ref{algorithm:bb}
  % solves the optimization problem $\max_{\pi \in inst(X)} f(\pi)$.
\end{theorem}

\textit{Remark.}
It should be noted that, although we have put in hard work to take advantage of
the factorized representation of the BBIR as much as possible,
\Cref{algorithm:bb} can run in possibly exponential time with respect
to the size of $A$ in the worst case. This is because in the worst case we still
face the \textit{search-space explosion} discussed in
Section~\ref{sec:overview}.  The worst case will happen when there is no
pruning: if the guard of Line~\ref{line:prune} is always satisfied, we will
iterate through all possible partial models, which is of size $2^{|A|}$.

However,
we ensured that the inner-loop of partial and total policy evaluation (\cref{line:join}
of \Cref{algorithm:bb}) runs in polynomial time \textit{with respect
to the size of the already factorized representation of the BBIR}.
So, even though we have a search-space explosion,
we can much more efficiently search through that policy space than an approach
that does not leverage compilation.

%%% Local Variables:
%%% mode: LaTeX
%%% TeX-master: "../oopsla-appendix.tex"
%%% End:

\section{$\dappl$: A Language for Maximum Expected Utility}\label{sec:dappl}

In the next two sections we will showcase the flexibility of our new
branch-and-bound IR by using it to implement two languages that support
very different kinds of reasoning over optimization.  By design we keep these
languages small so that they can be feasibly compiled into BBIR: in particular,
they will both support only bounded-domain discrete random variables and
statically bounded loops. These two restrictions are common in existing
compiled PPLs such as \texttt{Dice}~\citep{holtzen2020scaling}.

In this section we describe the syntax and semantics of \dappl{}.
In order to do this we describe first a small sublanguage of \dappl{},
named \util{}, in Section~\ref{subsec:util}.
Our goal for the semantics of \util{} is to yield the expected utility
of a policy, akin to the computations
via expectations done
in~\cref{eqn:example}.
Then, in Section~\ref{subsec:dappl},
we give \dappl{}'s syntax as an extension of that of \util{},
and its semantics as an evaluation function $\mathsf{MEU}$, a maximization
over \util{} programs derived from
applying a policy to a \dappl{} program,
The compilation rules to BBIR are given in~\cref{subsec:compiling-dappl},
concluding with an example compilation of~\cref{fig:motivation-dappl} to BBIR.

\subsection{The Syntax and Semantics of~\util}\label{subsec:util}

\util{} is a small functional first-order probabilistic programming language
with support for Bayesian conditioning, if-then-else,
and \texttt{flip}s of a biased coin with bias in the interval $[0,1]$.
We augment the syntax with the additional syntactic form, $\reward k e$,
to specify a utility of $k$ awarded before evaluating expression $e$.

\begin{wrapfigure}{r}{0.6\linewidth}
  \begin{mdframed}
  {\footnotesize\begin{align*}
    \text{Atomic expressions } \texttt{aexp} ::= \ & x \mid \tt \mid \ff &\\
    \text{Logical expressions } P ::= \ & \texttt{aexp} \mid P \land P \mid P \lor P \mid \neg P &\\
    \text{Expressions } \texttt{e} ::=\ & \return P  \mid \flip{\theta} \mid  \reward{k}{e} &\\
    &\mid  \ite x e e &\\ &\mid \bind{x}{e}{e} \mid \observe{x}{e} &
  \end{align*}}
  \end{mdframed}
  \caption{Syntax of \util{}, our core calculus for computing expected utility without decision-making.}
  \label{fig: util syntax}
\end{wrapfigure}

The syntax of $\util$ is given in Figure \ref{fig: util syntax}.
Programs are expressions without free variables.
We distinguish between pure computations $P$, which
take the form of logical operations as the only values
are Booleans, and impure computations $e$, which
represent probablistic \texttt{flip}s, \texttt{reward} accumulation,
and their control flow.
Observed events take the form
of exclusively pure computations.
We enforce such restrictions via the
more general \dappl{} type system given in~\cref{appendix:typesystem}.
There are only two types in
\util{}: the Boolean type \Bool{} and distributions over $\Bool$, $\Giry \Bool$,
constructed via the Giry monad~\citep{giry2006categorical}.

The semantics follows the denotational approach
of~\citet{staton2020probabilistic} or~\citet{li2023lilac}.
Expressions $\Gamma \proves e : \Giry \Bool$\footnote{
all $\util$ expressions are of type $\Giry \Bool$, proven in~\citet{cho2025scaling}.}
are interpreted as
a function $\denote e$ from assignments of free variables to Booleans ($\denote \Gamma$)
to a distribution over either pairs of Booleans and reals or $\bot$:
$\mathcal D ((\Bool \times \R) \cup \{ \bot \})$.
The intuition is that utilities are attached to successful program executions--that is,
programs that do not encounter a falsifying \texttt{observe}.
A successful $\util$ program execution will either end in $\tt$ or $\ff$; the
\texttt{reward}s encountered along the way are summed up and weighted by the
probability of the successful trace.
For details see~\cref{appendix:util semantics}.

Using this definition, we can define the expected utility of a $\util$ program.

\begin{definition}\label{def:eu util}
  Let $\cdot \proves e : \Giry \Bool$ be a $\util$ program.
  Let $\mathcal D = \denote{e}\bullet$, where $\denote{e}$ is the map taking
  the empty assignment $\bullet \in \denote{\cdot}$
  to a distribution $\mathcal D$ over either pairs of Booleans and reals or $\bot$.
  The \emph{expected utility} of $e$ is defined to be the
  conditional expected value of the real values in $\mathcal D$
  attached to a successful program execution returning $\tt$
  conditional on not achieving $\bot$:
  \begin{equation}
    \EU(e) = \sum_{r \in \R} r \times \mathcal D ((\tt,r) \mid \text{not }\bot).
    \footnote{The sum is computable because there can only be a finite number of program traces evaluating to true.}
  \end{equation}
\end{definition}

\subsection{The Syntax and Semantics of \dappl}\label{subsec:dappl}

\dappl~augments the syntax of \util{}
(as shown in Figure \ref{fig: util syntax}) with two new expressions:
\begin{itemize}
  \item $[\alpha_1, \cdots, \alpha_n]$,
  where $\alpha_1,\cdots, \alpha_n$ are a nonzero number of fresh names,\footnote{
  We style
  the capitalization of names of $\alpha_1, \cdots \alpha_n$,
  in a manner consistent with how variant names are capitalized in ML.}
  to construct a \textit{choice} between
  binary \textit{alternatives} $\alpha_1, \cdots, \alpha_n$, and
  \item $\choose e {\alpha_i \implies e_i}$ to destruct a choice
  in a syntax akin to ML-style pattern matching.
\end{itemize}

However, writing a semantics for $\dappl{}$ in the same fashion as $\util{}$
is not as simple as it looks.
The problem lies in the type of optimization problem being solved:
recall that MEU takes the maximum over expected utilities
(see \cref{subsubsec:meu-example}).
In particular, we are not nesting maxima and expected utility calculation,
of the form $\max\sum\max\sum \cdots \sum f(x)$, which is not equal to,
in general, to the general form of an MEU computation $\max \sum f(x)$,
a phenomenon we noticed in~\cref{sec:overview}.

To avoid this,
we use \util's already established semantics to our advantage.
For a \dappl~program $e$ with $m$ many choices,
let $C_k$ denote the $k$-th choice in some arbitrary ordering.
Then we say $\mathcal A = C_1 \times C_2 \times \cdots \times C_m$
is the \textit{policy space} for the expression
in which elements $\pi \in \mathcal A$ are \textit{policies}.
In essence, each $\pi$ denotes a sequence of
valid alternatives that can are chosen in a \dappl~program.

Given a \dappl~program $e$ and a policy $\pi$ for the program,
we can reduce $e$ into a \util~program
by (1) removing any syntax constructing choices $[\alpha_1, \cdots, \alpha_n]$, and
(2) reducing each choice destructor $\choose e {\alpha_i \implies e_i}$ to the $e_i$
corresponding to the name $\alpha_i$ present in $\pi$.
We make formal this transformation in~\citet{cho2025scaling},
as well as prove it sound for well-typed \dappl~programs.
We refer as $e|_{\pi}$ the \util~program derived by applying policy $\pi$ to \dappl~program $e$.

With this in mind, we can introduce an \emph{evaluation function} $\mathrm{MEU}:
\dappl{} \to \R$ which computes the maximum expected utility, completing our
semantics. This evaluation function is proved total for all well-typed \dappl{}
programs in~\cref{appendix:util semantics}.

\begin{definition}\label{def:meu for dappl}
  For a well-typed \dappl~program $e$, define
  \begin{equation}\label{eq:meu for dappl}
    \MEUfn{e} \triangleq \max_{\pi \in \mathcal A} \EU(e|_{\pi}),
  \end{equation}
  in which $\mathcal A$ is the policy space defined by all of the decisions in $e$.
\end{definition}

We endow
\dappl{} with significant syntactic sugar, including
discrete random variables and statically-bounded loops.
\subsection{Compiling \dappl}
\label{subsec:compiling-dappl}

In Section~\ref{subsec:dappl} we described
the syntax and semantics of \dappl.
In Section~\ref{sec:bbir} we described the BBIR
and how it admits an algorithm to solve MEU with evidence.
Now we discuss \dappl{}'s compilation to BBIR, formalizing
our intuition from
computing the example in
Figure~\ref{fig:motivation-dappl} into
the BDD in Figure~\ref{fig:motiv-a-bdd}.

We compile \dappl~expressions into a tuple $(\varphi, \gamma, w, R)$, where:
\begin{itemize}[leftmargin=*]
  \item $\varphi$ is an \textit{unnormalized formula},
  representing the control and data flow without observations,
  \item $\gamma$ is an \textit{accepting formula},
  representing observations,
  \item $w: vars(\varphi) \to \mathcal S$ is a weight function, and
  \item $R$ is a set of reward variables.
\end{itemize}

We write $e \leadsto \target$
to denote that a \dappl~program compiles to the tuple $\target$.
Then we apply the map $\target \mapsto (\{\varphi \land R,\gamma\},  D(\varphi), w)$,
where $D(\varphi)$ is the set of Boolean variables representing choices in $\varphi$,
to transform it into a BBIR for Algorithm~\ref{algorithm:bb}.

Selected compilation rules are given in Figure~\ref{fig:dappl bc},
and full compilation rules are given in~\citet{cho2025scaling}.
Many rules are
influenced by similar compilation schemes found in the
literature~\citep{holtzen2020scaling,de2007problog,saad2021sppl}.
We use $T,F$ to denote true and false in propositional logic,
distinguishing it from the $\tt,\ff$ Boolean values in \dappl.
We write $\conjneg{R}$ to denote the conjunction of all
negations of variables in $R$. To remark on the intution behind
several rules:

\begin{figure}
\begin{mdframed}
{\footnotesize
\begin{align*}
  \infer[\texttt{bc/reward}]
  {
    \reward k  e \leadsto
    (\varphi, \gamma, R \cup \{r_k\},
      w \cup \{r_k \mapsto (1,k), \overline{r_k} \mapsto (1,0)\})}
  {
    \text{fresh } r_k
    & e \leadsto (\varphi, \gamma, R,w)
  }
\end{align*}
\begin{align*}
  \infer[\texttt{bc/[]}]
  {[a_1, \cdots, a_n] \leadsto
  (\exactlyone{(v_1,\cdots, v_n)}, \top, \{v_i \mapsto (1,0), \overline{v_i} \mapsto (1,0)\}_{i \leq n}, \eset )}
  {\text{fresh }v_1, \cdots, v_n}
\end{align*}
\begin{align*}
  \infer[\texttt{bc/ite}]
  {
    \ite{x}{e_T}{e_E} \leadsto
    \begin{gathered}
      \big(
        (x \land \varphi_T \land R_T \land \conjneg{R_E})
          \lor (\overline{x} \land \varphi_E \land R_E \land \conjneg{R_T}), \\
        (x \land \gamma_T) \lor (\overline{x} \land \gamma_E),
        w_T \cup w_E,
        \eset \big)
    \end{gathered}
  } {
    x \leadsto
    (x, \top, \eset, \eset, \eset)
    &
    e_T \leadsto
    (\varphi_T, \gamma_T, w_T, R_T)
    &
    e_E \leadsto
    (\varphi_E, \gamma_E, w_E, R_E)
  }
\end{align*}
\begin{align*}
  \infer[\texttt{bc/choose}]
  {\choose x {a_i \implies e_i}
  \leadsto
  \begin{gathered}
    \Big(x \land \bigvee_i (a_i \land e_i
    \land \bigwedge_{j \neq i} \conjneg{R_j}),
    x \land \bigvee_i (a_i \land \gamma_i), \\
    \bigcup_i w_i, \bigcup_i R_i \Big)
  \end{gathered}
}
  {x \leadsto
  (x, \top, \eset, \eset, \eset)
  &
  \forall \ i.  \ e_i \leadsto (\varphi_i, \gamma_i, w_i, R_i)}
\end{align*}
}
\end{mdframed}
\caption{Selected Boolean compilation rules of \texttt{dappl}.
For complete rules see~\cref{appendix:dappl bc}.
}
\label{fig:dappl bc}
\end{figure}

\begin{enumerate}[leftmargin=*]
  \item The union of weight functions $w \cup w'$ is
  non-aliased -- there will never be
  $x \in \mathrm{dom}(w) \cap \mathrm{dom}(w')$ such that
  $w(x) \neq w'(x)$ or $w(\overline x) \neq w' (\overline x)$.
  \item The \texttt{bc/[]} enforces an ExactlyOne constraint
  on the introduced fresh Boolean variables $v_1,\cdots, v_n$.
  This is to disallow the behavior of evaluating multiple
  patterns in a \texttt{choose} statement.
  \item \texttt{bc/ite}
  enforces the condition that one cannot
  incorporate the rewards of one branch while branching into
  another by conjoining $\conjneg{R_E}$ and $\conjneg{R_T}$
  onto the disjuncts.
  We did this implicitly in the examples of
  Section~\ref{sec:overview} --
  without this constraint, we would get the incorrect
  expected utility for the policy \texttt{Umbrella},
  as the model $\{u, r, R_{10}, R_{-5}, R_{100}\}$
  would be a valid model.
  The $\bigwedge_{j \neq i} \conjneg{R_j}$ in \texttt{bc/choose}
  imposes the same restriction for choice pattern matching.
  \item We reset the accumulated rewards in \texttt{bc/ite}, as the rewards
  need to be scaled by the probability distribution defined by the value to be
  substituted into $x$. Thus, we start discharge our accumulated rewards to scale them
  appropriately and start anew.
\end{enumerate}

\begin{figure}
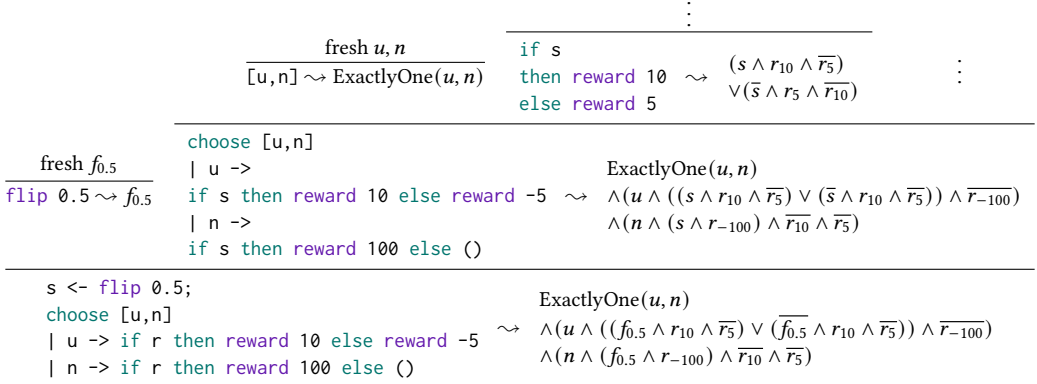

{\footnotesize
\begin{align*}
  \infer
  {
    \begin{array}{l}
      \dapplcode{s <- flip 0.5;} \\
      \dapplcode{choose [u,n]} \\
      \dapplcode{| u -> if r then reward 10 else reward -5} \\
      \dapplcode{| n -> if r then reward 100 else ()}
    \end{array}
    \leadsto
    \begin{array}{l}
      \mathrm{ExactlyOne}(u,n) \\
      \land (u \land ((f_{0.5} \land r_{10} \land \overline{r_5})
      \lor (\overline{f_{0.5}} \land r_{10} \land \overline{r_5})) \land \overline{r_{-100}})\\
      \land (n \land (f_{0.5} \land r_{-100}) \land \overline{r_{10}} \land \overline{r_5})\\
    \end{array}
  }
  {
    \infer
    {
      \dapplcode{flip 0.5} \leadsto f_{0.5}
    }
    {
      \text{fresh }f_{0.5}
    }
    &
    \infer
    {
      \begin{array}{l}
        \dapplcode{choose [u,n]} \\
        \dapplcode{| u ->} \\
        \dapplcode{  if s then reward 10 else reward -5} \\
        \dapplcode{| n ->} \\
        \dapplcode{  if s then reward 100 else ()}
      \end{array}
      \leadsto
      \begin{array}{l}
        \mathrm{ExactlyOne}(u,n) \\
        \land (u \land ((s \land r_{10} \land \overline{r_5})
        \lor (\overline{s} \land r_{10} \land \overline{r_5})) \land \overline{r_{-100}})\\
        \land (n \land (s \land r_{-100}) \land \overline{r_{10}} \land \overline{r_5})\\
      \end{array}
    }
    {
      \infer
      {
        \dapplcode{[u,n]} \leadsto \mathrm{ExactlyOne}(u,n)
      }
      {
        \text{fresh } u,n
      }
      &
      \infer
      {
        \begin{array}{l}
        \dapplcode{if s} \\
        \dapplcode{then reward 10} \\
        \dapplcode{else reward 5}
        \end{array}
        \leadsto
        \begin{array}{l}
          (s \land r_{10} \land \overline{r_5}) \\
          \lor (\overline{s} \land r_{5} \land \overline{r_{10}})
        \end{array}
      }
      {
        \vdots
      }
      &&&&
      \vdots
    }
  }
\end{align*}
}
\caption{Partial compilation tree of the code in Figure~\ref{fig:motivation-dappl},
showing the compiled unnormalized formula. We omit the accepting formula
as it evalutes to $\top$ as there is no evidence. We give only $\varphi$
for visual clarity.}
\label{fig:compilation tree}
\end{figure}

The following theorem
connects the \dappl{} semantics of Section~\ref{sec:dappl}
to the branch-and-bound algorithm discussed in
Section~\ref{subsec:meu with evidence}.
For proofs see~\cref{appendix:dappl correctness}:

\begin{theorem}\label{thm:compiler correctness}
  Let $e$ be a well-typed \dappl{} program.
  Let $e \leadsto \target$. Then we have
  \begin{equation}
    \MEUfn{e} = \mathrm{bb}(\{\varphi \land R, \gamma\}, w, D(\varphi)).
  \end{equation}
\end{theorem}

To see this theorem in action, we return to our original example code in
Figure~\ref{fig:motivation-dappl}. It compiles to the Boolean formula seen in
Figure~\ref{fig:compilation tree}. Let the compiled formula be $\varphi$.
Then we see that
\begin{align}
  \varphi|_u &= (f_{0.5} \land r_{10} \land \overline{r_5})
  \lor (\overline{f_{0.5}} \land r_{10} \land \overline{r_5}) \land \overline{r_{-100}} \\
  \varphi|_n &= (f_{0.5} \land r_{-100}) \land \overline{r_{10}} \land \overline{r_5}.
\end{align}
The $\AMC$ of $\varphi|_u$ and $\varphi|_n$ exactly match that of
$\varphi_u$ and $\varphi_{\overline u}$ in Equation~\ref{eq:formula-umbrella},
which completes the picture.

\section{$\pineappl$: A Language for MMAP}\label{sec:pineappl}

In this section, we describe the syntax, semantics, and boolean compilation of
$\pineappl$. $\pineappl$ is
different from $\dappl$ in the fact that it is
a first-order, imperative probabilistic programming
language with support for first-class MMAP computation, along with
marginal probability computations.
Much like the organization of~\cref{sec:dappl},
we will first introduce the syntax and semantics (\cref{subsec:pineappl-sem}),
then outline the Boolean compilation (\cref{subsec:pineappl-compl}).

\subsection{Syntax and Semantics of $\pineappl$}
\label{subsec:pineappl-sem}

\begin{wrapfigure}{r}{0.61\linewidth}
  \begin{mdframed}
    {\footnotesize\begin{align*}
    \text{Expressions } \texttt{e} ::= \ & \texttt{x} \mid \tt \mid \ff \mid
    \texttt{e} \land \texttt{e} \mid \texttt{e} \lor \texttt{e} \mid
    \neg \texttt{e} &\\
    \text{Statements } \texttt{s} ::= \ & \texttt{x} = e \mid x = \flip{\theta} \mid \texttt{s} ; \texttt{s}&\\
    &\mid
    \texttt{if e \{s\} else \{s\}} &\\
    &\mid (\texttt{m}_1, \hdots, \texttt{m}_n) = \texttt{mmap}(\texttt{x}_1, \hdots, \texttt{x}_n) &\\
    &\mid (\texttt{m}_1, \hdots, \texttt{m}_n) = \texttt{mmap}(\texttt{x}_1, \hdots, \texttt{x}_n) \; \texttt{with} \; \{ \texttt{e} \} &\\
    \text{Programs } P ::= \ & \texttt{s} ; \Pr(\texttt{e}) \mid \texttt{s} ; \Pr(\texttt{e}) \; \texttt{with} \; \{ \texttt{e}  \}
    \end{align*}}
  \end{mdframed}
    \caption{Full $\pineappl$ syntax.}
    \label{fig:pineappl-syntax}
\end{wrapfigure}

The full syntax of pineappl is in~\cref{fig:pineappl-syntax}.  A $\pineappl$
program is made of two parts: statements and a query.  Statements consist of
(1) variables bound to either \pineapplcode{flip}s or expressions over them,
(2) a \pineapplcode{mmap} statement for binding a set of variables
$\texttt{x}_1, \hdots, \texttt{x}_n$
to the MMAP state of variables $\texttt{m}_1, \hdots, \texttt{m}_n$, or
(3) a sequence of the above.
A query asks for the marginal probabilty of an
expression.
We assume all variables have unique names.
Note that
\pineapplcode{mmap} and \pineapplcode{Pr} can be followed by \pineapplcode{ with \{e\}},
denoting the \textit{observation} of expression \pineapplcode{e}.
We impose the additional restriction that no variables referenced in the observed
expression have been bound by \pineapplcode{mmap}.
We endow more sugar in the full language, including support for
multiple queries, categorical discrete random variables, and bounded loops
in~\citet{cho2025scaling}.

\begin{figure}
\begin{mdframed}
  {\footnotesize
  \begin{align*}
    \infer[\texttt{s/flip}]
    {
      (\texttt{x = flip} \ \theta, \mathcal{D}) \Downarrow
      \{\sigma \cup [x \mapsto \top] \mapsto \theta \times \mathcal{D}(\sigma)\} \cup
      \{\sigma \cup [x \mapsto \bot] \mapsto (1 - \theta) \times \mathcal{D}(\sigma)\}
    }
    {
      \mathrm{fresh} \; x
    }
  \end{align*}
  \begin{align*}
    \infer[\texttt{s/assn}]
    {
      (\texttt{x = e}, \mathcal{D}) \Downarrow
      \begin{gathered}
      \{\sigma \cup [x \mapsto \top] \mapsto p \times\mathcal{D}(\sigma) \mid e[\sigma] = \top, \sigma \in \dom(\mathcal{D})\}
      \\
      \cup
      \{\sigma \cup [x \mapsto \bot] \mapsto (1-p) \times \mathcal{D}(\sigma) \mid e[\sigma] = \bot, \sigma \in \dom(\mathcal{D})\}
      \end{gathered}
    }
    {
      \mathrm{fresh} \; x
      &
      p = \Pr_{\mathcal D}[\texttt e]
    }
  \end{align*}
  \begin{align*}
    \infer[\texttt{s/seq}]
    {
        (s_1; s_2, \mathcal{D}) \Downarrow \mathcal{D}''
    }
    {
        (s_1, \mathcal{D}) \Downarrow \mathcal{D}'
        & (s_2, \mathcal{D}') \Downarrow \mathcal{D}''
    }
  \end{align*}
  \begin{align*}
    \infer[\texttt{s/if}]
    {
      (\texttt{if e \{s$_1$\} else \{s$_2$\}}, \mathcal{D}) \Downarrow
      \{\sigma \mapsto p \times \mathcal{D}_1(\sigma) +
      (1-p) \times \mathcal{D}_2(\sigma)| \sigma \in \dom(\mathcal{D}_1)\}
    }
    {
      (\texttt s_1, \mathcal{D}) \Downarrow \mathcal{D}_1
      & (\texttt s_2, \mathcal{D}) \Downarrow \mathcal{D}_2
      & \Pr_{\mathcal D}[\texttt e] = p
    }
  \end{align*}
  \begin{align*}
    \infer[\texttt{s/mmap}]
    {
      (\vec{\texttt{m}}\texttt{ = mmap }\vec{\texttt{x}}, \mathcal{D}) \Downarrow
      \{\sigma \cup \sigma_m \mapsto \mathcal{D}(\sigma) \mid \sigma \in \dom(\mathcal{D}) \}
    }
    {
      \vec{A} = MMAP_{\mathcal D}(\vec{\texttt{x}})
      & \sigma_m = \{m_i \mapsto A_i \mid i \in [1,n]\}
    }
  \end{align*}
  \begin{align*}
    \infer[\texttt{s/mmap/with}]
    {
      (\vec{\texttt{m}}\texttt{ = mmap }\vec{\texttt{x}} \ \pineapplcode{ with \{e\}}, \mathcal{D})
      \Downarrow
      \{\sigma \cup \sigma_m \mapsto \mathcal{D}(\sigma) \mid \sigma \in \dom(\mathcal{D}) \}
    }
    {
      \vec{A} = MMAP_{\mathcal D}(\vec{\texttt{x}} \mid e)
      & \sigma_m = \{m_i \mapsto A_i \mid i \in [1,n]\}
    }
  \end{align*}
  \begin{align*}
    \infer[\texttt{p/pr}]
    {
      \texttt{s; Pr(e)}  \Downarrow_P \Pr_{\mathcal D} [\texttt{e}]
    }
    {
      (\texttt s, \eset) \Downarrow \mathcal{D}
    } \;\;
    \qquad
    \infer[\texttt{p/pr/with}]
    {
      \texttt{s; Pr(e$_1$) with \{e$_2$\}} \Downarrow_P \frac{\Pr_{\mathcal{D}}[\texttt e_1 \land \texttt e_2]}{\Pr_{\mathcal{D}}[\texttt e_2]}
    }
    {
    (\texttt s, \eset) \Downarrow \mathcal{D}
    }
  \end{align*}
  }
\end{mdframed}
\caption{Operational semantics for $\pineappl$.
  All variable names in a $\pineappl$ program are assumed unique.
  $x_i$ denotes the $i$-th component of a vector $\vec{x} = (x_1, \cdots, x_n)$.}
\label{fig:pineappl-os}
\end{figure}

$\pineappl{}$'s semantics are given by two relations: $\Downarrow$ and $\Downarrow_P$,
described in~\cref{fig:pineappl-os}.
The
$\Downarrow$ relation is a big-step operational semantics
relating pairs of statements and distributions $(s, \mathcal D)$
to a new distribution $\mathcal D'$.
These distributions are over assignments of variables.
The
$\Downarrow_P$ relation relates a $\pineappl{}$ program $P = s; q$
to a real number correponding to the probability of query $q$.

To remark on the notation behind several rules:

\begin{enumerate}[leftmargin=*]
  \item The $\Pr_{\mathcal D}[e]$ notation used in
  \texttt{s/assn}, \texttt{s/if}, \texttt{p/pr}, \texttt{p/pr/with}
  denotes the probability of the event in $\mathcal D$ that the Boolean expression $e$ is satisfied.
  \item $MMAP_{\mathcal D}$, as used in \texttt{s/mmap} and \texttt{s/mmap/with},
  is the marginal MAP operator of some vector of variables $\vec{\mathtt{x}}$
  over a distribution $\mathcal D$, potentially conditioned on an expression $e$. More precisely we can define $MMAP_{\mathcal D}$ as follows:
  \begin{equation}
    MMAP_{\mathcal D}(\vec{\mathtt x} \mid e) = \max_{\sigma \in inst(\vec{\mathtt x})} \mathcal D(\sigma \mid e),
  \end{equation}
  where $\mathcal D(\sigma \mid e)$ is the probability of the instantiation $\sigma$ in $\mathcal D$ conditional on $e$.
\end{enumerate}

Finally, we can query the probability of an expression \texttt{e} over the
compiled distribution via $\Downarrow_P$. To handle observation, \pineapplcode{Pr(e) with \{o\}},
as with the rule \texttt{p/pr/with},
we first compute the unormalized probability of the observation being
true jointly with the query, $\Pr_{\mathcal{D}}[e \land o = \tt]$, and then divide by the
normalizing constant, $\Pr_{\mathcal{D}}[o = \tt]$; this is Bayes' rule.

\subsection{Boolean Compilation of $\pineappl{}$}
\label{subsec:pineappl-compl}
\begin{figure}
\begin{mdframed}
  {\footnotesize
  \begin{align*}
    \infer[\texttt{bc/mmap}]
    {
      (\vec{\texttt{m}}\texttt{ = mmap }\vec{\texttt{x}}, \mathcal{F}, w) \leadsto
      (\mathcal{F} \cup \{(m_i, k_i) \}, w \cup w_{M})
    }
    {
      \mathrm{fresh} \; k_i
      & \vec A = MMAP(\{\bigwedge_{(x, \varphi) \in \mathcal{F}} x \leftrightarrow \varphi, \eset\}, \vec{\texttt{x}}, w)
      &
      w_{M} = \{m_i \mapsto (1,1), k_i \mapsto A_i \}
    }
  \end{align*}
  \begin{align*}
    \infer[\texttt{bc/mmap/with}]
    {
      (\vec{\texttt{m}}\texttt{ = mmap }\vec{\texttt{x}} \ \pineapplcode{ with \{e\}}, \mathcal{F}, w) \leadsto
      (\mathcal{F} \cup \{(m_i, k_i) \}, w \cup w_{M})
    }
    {
      \mathrm{fresh} \; k_i
      & \texttt e \leadsto_E \psi
      & \vec A = MMAP(\{\bigwedge_{(x, \varphi) \in \mathcal{F}} x \leftrightarrow \varphi, \psi\}, \vec{\texttt{x}}, w)
      &
      w_{M} = \{m_i \mapsto (1,1), k_i \mapsto A_i \}
    }
  \end{align*}
  \begin{align*}
    \infer[\texttt{bc/pr}]
    {
      \texttt{s; Pr(e)} \leadsto_P (\chi \land \left( \bigwedge_{(x, \varphi) \in \mathcal{F}} x \leftrightarrow \varphi \right), \top,w)
    }
    {
      (\texttt{s}, \eset, \eset) \leadsto (\mathcal F, w)
      &
      \texttt{e} \leadsto_E \chi
    }
  \end{align*}
  \begin{align*}
    \infer[\texttt{bc/pr/with}]
    {
      \pineapplcode{s; Pr(e$_1$) with \{e$_2$\}} \leadsto_P (\chi \land \left(\bigwedge_{(x, \varphi) \in \mathcal{F}} x \leftrightarrow \varphi \right), \psi, w)
    }
    {
      (s, \eset, \eset) \leadsto (\mathcal F, w)
      &
      \texttt{e$_1$} \leadsto_E \chi
      &
      \texttt{e$_2$} \leadsto_E \psi
    }
  \end{align*}
  }
\end{mdframed}
\caption{Selected Boolean compilation rules for $\pineappl$. As shorthand, we write
$w \cup \{x \mapsto (a, b)\}$ instead of $w \cup \{(x \mapsto \top) \mapsto a, (x \mapsto \bot) \mapsto b\}$.
The $\leadsto_E$ relation translates expressions into Boolean formulae; explicit rules are given in~\cref{appendix:pineappl bc}. The symbol $\leftrightarrow$ denotes logical if-and-only-if.}
\label{fig:pineappl-compl}
\end{figure}

%%% Local Variables:
%%% mode: LaTeX
%%% TeX-master: "../oopsla-appendix"
%%% End:

Like $\dappl{}$, we compile $\pineappl$ programs to Boolean formulae as a tractable
representation. Key rules are in~\cref{fig:pineappl-compl} and full rules are in~\cref{appendix:pineappl-compl-complete}.
The BBIR is used in the \texttt{bc/mmap} and \texttt{bc/mmap/with} rule, where
the premise $MMAP$ is identical to that defined in~\cref{subsubsec:mmap}, and is solved via
Algorithm~\ref{algorithm:bb}. We define three relations:
\begin{itemize}[leftmargin=*]
  \item $e \leadsto_E \varphi$ compiles a $\pineappl$ expression to a Boolean formula,
  \item $(s, \mathcal{F}, w) \leadsto (\mathcal{F}', w)$ compiles a $\pineappl$ statement $s$,
  a set of pairs of identifers and formulae $\mathcal F$, and a weight map of literals $w$
  into a set $\mathcal{F'}$ and weight map $w'$, and
  \item $s ; q \leadsto_P (\varphi, \psi, w)$ with an unnormalized formula $\varphi$, an accepting formula $\psi$, and a weight map $w$.
\end{itemize}
% The restriction of $\mathcal F$ to formulas which take the form $x \leftrightarrow \varphi$
% is important as it preserves the binding structure of new variables.

To conclude the section, we give a correctness theorem, akin to~\cref{thm:compiler correctness}, proven in~\cref{appendix:proof-pineappl-correctness}.
% ~\citet{cho2025scaling}.

\begin{theorem}\label{thm:pineappl correctness}
  For a $\pineappl{}$ program $s;q$, let $s ; q\Downarrow_P p$ and
  $s ; q \leadsto_P (\chi \land \paren{\bigwedge_{(x, \varphi) \in \mathcal{F}} x \leftrightarrow \varphi},\psi, w)$.
  Then
  {\footnotesize\begin{equation}
    p = \frac{\AMC_{\R}\left(\chi \land \paren{\bigwedge_{(x, \varphi) \in \mathcal{F}} x \leftrightarrow \varphi} \land \psi,w\right)}{\AMC_{\R}(\psi \land \paren{\bigwedge_{(x, \varphi) \in \mathcal{F}} x \leftrightarrow \varphi}, w)}.
  \end{equation}}
\end{theorem}

%%% Local Variables:
%%% mode: LaTeX
%%% TeX-master: "../oopsla-appendix"
%%% End:

\section{Empirical Evaluation of $\dappl$ and $\pineappl$}\label{sec:eval}

\cref{sec:bbir} outlined how BBIR can both \textit{factorize} program
structure and \textit{prune} ineffective strategies over such a representation.
But, the question still remains: does this translate into a fast
language for optimization in practice?
To answer this question we compare \dappl{} and \pineappl{} against
existing languages to express and solve MEU and MMAP problems.\footnote{\emph{Evaluation and implementation details:} All timings of
benchmarks were run on a single thread,
on a server with 512GB of RAM and two AMD EPYC 7543
CPUs.
The BBIR and and associated algorithms are written in Rust.
$\pineappl$ was written in Rust, while $\dappl$ was written in OCaml.
When feasible, the output by ProbLog and its variants
were verified to match the policies output by \dappl{} and $\pineappl$.
}

\subsection{Empirical Evaluation of $\dappl$}\label{subsec:dappl-eval}
We compared $\dappl$'s MEU evaluation via BBIR to two existing approaches:

\begin{itemize}[leftmargin=*]
  \item \textit{Enumeration.}
  Every possible policy is enumerated, then evaluated according to the expected utility.
  % This strategy closely resembles an enumerative solution to solving the optimization problem
  % given in \ref{def:meu for dappl}, perhaps with additional room for
  % factorized representations of the program and dynamic programming to reduce the search space.
  We compare against \problog~2 as a representative of this strategy~\citep{de2007problog}.
  \item \textit{Order-constrained 2AMC approaches}.
  % The key hardness of MEU problems comes from the fact that maximization must happen
  % after summing; hence, most approaches to MEU work by sequentially summing and then
  % maximizing (as efficiently as possible).
  % \citet{kiesel2022efficient}
  % introduced a strategy for solving optimization-over-semiring problems
  % in which all the choice variables are sifted to the top of the variable ordering,
  % reducing the MEU problem into a model counting problem over two different semirings.
  \citet{derkinderen2020algebraic} introduced a state-of-the-art decision-theoretic \problog{}
  implementation that
  compiles programs into an order-constrained representation;
  we use this implementation as a representative strategy from this category.
\end{itemize}

Thus, we generated several benchmarks as both \dappl{} and \problog{} programs
to test the performance of the IRs.  As of yet there is no standard suite of
benchmarks for evaluating the MEU task, so we generated a new set of benchmarks
for validating performance.  Throughout our experiments we made a
best-effort attempt to write the most efficient programs in all languages.

\subsubsection{Bayesian Network Experiments}\footnote{Bayesian networks were selected
from \url{https://www.bnlearn.com/bnrepository/}.}
\label{sec:bn-eval}
  Bayesian networks are a well-established source of difficult, realistic,
  and useful probabilistic inference problems. It is straightforward to
  translate a Bayesian network into a \dappl{} or \problog{} program.
  However, Bayesian networks only represent probabilistic inference, and not decision making.
  We generated a standardized suite of challenging decision-theoretic problems on
  Bayesian networks by following the process in~\citet{derkinderen2020algebraic}.
  First, we transformed the root nodes of a Bayesian network into a decision.
  Then, if there were less than four decisions made through this process,
  each node of the Bayesian network was converted into a decision
  with probability 0.5. Utilities were added via one of two random methods:
  \begin{enumerate}[leftmargin=*]
    \item For each node in the Bayesian network, a utility of an integer
    between 0 and 100 was assigned with probability 0.8 for when the node yielded true,
    and assigned with probability 0.3 for when the node yielded false.
    \item We introduced five new ``reward nodes'' in the Bayesian network,
    on which rewards were assigned whether it was true or not.
    The reward nodes are true if and only if at least one of five randomly generated
    assignments to the existing nodes of the Bayesian network are true.
  \end{enumerate}
  We call the first utility assignment strategy ``Existing'',
  and the second strategy ``New nodes''.
  The Bayesian networks studied were Asia, Earthquake, and Survey,
  as they were the ones studied in previous work~\citep{derkinderen2020algebraic}.
  % \sh{say something about size; what's the biggest we can handle}
  % cant do this bc time constraint
Table \ref{table : bn} reports the performance of \dappl{} in comparison with
\dtproblog{}. We observe that \dappl~excels at computing the MEU over all three
Bayesian networks, across both methods to add utilities.
It is not surprising to
see an improvement over the enumerative strategy, but it is surprising to see that
the cost of
constraining the variable order to have choices-first is burdensome to the point
of timeout.
This is most likely because moving each choice to the
top of the order can lead to blowup, and this happens multiple times.

\subsubsection{Scaling Experiments}

In these experiments, we generate a family
of progressively larger examples to study how \dappl{} and \dtproblog{} scale
as the size of the example grows.
\begin{itemize}[leftmargin=*]
  \item \textit{Diminishing Returns (DR).}
  The scenario goes as this: we flip a coin with some bias.  If heads, we choose
  between 2-6 utilities, uniformly distributed between 0 and 100.  If tails, we
  flip another coin with another bias, but enter the same scenario.  This
  example scales in $n$ coin flips.
  This behavior is nicely modeled in \dappl{}: see the supplementary
  materials for example programs.
  The decision scenario has a simple solution to us: since each decision is independent
  of each other, it suffices to pick the choice maximizing the utility for each coin flip.
  \item \textit{One-shot faulty network ladder diagnosis (One-shot ladder).}
  We adapt a ladder network model
  as outlined in~\citet{holtzen2020scaling} into a decision--theoretic scenario.
  The network topology is visualized as follows: \begin{tikzpicture}[node distance=0.2cm, baseline=-1em]
  \node[] (init1) {$\dots$};
  \node[below = of init1] (init2) {$\dots$};
  \node[draw, circle, right=of init1] (S1) {};
  \node[draw, circle, right=of init2] (S2) {};
  \node[draw, circle, right=of S1] (S3) {};
  \node[draw, circle, right=of S2] (S4) {};
  \node[right=of S3] (final1) {$\dots$};
  \node[right=of S4] (final2) {$\dots$};

  \draw[->] (init1) -- (S1);
  \draw[->] (init2) -- (S2);

  \draw[->] (S1) -- (S3);
  \draw[->] (S1) -- (S4);

  \draw[->] (S2) -- (S3);
  \draw[->] (S2) -- (S4);

  \draw[->] (S3) -- (final1);
  \draw[->] (S4) -- (final2);
  % \draw[->] (S1) -- (S2);
  % \draw[->] (S1) -- (S3);
  % \draw[->] (S2) -- (S4);
  % \draw[->] (S3) -- (S4);
  % \draw[->] (S4) -- (final);
\end{tikzpicture}. Each circle represents a router, and each arrow
represents a link.
  We construct a ladder network with $2n$ routers, where $n$ is the scaling parameter. We observe that
  an incoming packet does not make it to the end of the network. Then,
  the task is to find a faulty router. If we choose a faulty router,
  then we obtain a reward uniformly distributed between 0 and 100;
  otherwise we receive a reward of 0.
  This benchmark is difficult as performing inference on the network is
  already quite difficult~\citep{holtzen2020scaling}
  but we additionally introduce a choice with $2n$ many alternatives.
  \item \textit{$k$-shot faulty network ladder diagnosis ($k$-shot ladder).}
  We keep the same ladder network as above,
  but if we fail to find a faulty router the first try,
  we can continue up until $k$ tries,
  where $k$ is less than the total number of nodes in the network ladder.
  This benchmark is the hardest, as the number of possible policies is factorial
  with respect to the number of nodes.
\end{itemize}
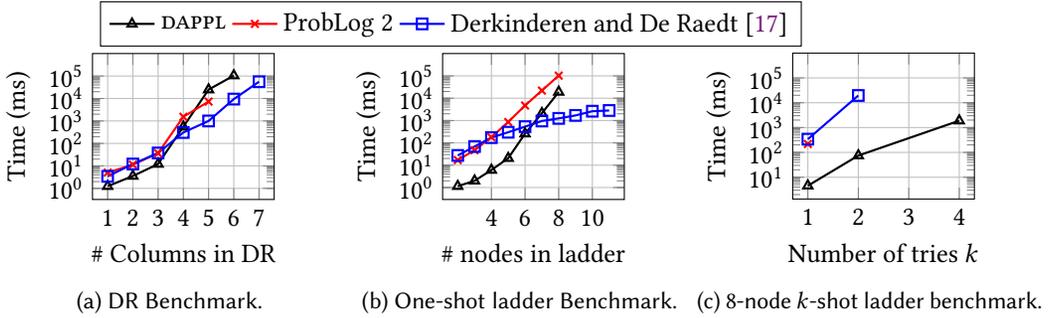
\begin{figure}
  \begin{subfigure}[t]{0.33\linewidth}
    \centering
     \begin{tikzpicture}
    \begin{axis}[
		height=3.5cm,
    ymin=0,
    ymax=1000000,
		grid=major,
    width=4cm,
    % xmode=log,
    ymode=log,
  xtick={1,2,3,4,5,6,7,8},
    ytick={1,10,100,1000,10000,100000},
    xlabel={\# Columns in DR},
    ylabel={Time (ms)},
    legend columns=6,
    legend style={at={(13em,7em)},anchor=north}
    ]
    \addplot[mark=triangle, thick] table [x index={0}, y index={1}] {\hmmdata};
    \addlegendentry{\dappl{}};
    \addplot[mark=x, thick, red] table [x index={0}, y index={2}] {\hmmdata};
    \addlegendentry{ProbLog 2};
    \addplot[mark=square, thick, blue] table [x index={0}, y index={3}] {\hmmdata};
    \addlegendentry{\citet{derkinderen2020algebraic}};
  \end{axis}
  \end{tikzpicture}

  \caption{DR Benchmark.}
  \label{fig : dr}
  \end{subfigure}
  ~
  \begin{subfigure}[t]{0.33\linewidth}
     \begin{tikzpicture}
    \begin{axis}[
		height=3.5cm,
    ymin=0,
    ymax=1000000,
		grid=major,
    width=4cm,
    ymode=log,
    xtick={1,4,6,8,10},
    ytick={1,10,100,1000,10000,100000},
    xlabel={\# nodes in ladder},
    ylabel={Time (ms)},
    ]
    \addplot[mark=triangle, thick] table [x index={0}, y index={1}] {\ladderlongdata};
    \addplot[mark=x, thick, red] table [x index={0}, y index={2}] {\ladderlongdata};
    \addplot[mark=square, thick, blue] table [x index={0}, y index={3}] {\ladderlongdata};
  \end{axis}
  \end{tikzpicture}

  \caption{One-shot ladder Benchmark.}
  \label{fig : 1shotladder}
  \end{subfigure}
  ~
  \begin{subfigure}[t]{0.33\linewidth}
     \begin{tikzpicture}
    \begin{axis}[
		height=3.5cm,
    ymin=0,
    ymax=1000000,
		grid=major,
    width=4cm,
    ymode=log,
    xtick={1,2,3,4},
    ytick={1,10,100,1000,10000,100000},
    xlabel={Number of tries $k$},
    ylabel={Time (ms)},
    ]
    \addplot[mark=triangle, thick] table [x index={0}, y index={1}] {\ladderfour};
    \addplot[mark=x, thick, red] table [x index={0}, y index={2}] {\ladderfour};
    \addplot[mark=square, thick, blue] table [x index={0}, y index={3}] {\ladderfour};
  \end{axis}
  \end{tikzpicture}
  \caption{8-node $k$-shot ladder benchmark.}
  \label{fig : ladder4}
  \end{subfigure}
  \caption{Scaling results comparing \dappl{}, ProbLog 2, and \citet{derkinderen2020algebraic} on
  MEU tasks.
  The average number of choices in DR is $4 \times \text{\# of columns}.$
  The number of choices in one-shot ladder is twice the number of nodes.
  The number of choices in one-shot ladder is $\prod_{i \leq \text{\# of tries}} 8-(1-i)$.}
  \label{fig:scaling}
\end{figure}
{\footnotesize
\begin{table}
\caption{Comparison of different MEU tools on Bayesian network benchmarks.
  Time is in milliseconds (ms), with timeout 5 minutes = 300000ms.
  All reported times are the average over several runs;
  see the text for details.
  ``Avg. Times Pruned'' is the average number of times a partial policy (of any size)
  was not traversed in our randomly generated experiments.
  }
\begin{tabular}{lllll|l}
\toprule
\multicolumn{2}{c}{}  & $\dappl$ & \problog~2 & 2AMC &\\
Bayesian Network & Utility Method &  &  &  & Avg. Times Pruned\\
\midrule
\multirow[m]{2}{*}{Asia} & Existing & \textbf{1.4±0.3} & 28.6±11.4 & 86.5±40.8 & 7.6\\
 & New nodes & \textbf{6.0±0.3} & 53.4±5.5 & 119.2±20.7 & 4.7\\
\cline{1-6}
\multirow[m]{2}{*}{Earthquake} & Existing & \textbf{1.0±0.2} & 15.2±4.8 & 19.4±5.6 &3.0\\
 & New nodes & \textbf{2.4±0.2} & 33.2±2.3 & 24.6±4.5 & 3.7\\
\cline{1-6}
\multirow[m]{2}{*}{Survey} & Existing & \textbf{8.3±0.8} & 319.1±194.3 & 16532.8±1096.3 &3.4\\
 & New nodes & \textbf{103±0.8} & 182.3±43.2 & 19485.3±8173.8 & 2.7\\
\cline{1-6}
\bottomrule
\end{tabular}
\label{table : bn}
\end{table}
}
The results of these scaling experiments are reported in~\cref{fig:scaling}.
We observe that in the DR and one-shot ladder benchmarks,
\dappl{} feels the effects of its theoretical worst-case performance,
performing marginally better or worse than its competitors.
We believe that the primary reason \dappl{} scales poorly for these
examples is because
as the number of policies grow, there are many policies that are similar in expected reward yet incompatible,
decreasing opportunities for pruning.
% There are several reasons why \dappl~may
% suffer on this example.
% The first is that the expected utility of each policy requires inference
% on the ladder network, which is difficult to scale without the use of
% first-class functions and iterators \`a la Dice~\citep{holtzen2020scaling}.
% The second is that the number of alternatives is $2n$, and having an
% $\text{ExactlyOne}(2n)$ constraint in the IR greatly increases the treewidth,
% which is known to be a large hindrance in scaling inference problems~\citep{darwiche2002knowledge}.
On the contrary, we see that the order-constrained approach of 2AMC IR
is particularly performant on this task.
We believe this is because the structure of this problem is particularly
amenable to a constrained approach: the decision problem and the
ladder network can be defined almost entirely separately from each other,
resulting in an easier constraint on the order. Furthermore, bringing
all choice variables to the front of the order mitigates much of the treewidth
blowup faced in \dappl. In the future, we hope to synthesize the strengths
of order-constrainedness and our branch-and-bound approach to scale to
these examples that require exploiting this form of structure.

% \begin{wrapfigure}{R}{0.5\textwidth}
%   \includegraphics[width=0.5\textwidth]{figs/ladder_4.pdf}
%   \caption{$k$-shot ladder benchmark with eight nodes.}
%   \label{fig : ladder4}
% \end{wrapfigure}

% \sh{maybe we can transition a bit more delicately here; it feels like we just got a bit defensive
% and made the example harder. Let's motivate why this new problem is cool first.}
Next, we consider the sequential decision-making task of
diagnosing a faulty router, the $k$-shot ladder benchmark.  For a ladder with
eight nodes (four columns), we see that \dappl{} outscales 2AMC, although
neither were able to go past 3-shot ladder within the timeout. This example
was particularly challenging and performance was dependent
on our randomized
strategy for creating rewards
and heuristics for selecting where to branch first; due to this variability, \dappl{} timed out on
3 tries but successfully computed the MEU for 4.

\subsubsection{Gridworld: Scaling on Markov Decision Processes}
\label{sec:gridworld}
{\footnotesize
\begin{table}
\caption{Comparison of finite unrolling of Gridworld MDP benchmarks.
  The grid was an $n\times n$ grid of dimension $n$ with randomly generated start, finish, and obstacles.
  Time is in milliseconds (ms), with timeout 5 minutes = 300000ms.
  All reported times are the average over several runs;
  see \Cref{sec:gridworld} for details.}
\begin{tabular}{l|llll|llll|llll}
\toprule
Grid dim.  & \multicolumn{4}{c}{3} & \multicolumn{4}{c}{4} & \multicolumn{4}{c}{5} \\
\midrule
Horizon     & 1 & 2 & 3 & 4         & 1 & 2 & 3 & 4         & 1 & 2 & 3 & 4 \\
\midrule
$\dappl$    & \textbf{0.30} & \textbf{0.31} & \textbf{0.52} & \textbf{19.36} & \textbf{0.29} & \textbf{1.12} & \textbf{811.38} & \textbf{24511.99} & \textbf{0.30} & \textbf{0.81} & 21012.71 & \tiny{TO} \\
\problog~2  & 2.07 & 6.17 & 839.36 & 3904.34 & 2.95 & 41.51 & 7988.60 & \tiny{TO} & 2.92 & 176.95 & \tiny{TO} & \tiny{TO} \\
2AMC        & 0.49 & 5.20 & 61.96 & 183.10 & 0.88 & 50.36 & 12009.15 & 82836.61 & 1.40 & 41.03 & 24596.71 & \tiny{TO}\\
\bottomrule
\end{tabular}

\label{table : mdp}
\end{table}
}

Next we evaluate \dappl{}'s scalability on a the \emph{grid world} task, a
standard example commonly used to introduce Markov decision processes
(MDPs)~\citep{russell2016artificial}. The grid world task is defined as follows:
\emph{  A robot is in an $n \times n$ grid and starts at location $(0, 0)$. Some grid
  cells are \emph{traps}: if the robot enters these, it receives 0 utility and can
  no longer move, ending the simulation.
  Some grid cells are \emph{obstacles}: the robot cannot pass through these.
  One grid cell is a \emph{goal}: if the robot enters this cell, it receives a
  fixed positive utility. On each time step, the robot picks a direction (up,
  down, left, or right) to make a move. There is some probability that
  this move goes wrong: with probability $p$, the robot will accidentally move in a
  random wrong direction.}

\Cref{table : mdp} shows the results that compare \dappl{}, \problog{}~2, and
2AMC~\citep{derkinderen2020algebraic} on encodings of this example: \dappl{}
significantly outperforms these existing PPL-based approaches.

An alternative approach to solving the grid world example is to explicitly model
the problem as an MDP and solve for the optimal policy using a specialized MDP
solution method such as value iteration or policy
iteration~\citep{sutton2018reinforcement}. These MDP-specific approaches scale
much better than PPL-based approaches on this example: using value iteration,
the optimal policy can be solved on these small-scale MDPs in only a few
iterations, taking microseconds~\citep[Ch.  17]{russell2016artificial}.
However, like all inference strategies, MDP-specific solution methods have
tradeoffs that make them better for some problem instances and worse for others.
Value iteration and policy iteration excel at long-horizon low-dimensional
problems like the grid-world problem. For example, during value iteration, it is
only necessary to keep track of the expected utility of $n \times n$ states for
the grid world; this is quite feasibly represented as a matrix. However, MDPs struggle with
high-dimensional short-horizon decision-making
problems like those encoded by large Bayesian networks~\citep{holtzen2021model}:
in these problems, it is difficult for MDPs to efficiently reason about the
high-dimensional probability distribution on many random variables.
Additionally, the branch-and-bound approach is guaranteed to compute an
\emph{exact} optimal policy, while MDP solution strategies are not guaranteed to
produce the optimal policy unless they are run to a fixpoint, which can take an
unbounded number of iterations.

It is possible to use PPLs that do not support first-class decision-making as
part of an inner-loop for an MDP solving algorithm: this strategy is showcased
by WebPPL, where a probabilistic program computing the expected
utility of a fixed policy can be then used as an inner-loop for policy
evaluation during policy iteration~\citep{agentmodeling,dippl}. Such
specialized MDP-solutions, this approach scales quite well on the grid world
examples, completing in milliseconds.  However, WebPPL struggles to perform
inference on Bayesian networks (see \citet[Fig. 10]{holtzen2020scaling}), and so
this strategy cannot scale to the high-dimensional decision-making problems
considered in \Cref{sec:bn-eval}.

\subsection{Empirical Evaluation of $\pineappl$}
\begin{figure*}[]
\centering
\begin{subfigure}[b]{0.25\linewidth}
 \centering
 \begin{tikzpicture}
  \begin{axis}[
  height=3.5cm,
  width=3.5cm,
  grid=major,
  ymode=log,
  xlabel={\# Solved Cancer},
  legend style={at={(0em,5em)},anchor=west},
  legend columns=-1,
  ylabel={Time (s)},
  xmin=0,
  ]
  \addplot[mark=none, thick, blue] table [x index={0}, y index={1}] {\CancerPNCactus};
  \addlegendentry{\pineappl{}}
  \addplot[mark=none, dashed, thick, red] table [x index={0}, y index={1}] {\CancerPLCactus};
  \addlegendentry{\problog{}}
  \end{axis}
  \end{tikzpicture}
  % \caption{Cancer.}
\end{subfigure}
~~~
\begin{subfigure}[b]{0.25\linewidth}
 \centering
 \begin{tikzpicture}
  \begin{axis}[
  height=3.5cm,
  width=4cm,
  grid=major,
  ymode=log,
  xlabel={\# Solved Sachs},
  ymin=0.001,
  ]
  \addplot[mark=none, thick, blue] table [x index={0}, y index={1}] {\SachsPNCactus};
  \addplot[mark=none, dashed, thick, red] table [x index={0}, y index={1}] {\SachsPLCactus};
  \end{axis}
  \end{tikzpicture}
  % \caption{Sachs.}
\end{subfigure}
~~~
\begin{subfigure}[b]{0.2\linewidth}
 \centering
 \begin{tikzpicture}
  \begin{axis}[
  height=3.5cm,
  width=3.5cm,
  grid=major,
  ymode=log,
  xlabel={\# Solved Insurance},
  ymin=0.001,
  ]
  \addplot[mark=none, thick, blue] table [x index={0}, y index={1}] {\InsurancePNCactus};
  \end{axis}
  \end{tikzpicture}
  % \caption{Insurance.}
\end{subfigure}
~
\begin{subfigure}[b]{0.2\linewidth}
 \centering
 \begin{tikzpicture}
  \begin{axis}[
  height=3.5cm,
  width=3.5cm,
  grid=major,
  ymode=log,
  xlabel={\# Solved Alarm},
  ymin=0.001,
  ]
  \addplot[mark=none, thick, blue] table [x index={0}, y index={1}] {\AlarmPNCactus};
  \end{axis}
  \end{tikzpicture}
  % \caption{Alarm.}
\end{subfigure}
\caption{Cactus plots visualizing the number of solved benchmarks for \problog{} and \pineappl{}.
Plots without \problog{} results indicate that \problog{} failed to complete a single MMAP query.
% Each
% of these benchmarks are computing sets of \map{} queries on the specified model.}
}
\label{fig:cactus}
\end{figure*}
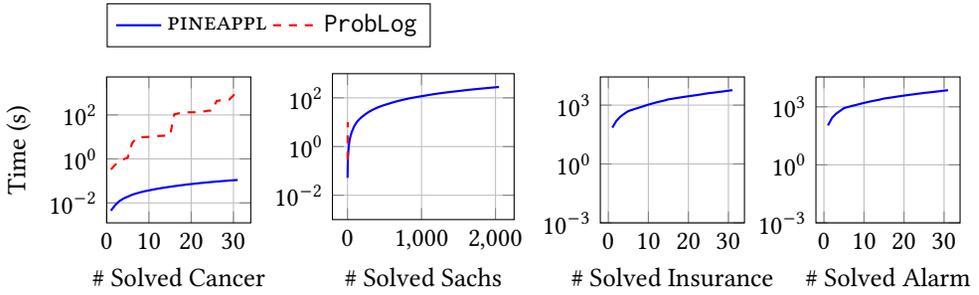

Here we aim to establish that the BBIR
is an effective target for scalably solving MMAP. First we note that, when specialized
to the real semiring, our approach specializes to the
approach in \citet{huang2006solving} for solving MMAP for Bayesian networks:
hence, we focus our evaluation instead on comparing against existing PPL implementations of MMAP and do elide comparing against Bayesian network baselines.
We compared \pineappl{} against
\problog{}, which uses an
enumerative strategy to solve MMAP, much like MEU.
% Indeed, its strategy for MMAP is to use
% existing transformations from MMAP to MEU problems~\citep{maua2016equivalences}
% and solve the MEU problem instead~\cite{van2010dtproblog}.
% By doing this $\problog$ with MMAP inherits the limitations
% $\dtproblog$: This is no longer true, we compared to problog2, which just does
% enumeration
MMAP in \problog{} is not first-class and can only be performed once every
program run, thus there is no possibility for meta-optimization.
There is no standard set of probabilistic programming problems to
benchmark the performance of MMAP, let alone meta-optimization.
Thus, we introduce a simple, illustrative
selection of benchmarks based on discrete Bayesian networks and compiled these
networks into equivalent \pineappl{} and \problog{} programs. The ``Cancer'' and ``Sachs''
networks are small enough to run MMAP queries over the entire powerset of
possible variables. For the ``Alarm'' and ``Insurance'' networks, we selected 5
variables uniformly at random, and ran the powerset of possible queries over
those 5 variables.

Figure~\ref{fig:cactus} gives a cactus plot showing the
relative performance of these two MMAP inference algorithms on four selected
Bayesian networks. To our
knowledge, these are by far the largest probabilistic programs that exact MMAP
inference has been performed on. On two of the examples (Insurance and Alarm),
\problog{} failed to complete a single MMAP query within the time limit,
mirroring the results of~\cref{subsec:dappl-eval}.

\subsubsection{Scalability of Meta-Optimization}
\label{subsubsec:meta-optimization-eval}
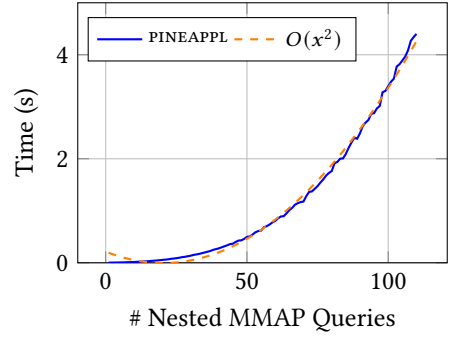
\begin{figure}[t!]
  \begin{subfigure}{0.45\textwidth}
    \begin{pineapplcodeblock}[basicstyle=\tiny\ttfamily]
m = true;
loop n {
  if m {
    x = flip 0.5; y = flip 0.5;
    if x && y { z = flip 0.5; }
    else { z = flip 0.5; }
  } else {
    x = flip 0.5; y = flip 0.5;
    if !x && !y { z = flip 0.5;}
    else { z = flip 0.5; }
  }
  (m) = mmap(z);
}
pr(z)\end{pineapplcodeblock}
    \caption{Template for $\pineappl$ program to demonstrate scaling of MMAP calls.}
    \label{fig:mmap-scaling-program-pineappl}
  \end{subfigure}\hfill
  \begin{subfigure}{0.45\textwidth}
    \begin{tikzpicture}
      \begin{axis}[
        height=5cm,
        ymin=0,
        ymax=5,
        grid=major,
        width=6.5cm,
        xlabel={\# Nested MMAP Queries},
        ylabel={Time (s)},,
        legend columns=6,
        legend style={at={(4.3em,9.2em)},anchor=north, font=\footnotesize}
        ]
        \addplot[mark=none, thick, blue] table [x index={0}, y index={1}] {\Nested};
        \addlegendentry{\pineappl{}};
        \addplot[mark=none, thick, orange, dashed] table [x index={0}, y index={1}] {\Nestedfit};
        \addlegendentry{$O(x^2)$};
      \end{axis}
    \end{tikzpicture}
   \caption{Performance of nested MMAP (blue).
   The plot fits to $O(x^2)$ (orange) with $r^2 = 0.996$.}
   \label{fig:pineappl-nested}
  \end{subfigure}
  \caption{Evaluating the scalability of nested calls to MMAP in $\pineappl$
    programs}
\end{figure}

To demonstrate the utility of staged compilation of BBIR, we construct
$\pineappl$ programs with sequential nested calls to \pineapplcode{mmap}. In
particular, we instantiate line 2 of the program in
\cref{fig:mmap-scaling-program-pineappl} with values ranging from 2 to 140,
corresponding to the number of loop iterations.
Bounded
loops are a hygenic macro in $\pineappl$ that expand to their unrolling with
fresh names (the details of this expansion are described in~\citet{cho2025scaling}).
Since $\pineappl$ programs compile to circuits,
inference performance is not parameter sensitive, hence the use of
\pineapplcode{flip 0.5} for all randomness in the program.

Recall from the motivating example of~\cref{subsec:pineappl-overview} that
evaluating of each call to \pineapplcode{mmap} at the end of compilation
will cause exponential blowup in performance.
This is because we will need to compute and compare marginal probabilities over
Boolean formulae
exponential in the number of variables that we \pineapplcode{mmap} over.
However, BBIR allows for staged compilation, which is reflected in~\cref{fig:pineappl-compl},
which drastically reduces such blowup, as seen in~\cref{fig:pineappl-nested}.
With staged compilation, we face quadratic-time scaling in
the number of calls to \pineapplcode{mmap} despite the exponential blowup in assignment to variables,
as there are only a fixed number of variables defined
before each subsequent \pineapplcode{mmap} call.

%%% Local Variables:
%%% mode: LaTeX
%%% TeX-master: "../oopsla-appendix.tex"
%%% End:

\section{Related Work}\label{sec:related-work}

\subsubsection*{Languages for Optimization and Decision-Making.}~~There have been many
proposed languages for modeling decision-making and optimization from both the
artificial intelligence and programming languages communities.
Influence
diagrams~\citep{khaled2013solving,maua2016equivalences},
planning languages like PDDL and RDDL~\citep{sanner2010relational},
\dtproblog{}~\citep{van2010dtproblog}, and DT-Golog~\citep{boutilierdecision} give
a declarative or graphical description language for describing decision-making
scenarios. The typical approach to performing MEU on in this setting is
order-constrained variable elimination, which has the same worst-case complexity
as order-constrained knowledge compilation.  The problem of solving MEU has been
well-studied on influence diagrams, and branch-and-bound is a common approach in
this setting~\citep{yuan2012solving}; however, we believe our approach here is
the first to leverage knowledge compilation in conjunction with branch-and-bound
for solving MEU.  In the programming languages community, the problem of
designing languages for decision-making has been increasingly of interest and
sparked several recent languages and
systems~\citep{abadi2021smart,lago2022reinforcement}. These listed systems
support more sophisticated language features than \dappl{}, but no
implementation is provided for us to compare performance against.  There are a
number of existing approaches describing programs that model computations over semirings, such as
aProbLog~\citep{kimmig2011algebraic,kimmig2017algebraic} and weighted
programming~\citep{batz2022weighted}; these approaches do not aim to solve
semiring optimization problems such as what we propose here.

\subsubsection*{Knowledge Compilation for Optimization Problems.}
Broadly there are two main approaches within the literature for leveraging
knowledge compilation during optimization: branch-and-bound and
order-constrained approaches. The branch-and-bound approach was
originally proposed by \citet{huang2006solving} for solving the MMAP
problem in Bayesian networks.  Since then the approach has been refined and
improved, but remains the state-of-the-art approach for solving
MMAP on many problem
instances~\citep{choi2022solving,conaty2017approximation}.
\citet{kimmig2011algebraic} introduced AMC, and
\citet{kiesel2022efficient} introduced two-level AMC
to show how to solve MEU by combining the
expected utility and tropical semiring
for computing the MEU of \dtproblog{} programs, generalizing the work of~\citet{derkinderen2020algebraic}.
The primary limitation of two-level AMC is that it requires a fixed
variable order, which can lead to blowup, as we have seen in Section~\ref{sec:eval}.
Seen from this perspective, our branch-and-bound IR can be thought of as a
generalization of the branch-and-bound approach of \citet{huang2006solving}
to work over a much broader class of semirings than just the real semiring,
enabling it to be applied to problems such as MEU.

\subsubsection*{Meta-Reasoning in PPLs.}
Some PPLs today contain some support for forms of \emph{meta-inference}: the
ability to evaluate a marginal query while running a program.  Concretely,
languages with meta-inference typically include an \texttt{infer e} or
\texttt{normalize e} keyword that queries for the probability that a (closed)
program \texttt{e} evaluates to a particular value.
Examples include
\texttt{Church}~\citep{goodman2008church},
Anglican~\citep{tolpin2015probabilistic}, \texttt{Gen}~\cite{cusumano2019gen},
meta-\problog{}~\citep{mantadelis2011nesting},
\texttt{Venture}~\citep{mansinghka2014venture}, and
\texttt{Omega}~\citep{tavares2019random}.
The difference between MMAP and nested inference is that MMAP is finding the
optimal assignment to free variables.

It is possible to use meta-inference
to solve MMAP by enumerating over assignments to free variables,
and selecting the assignment that has the greatest marginal probability.
However, this runs into a clear state-space explosion challenge:
exhaustively enumerating the space of possible assignments during meta-inference
is infeasible for many of the examples we showed in our experiments (for
instance, the examples in our Bayesian network benchmarks query the MMAP state
of over 100 variables in some instances). Hence, for scalability reasons, we
argue that an MMAP query is an invaluable first-class citizen in addition to
meta-inference, and that staging is a useful framework for leveraging
compilation in order to scale. Anglican supports first-class MAP (maximum a posteori) inference, but does not
support MMAP queries~\cite{tolpin2015maximum}.

\subsubsection*{Probabilistic Model Checking and MDPs.} Probabilistic model checkers
such as Storm~\citep{dehnert2017storm} and PRISM~\citep{kwiatkowska2002prism} give
a specification language and query language for describing, solving, and
verifying Markov decision processes, and hence are capable of solving MEU
problems. These languages can scale quite well, and are especially useful for
verifying complex temporal queries. However, these systems require describing
the probabilistic system as an MDP, which can be very expensive; as shown in~\citet{holtzen2021model}, MDP-based
representations can scale poorly when compared with approaches that leverage
factorization on problem instances that exhibit independence structure.
Concretely, the Bayesian network examples given in our experiments would pose
significant scaling challenges to these systems, especially the large
hidden-Markov-model in Figure~\ref{fig : dr}. Additionally, MDPs do not
support first-class conditioning on evidence, which \dappl{} and many
other probabilistic programming languages support.

% \subsubsection{Weighted Programming}

\section{Conclusion and Future Work}

We presented the BBIR, a new intermediate representation for optimization problems
over discrete probabilistic inference. The BBIR can represent important optimization
problems such as MEU and MMAP with additional features such as staged compilation,
conditioning, and reasoning beyond probabilities. The flexibility of the BBIR
was showcased through two very different programming languages: $\dappl$,
a function decision-theoretic PPL with support for Bayesian conditioning,
and $\pineappl$, an imperative PPL with first-class meta-optimization support
via MMAP.

Our efforts in this paper focused on designing a new scalable intermediate
representation to support a broad class of optimization problem; hence, we
simplified the design of our surface-level languages to simplify this
compilation.  In the future we aim to provide more expressive surface-level
languages that compile to BBIR. The most tractable would be to add support
for language features like top-level functions and dynamically-bounded
surely-terminating loops; languages like \dice{} and \texttt{ProbLog} support
these features.
Next, it would also be interesting to explore adding features to support applications in game theory,
such as multiple decision-making agents or stochastic policies,
as our current framework is limited to one decision-making agent and deterministic policies.
Finally, it would be interesting to explore
the extent to which we can provide more ergonomic and unified surface languages
for efficiently programming with decision-making, for instance by developing
efficient implementations of selection monad~\citep{abadi2021smart,lago2022reinforcement}.

% The first is to extend $\dappl$ and $\pineappl$ to not complement each other,
% but have the same features: that is, allow meta-optimization in $\dappl$
% and allow conditioning via $\dapplcode{observe}$s in $\pineappl$.
% Another direction along this vein is to combine
% multiple optimization problems
% into one language that compiles into one BBIR.

% We presented \dappl{}, a functional decision-theoretic probabilistic programming language
% with support for Bayesian reasoning. The core contribution of \dappl{} is its new
% approach to scalable expected utility computation based on knowledge compilation.
% As future work, we see several future directions. One direction, as hinted in Section~\ref{sec:eval},
% is a hybrid MEU solving strategy via a synthesis of branch-and-bound and
% and order-constrained approaches over BBIR. Another direction would be to adapt the language
% to support MEU with multiple decision makers, which would allow us to
% model game-theoretic scenarios~\citep{gan2022bayesian,osborne2004introduction},
% or to support addition problems such as MMAP along with MEU.
% Lastly, we would like to see a compilation scheme to BBIR that is language-agnostic,
% allowing for easy development of probabilistic programming languages that support
% the optimization-via-compilation scheme.

\section*{Data Availability Statement}
The software that supports~\cref{sec:eval} is available on Zenodo~\citep{artifact}.

\section*{Acknowledgments}
We thank the anonymous reviewers for their helpful guidance.
This project was supported by the National Science Foundation under grant \#2220408.

\bibliography{meu.bib}

% Appendix
\appendix

\newpage
\section{Supplementary Material for Section \ref{sec:overview}}

\subsection{The $\AMC{}$ invariant}\label{appendix:amc invariant}

We can define the expected utility of a Boolean formula as an expectation:

\begin{definition}[Expected Utility of a Boolean Formula]\label{def:eu-of-formula}
  Let $\varphi$ be a Boolean formula that consists of reward variables
  $\mathcal{R} = \{R_i\}$ and probabilistic variables $\mathcal{P} = \{P_j\}$; 
  we will denote literals -- assignments to Boolean variables -- using lower-case letters.
  Assume we are given a \emph{utility map} $U$ that maps
  reward literals to a real-valued reward, and \emph{probability map} $\Pr$ that maps
  probabilistic literals to probabilities.
  Then, the \emph{probability of model} $\{r_i,p_j\}$
  is the product of probabilities of each probabilistic variable:
  % Let $\Omega$ be the set of all assignments to variables in $\varphi$. The 
  % distribution on $\Omega$ 
  % induced by $\Pr$ is defined by, for each $m \in \Omega$:
  % \begin{align}
    $\Pr(\{r_i, p_j\}) \triangleq \prod_{j} \Pr(p_j).$
  % \end{align}
  The expected utility of $\varphi$ is the probability-weighted sum of
  utilities that satisfy the formula:
  \begin{align}
    \EU[\varphi] \triangleq \sum_{\{r_i, p_j\} \models \varphi} \Pr(\{r_i, p_j\}) \paren{\sum_{i} U(r_i)}.
  \end{align}
\end{definition}

The relation is as follows:

\begin{theorem}\label{thm:amc invariant}
  Let $\varphi$ be a Boolean formula consisting of probabilistic variables $P_i$
  and reward variables $R_i$, with utility map $U$ and probability map $\Pr$.
  Let $w: \lits(\varphi) \rightarrow \mathcal{S}$ be a weight function that maps
  literals to elements of the expectation semiring, where for probabilistic 
  variables we assign $w(P_i) = (\Pr(P_i), 0), w(\overline{P_i})$ and
  $w(\overline{P_i}) = (\Pr(\overline{P_i}), 0)$, and for reward variables we
  assign $w(R_i) = (1, U(R_i))$ and $w(\overline{R_i}) = (1, 0)$. Then
  $\AMC(\varphi, w)_{\EU} =  \EU[\varphi]$.
\end{theorem}

The proof relies on the following lemmata, 
whose proofs are straightforward inductions:

\begin{lemma}\label{lemma:distrib.eu}
  Let $\{(p_i, v_i)\}\subseteq \mathcal S$. Then
  \begin{equation*}
    \left[\bigotimes_i (p_i, v_i)\right]_{\EU} = \sum_{i} v_i\paren{\prod_{j \neq i} p_j}.
  \end{equation*} 
\end{lemma}

\begin{lemma}\label{lemma:pull.eu}
  Let $\{(p_i, u_i)\}\subseteq \mathcal S$. Then
  \begin{equation*}
    \left[\bigoplus_i (p_i, u_i)\right]_{\EU} = \sum_{i} (p_i, u_i)_{\EU}.
  \end{equation*} 
\end{lemma}

Unfolding, we see that
  \begin{align*}
    \EU[(\varphi, w)]
      &=\sum_{m \models \varphi} \prod_{\ell \in m} w(\ell)_{\Pr} \paren{\sum_{\ell \in m} \frac{w(\ell)_{\EU}}{w(\ell)_{\Pr}}} = \sum_{m \models \varphi}  \sum_{\ell \in m}\paren{ \frac{w(\ell)_{\EU}}{w(\ell)_{\Pr}}\prod_{j \in m} w(j)_{\Pr}}\\
      &=\sum_{m \models \varphi} \sum_{\ell \in m} \paren{ w(\ell)_{\EU}\paren{\prod_{j \neq \ell} w(j)_{\Pr}}}\\
      &=\sum_{m \models \varphi} \left[\bigotimes_{\ell \in m} w(\ell)\right]_{\EU}&(\star)\\&=\left[\bigoplus_{m \models \varphi} \bigotimes_{\ell \in m} w(\ell)\right]_{\EU} = [\AMC_{\mathcal S}(\varphi, w)]_{\EU}&(\dagger)
  \end{align*}
  where $(\star)$ is the usage of Lemma \ref{lemma:distrib.eu} and $(\dagger)$ denotes usage of Lemma \ref{lemma:pull.eu}.
\subsection{Join-Sum is lower bounded by Sum-join}
\label{appendix:commute join}

\begin{lemma}\label{lemma:commute join}
  Let $\mathcal R$ be a lattice semiring with partial order $\sqsubseteq$.
  Let $f : X \times Y \to \mathcal R$ be a function with codomain 
  $\mathcal R$ with $X, Y$ finite sets. Then
  \begin{equation}
    \sqcup_{x \in X} \sum_{y \in Y} f(x,y) \sqsubseteq 
    \sum_{y \in Y} \sqcup_{x \in X} f(x,y).
  \end{equation}
\end{lemma}

\begin{proof}
  It suffices to show that for all $x \in X$, $\sum_{y \in Y} f(x,y) \sqsubseteq
  \sum_{y \in Y} \sqcup_{x \in X} f(x,y).$ It suffices to show that for all
  $y \in Y$, $f(x,y) \sqsubseteq
  \sum_{y \in Y} \sqcup_{x \in X} f(x,y).$ This follows from the definition
  of join being a least comparable upper bound, and we are done.
\end{proof}

\newpage
\section{Supplementary material for Section \ref{sec:bbir}}
\subsection{Proof of Theorem~\ref{thm:ub correctness}}
\label{appendix:proof ub correctness}

\begin{proof}
  We induct on $|X| - |P|$. In the base case, $|X| = |P|$, so $P$ is
  a policy and we are done. For our inductive argument, consider
  $P$ a partial policy such that for all $P' \supsetneq P$,
  Equation~\ref{eq:ub correctness} holds.
  We want to show that $P$ still satisfies Equation~\ref{eq:ub correctness}
  for all of its completions.

  So fix $T$ a completion. Simplifying Equation~\ref{eq:ub correctness}
  we observe it suffices to show
  \begin{equation}
    h(\varphi|_P, w) \sqsupseteq h(\varphi|_T,w)\prod_{x \in P} w(x).
  \end{equation}

  Let $x$ be the variable chosen first by line 4 of Algorithm~\ref{algorithm:h}
  when computing $h(\varphi|_P, w)$.
  We know $x$ must be either a choice (i.e., lie in $T \setminus P$) or not; we case.
  \begin{itemize}
    \item If $x \in T \setminus P$, 
    then in particular $x \in T$; 
    we take the join as per line 5. We observe
    \begin{align*}
      h(\varphi|_P, w)
        &= w(x)h(\varphi|_{P}|_x) 
          \sqcup w(\overline x)h(\varphi|_{P}|_{\overline x},w) \\ 
        &\sqsupseteq w(x)h(\varphi|_T,w)\prod_{y \in P\setminus \{x\}}w(y) 
          \sqcup w(\overline x)h(\varphi|_{P}|_{\overline x},w)&\text{(IH)}\\  
        &\sqsupseteq w(x)h(\varphi|_T,w)\prod_{y \in P\setminus \{x\}}w(x) 
        = \boxed{h(\varphi|_T,w)\prod_{y \in P}w(x),}    
    \end{align*}
    where (IH) can be used as $\varphi|_{P}|_x = \varphi|_{P \cup \{x\}}$,
    and $P \cup \{x\}$ is still a partial policy. The case for 
    $\varphi|_P|_{\overline x}$ is identical.
    \item If $x \notin T\setminus P$, 
    we take the sum as per line 6. 
    We continue recursing until we hit a variable $x' \in T \setminus P$; 
    then we reduce to case 1 and we are done.
  \end{itemize}
\end{proof}

\noindent
\textit{Remark.} 
It is important to note that the above theorem 
cannot be generalized to give a relation between 
any two partial policies $P \subsetneq P'$. 
This is because that the first inductive case 
crucially relies on the fact that 
we are not applying the IH twice. 
Indeed, in general,
\begin{align*}
   w(x)h(\varphi|_T,w)\prod_{y \in P\setminus \{x\}}w(y) 
    \sqcup w(\overline x)h(\varphi|_T,w)\prod_{y \in P\setminus \{\overline x\}} w(y) 
  \\\quad 
   \not\sqsupseteq 
  \prod_{y \in P } w(y) \paren{w(x)h(\varphi|_T,w) \sqcup w(\overline x) h(\varphi|_T,w)}
\end{align*}
when $x \notin P'$; this is a manifestation of the more general phenomena that
$$a(b \sqcup c) \not \sqsubseteq ab \sqcup ac.$$
\subsection{$\mathsf{UB}_f$ for MEU and MMAP}
\label{appendix:ub_f}

$\mathsf{UB}_f$ for MEU is given in Algorithm~\ref{algorithm:ub_f meu}.
For notational simplicity, instead of using the BBIR 
$(\{\varphi \land \gamma_{\pi} : \pi \in \mathcal A\}, A, w)$, 
we will use
the tuple $(\{\varphi, \gamma\}, A, w)$ 
in which $\gamma$ is the formula in which
for all $\pi \in \mathcal A$, $\gamma|_{\pi} = \gamma_{\pi}$. 

We define 
scalar division for $\mathcal S$:
\begin{equation}
  \frac{(a,b)}{r} = \begin{cases}
    \paren{\frac a r , \frac b r} & r \neq 0, \\
    (0, -\infty) & r = 0.
  \end{cases}
\end{equation}
We note that, if utilities are all nonnegative, 
then Lines 3-4 are not needed, and we can instead let the returned value in
Line 5 be $t/\ell$; this follows from 
eliminating the casework done to 
prove Theorem~\ref{appendix:reduction soundness}.
Indeed, in the implementation, this is what happens.

\begin{figure}
  \begin{mdframed}{\footnotesize\begin{algorithmic}[1]
    \Procedure{$bb$}{$(\{\varphi, \gamma\}, A, w), P_{curr}, \ell$}
    \State $t = ub((\{\varphi|_{\ell}, \gamma|_{\ell}\}, A, w),
        (\varphi \land \gamma)|_{\ell} , P_{curr} \cup \{\ell\})$
        \State $x = lb((\{\varphi|_{\ell}, \gamma|_{\ell}\}, A, w),
        \gamma|_{\ell} , P_{curr} \cup \{\ell\})$
        \State $y = ub((\{\varphi|_{\ell}, \gamma|_{\ell}\}, A, w),
        \gamma|_{\ell} , P_{curr} \cup \{\ell\})$
        \State \textbf{return } $\sqcup (t/x_{\Pr}, t/y_{\Pr})$
  \EndProcedure
  \end{algorithmic}}\end{mdframed}
  \caption{$\mathsf{UB_f}$ for MEU.}
  \label{algorithm:ub_f meu}
\end{figure}

The $\mathsf{UB}_f$ for MMAP is omitted as it is known~\citep{huang2006solving} that

\begin{align*}
  \sum_{v \in inst(V)}
  \Pr[\{m \cup v \cup e \models \varphi\} | \{e \models \gamma\}]
  = \frac{\AMC_{\R} (\varphi |_m \land \gamma|_m)}{ \AMC_{\R}(\gamma|_m)}
\end{align*}

where $\AMC_{\R}$ is the algebraic model count taken over the reals.
\subsection{Proof of Theorem~\ref{thm:soundness of bb}}
\label{appendix:soundness of bb proof}

For MEU, we first prove a Lemma.

\begin{lemma}\label{lemma:soundness helper}
  Let $(\{\varphi, \gamma\}, A, w)$ be a BBIR for MEU,
  and let $P$ be a partial policy over $A$. Then 
  let $(a,b) = ub((\{\varphi, \gamma\}, A, w), \varphi \land \gamma, P)$,
  $\ell = lb((\{\varphi, \gamma\}, A, w), \gamma, P)_{\Pr}$, and
  $u = ub((\{\varphi, \gamma\}, A, w), \gamma, P)_{\Pr}$.
  Then if $\ell, u \neq 0$, for all total extensions $T \supsetneq P$ 
  we have
  \begin{equation}
    \frac{\AMC(\varphi\land \gamma|_T, w)}{\AMC(\gamma|_T, w)_{\Pr}}
    \sqsubseteq \paren{\frac a \ell, \frac b \ell}\sqcup \paren{\frac a u, \frac b u}.
  \end{equation}
\end{lemma}

\begin{proof}
  The proof follows from applications of
  Theorem~\ref{thm:ub correctness} and its dual for $lb$.
  Let $(u,v) = \AMC(\varphi\land \gamma|_T, w)$ and $k=\AMC(\gamma|_T, w)_{\Pr}$.
  We have that 
  \begin{equation}\label{eq:lemma soundness eq 1}
    (u,v) = \AMC((\varphi\land \gamma)|_T, w) \sqsubseteq
    ub((\{\varphi, \gamma\}, A, w), \varphi \land \gamma, P)
  \end{equation}
  and 
  \begin{equation}\label{eq:lemma soundness eq 2}
    lb((\{\varphi, \gamma\}, A, w), \gamma, P)
    \sqsubseteq (k, \_) = \AMC(\gamma|_T, w)
    \sqsubseteq ub((\{\varphi, \gamma\}, A, w), \gamma, P).
  \end{equation}
  From Equation~\ref{eq:lemma soundness eq 2}
  we know that $\ell \leq \AMC(\gamma|_T, w)_{\Pr} \leq u$.
  We know that $\ell$ and $u$ are within $(0,1]$ as they are 
  computing probabilities~\citep{darwiche2002knowledge}.

  We want to show that $u/k \leq \max(a/\ell, a/u)$ and 
  $v/k \leq \max(b/\ell, b/u)$. The former follows from
  Equation~\ref{eq:lemma soundness eq 2}. The latter requires casework on $b$:

  \begin{itemize}
    \item If $b$ is nonnegative, then $v/k \leq b/\ell$ and we are done,
    \item if $b$ is negative, then $k \leq u$. Then  $1/k \geq 1/u$
    and thus $b/k \leq b/u$. Then since $v \leq b$ by Equation~\ref{eq:lemma soundness eq 1}
    we have $v/k \leq b/u$ as desired.
  \end{itemize}

\end{proof}

\begin{proof}[Proof of Theorem~\ref{thm:soundness of bb}]
  It suffices to prove that the optimal policy is never pruned. 
  That is, let $T_{\MEU}$ be the total model witnessing 
  $\mathrm{MEU}((\{\varphi, \gamma\}, A, w))$. 
  We claim that $T_{\MEU} = P_{curr}$ at some recursive call of $bb$. 

  Suppose $T_{\MEU}$ is pruned on Line~\ref{line:prune}. Then, 
  there exists a partial policy $P' \subset T_{MSP}$ 
  such that 
  the join $m$ as calculated on Line~\ref{line:join} has
  $m \sqsubseteq b$ for some $b$. 

  By Lemma~\ref{lemma:ub on total policy is amc},
  $b = \AMC(\varphi|_T,w)$ for some total model $T$. 
  Then we have that:
  \begin{align*}
    \AMC( \varphi|_{T_{MSP}},w) 
      &\sqsubseteq m & \text{(by Lemma~\ref{lemma:soundness helper})} \\ 
      & \sqsubseteq \AMC( \varphi|_T , w) & \text{by assumption}.
  \end{align*}
  By compatibility we have that
  \begin{equation*}
    \AMC(\varphi|_{T_{MSP}}) \leq \AMC( \varphi|_T,w).
  \end{equation*}
  If the two sides are equal that means that 
  $MSP(\varphi)$ has multiple witnesses, 
  thus the branch was never pruned. 
  Otherwise, $b > MSP(\varphi)$, which is a contradiction.
\end{proof}

For MMAP, we defer the proof of correctness to~\citet{huang2006solving}.

\newpage
\section{Supplementary material for Section \ref{sec:dappl}}

\subsection{A \dappl{} typesystem}
\label{appendix:typesystem}
The type system of \dappl{} has the types
$\Bool$, $\Giry{\Bool}$, and $\Choice{\alpha_1, \cdots, \alpha_n}$,
for nonempty sets of nonconflicting names
$\{\alpha_1, \cdots, \alpha_n\}$.
The typing is with respect to
a standard context $\Gamma := \cdot \ \| \ x: \tau, \Gamma$.

\begin{figure}[H]
\begin{mdframed}
{\centering
  \begin{align*}
    \infer[\texttt{tp/T}]
      {\cdot \proves \tt : \Bool}
      {}
    \qquad
    \infer[\texttt{tp/F}]
      {\cdot \proves \ff : \Bool}
      {}
    \qquad
    \infer[\texttt{tp/var}]
      {\Gamma \proves x : \tau}
      {x : \tau  \in \Gamma}
  \end{align*}
  \begin{align*}
    \infer[\texttt{tp/and}]
      {\Gamma \proves P_1  \land P_2 : \Bool}
      {\Gamma \proves P_1  : \Bool
      &
      \Gamma \proves P_2  : \Bool}
    \qquad
    \infer[\texttt{tp/or}]
      {\Gamma \proves P_1  \lor P_2 : \Bool}
      {\Gamma \proves P_1  : \Bool
      &
      \Gamma \proves P_2  : \Bool}
  \end{align*}
  \begin{align*}
    \infer[\texttt{tp/neg}]
      {\Gamma \proves \neg P: \Bool}
      {\Gamma \proves P : \Bool}
    \qquad
    \infer[\texttt{tp/ret}]
      {\Gamma \proves \return P : \Giry \Bool}
      {\Gamma \proves P : \Bool}
  \end{align*}
    \begin{align*}
      \infer[\texttt{tp/flip}]
        {\cdot \proves \flip \theta : \Giry \Bool}
        {}
      \qquad
      \infer[\texttt{tp/reward}]
        {\Gamma \proves \reward k e: \tau}
        {\Gamma \proves e : \tau}
    \end{align*}
    \begin{align*}
      \infer[\texttt{tp/ITE}]
        {\Gamma \proves \ite e {e_T} {e_E} : \tau}
        {\Gamma \proves e : \Bool
        & \Gamma \proves e_T : \tau
        & \Gamma \proves e_E : \tau}
    \end{align*}
    \begin{align*}
      \infer[\texttt{tp/<-/G}]
        {\Gamma \proves \bind x e {e'} : \tau}
        {\Gamma \proves e : \Giry \Bool
        & x : \Bool, \Gamma \proves e' : \tau}
    \end{align*}
    \begin{align*}
      \infer[\texttt{tp/<-/Choice}]
        {\Gamma \proves \bind x e {e'} : \tau}
        {\Gamma \proves e : \Choice{\alpha_1, \cdots, \alpha_n}
        & x : \Choice{\alpha_1, \cdots, \alpha_n}, \Gamma \proves e' : \tau}
    \end{align*}
    \begin{align*}
      \infer[\texttt{tp/obsT}]
        {\Gamma \proves \observe P e : \tau}
        {\Gamma \proves P : \Bool
        & \Gamma \proves e : \tau}
    \end{align*}
    \begin{align*}
    \infer[\texttt{tp/[]}]
      {\cdot \proves [\alpha_1, \cdots, \alpha_n] :
        \Choice{\alpha_1, \cdots, \alpha_n}}
      {}
    \end{align*}
    \begin{align*}
    \infer[\texttt{tp/choosewith}]
      {\Gamma \proves \choose{x}{\alpha_i \implies e_i} : \tau}
      {\Gamma \proves x : \Choice{\alpha_1,\cdots, \alpha_n}
      &\forall i \in [n]. \ \Gamma \proves e_i : \tau}
    \end{align*}}
\end{mdframed}
\caption{Typing rules of \dappl{}.
The typing rules of \util{} are all of the above
rules except for \texttt{tp/[]}, \texttt{tp/choosewith}, and \texttt{tp/<-Choice}.
}
\label{fig:dappl typing}

We prove a Lemma:

\begin{lemma}
  Let $\Gamma \proves e : \tau$ a $\util$ expression. Then $\tau$ must be of type $\Giry \Bool$.
\end{lemma}

\begin{proof}
  Induction on the typing rules.
\end{proof}
\end{figure}

\subsection{Denotational semantics of $\util$}\label{appendix:util semantics}

We specify a Lemma:

\begin{lemma}\label{lemma:util context}
  Let $\Gamma \proves e : \tau$ be a $\util$ expression via the typing rules 
  of~\cref{fig:dappl typing}. Then $\Gamma$ can only be a list of variables of
  type $\Bool$.
\end{lemma}

\begin{proof}
  Proof is by induction on the typing rules of util.
\end{proof}

We define the a distribution $\mathcal D((\Bool \times \R) \cup \{\bot\})$
as a function $(\Bool \times \R) + \{\bot\} \to \R$, although we use the notation
$\{v_1 \mapsto p_1, \cdots, v_n \mapsto p_n\}$ 
for explicit values $v_i \in (\Bool \times \R) + \{\bot\}$ mapping to probabilities $p_i$
when it is more convenient, with 
the implicit assumption that any value not present has probability zero.

We use the shorthands $\mathbf{TT} = \{(\tt,0) \mapsto 1\}$, $\mathbf{FF} = \{(\ff,0) \mapsto 1\}$, and
$\pmb{\bot} = \{\bot \mapsto 1\}$.

By~\cref{lemma:util context}, we can say that the denotation for $\Gamma$, $\denote{\Gamma}$, 
are maps from free variables of $e$ to either $\mathbf{TT}$ or $\mathbf{FF}$.
Thus expressions $\Gamma \proves e : \Giry \Bool$ can be denoted as
functions $\denote{e} : \denote{\Gamma} \to \mathcal D((\Bool \times \R) \cup \{\bot\})$

The symbol \fishbone is the monadic bind operation for probability distributions with
finite support. The interpretation of logical operations over pure expressions are 
defined to be the operation lifted to probability distributions.

\begin{align*}
  \denote{x} &= \lambda g . \ g(x)  \\
  \denote{\tt} &= \lambda g . \ \mathbf{TT} \\
  \denote{\ff} &= \lambda g . \ \mathbf{FF} \\
  \denote{\return e} &= \lambda g . \denote{e} g\ \\
  \denote{\flip \theta} &= \lambda g . \ \{(\tt,0) \mapsto \theta , (\ff,0) \mapsto (1 - \theta)\} \\
  \denote{\reward k e} &= \lambda g . \ 
    \lambda v. \ 
    \begin{cases}
      \denote{e}(g)(b, r-k) & v = (b, r) \\
      \denote{e}(g)(v)  & \text{else}
    \end{cases}  \\
  \denote{\ite x {e_1} {e_2}} 
    &= \lambda g . \ \begin{cases}
      \denote{e_1}g & g(x) = \mathbf{TT} \\
      \denote{e_1}g & g(x) = \mathbf{FF} \\
      {\color{red}\texttt{abort}} & \text{else}
    \end{cases} \\
  \denote{\observe x e} 
    &= \lambda g . \ \begin{cases}
      \denote{e}g & g(x) = \mathbf{TT} \\
      \pmb{\bot} & g(x) = \mathbf{FF} \\
      {\color{red}\texttt{abort}} & \text{else}
    \end{cases} \\
  \denote{\bind x e {e'}}
    &= \lambda g . \ \denote{e} g  \ \fishbone \
      \lambda x. \ \begin{cases}
        \lambda y . \ 
          \begin{cases}
            \denote{e'} (g \cup \{x \mapsto \mathbf{TT}\}) (b,(s-r)) & y = (b,s)\\
            \denote{e'} (g \cup \{x \mapsto \mathbf{TT}\}) y     & \text{else}
          \end{cases} & x = (\tt, r) \\
          \lambda y . \ 
          \begin{cases}
            \denote{e'} (g \cup \{x \mapsto \mathbf{FF}\}) (b,(s-r)) & y = (b,s)\\
            \denote{e'} (g \cup \{x \mapsto \mathbf{FFS}\}) y     & \text{else}
          \end{cases} & x = (\ff, r) \\
        \pmb{\bot} & x = \bot 
      \end{cases}
\end{align*}
\subsection{Soundness of reduction from \dappl{} to \util{}}
\label{appendix:reduction soundness}

The transformation of \dappl{} to \util{} programs are given as
equational rules. To set up, let $\Gamma \proves e : \tau$ a \dappl{} 
expression. Let $\mathcal A$ be the policy space of $e$. Additionally we can 
consider policies on the context $\Gamma$, which we precisely define below.

\begin{definition}\label{def:context policy space}
  Let $\Gamma \proves e : \tau$ a \dappl{} expression. 
  For any $x_i : \Choice{\alpha_i} \in \Gamma$, 
  call $\{\alpha_i\}$ the \emph{choice} of $x_i$. 
  The product of all choices of $x \in \Gamma$ is called 
  the \emph{context policy space}, written $\mathcal{A}_{\Gamma}$.
\end{definition}

This will prove useful when we are attempting to reduce expression of form
$\choose x {\alpha_i \implies e_i}$ to \util{}.
Let $\pi \in \mathcal A$ and let $\pi_{\Gamma} \in \mathcal{A}_{\Gamma}$. 
Then we can consider the joint policy $\pi \cup \pi_{\Gamma}$ on which to
reduce $e$ with.
Selected rules are given in Figure~\ref{fig:dappl to util}.
Omitted rules follow the standard recursive application of $|_{\pi \cup \pi_{\Gamma}}$
to subexpressions.

\begin{figure}[H]
  \begin{mdframed}
    \begin{align*}
      [\alpha_1, \cdots, \alpha_n]|_{\pi \cup \pi_{\Gamma}} = \return \tt
    \end{align*}
    \begin{align*}
      {(\choose x {\alpha_i \implies e_i})|_{\pi \cup \pi_{\Gamma}} = e_i|_{\pi \cup \pi_{\Gamma}}}
      \text{ for $i$ s.t. $\alpha_i = \mathrm{proj}_k \pi \cup \pi_{\Gamma}$ for some $k$}
    \end{align*}
  \end{mdframed}
  \caption{Selected reduction rules from \dappl~to \util.}
  \label{fig:dappl to util}
\end{figure}

At this point it is important to prove 
our reduction sound, which we will do so.

\begin{lemma}\label{lemma:dappl to util}
  Let $\Gamma \proves e : \tau$ be a well-typed \dappl~expression, 
  and let $\pi \in \mathcal A$ and $\pi_{\Gamma} \in \mathcal{A}_{\Gamma}.$ 
  Then $e|_{\pi \cup \pi_{\Gamma}}$ is a well-typed \util~expression, and 
  in particular it is well-typed with respect to the context $\Gamma$ with 
  all instances of variables with type $\Choice S$ for some $S$ removed.
\end{lemma}

\begin{proof}
  We do this by induction on the structure of the expression $e$. 
  Most cases are omitted as they are straightforward; we show the cases for the 
  rules of~\cref{fig:dappl to util}.
  \begin{itemize}
    \item If $e = [\alpha_1, \cdots, \alpha_n]$, clearly $\return \tt$ is a valid
    \util{} expression. 
    Furthermore it is well-typed in any context, concluding the case.
    \item If $e = \choose x {\alpha_i \implies e_i}$, then we first need to prove 
    that there exists $i$ such that 
    $\alpha_i = \mathrm{proj}_k \pi$ for some $k$.
    We show that the set of names $\{\alpha_i, \cdots, \alpha_n\}$ 
    that we are matching on is a factor of the policy space 
    $\mathcal A$, as the result follows since $\mathcal A$ is a 
    finite product and $\{\alpha_i, \cdots, \alpha_n\}$ is a finite set.
    Indeed, this is enforced by the type of the subexpression $x$
    as seen in Figure~\ref{fig:dappl typing}. At this point, the 
    IH shows that $e_i|_{\pi \cup \pi_{\Gamma}}$ can be well-typed 
    with respect to the context $\Gamma$ with 
    all instances of variables with type $\Choice S$ for some $S$ removed,
    concluding the proof.
  \end{itemize}

\end{proof}

\subsection{\dappl~Ergonomics and Syntactic Sugar}\label{appendix:sugar}
We extend the \dappl~core calculus with several ergonomic features that makes
the modeling of decision scenarios easier.

\subsubsection{Ending an expression with \dapplcode{reward}.}

Instead of $\reward k {\return \tt}$ one can write $\dapplcode{reward k}$.

\subsubsection{Support for discrete distributions.}\label{sugar:discrete}

We give \dappl~support for explicit categorical distributions over
a set of variables. For example, the expression \texttt{disc[a : 0.5, b: 0.3, c: 0.2]}
defines a probability distribution over the set of names \texttt{\{a,b,c\}} in which \texttt{a}
has probability 0.5, \texttt{b} has probability 0.3, and \texttt c has probability 0.2.
Discrete distributions de-sugar into a style of \textit{one-hot encoding}, in which a 
distribution over $n$ categorical variables are represented $n$ Boolean variables~\citep{holtzen2020scaling}.

\subsubsection{Overloading of if-then-else and choose-with.}\label{sugar:overloading}

We allow the guard of an if-then-else statement to be a decision with one choice.
Intuitively this would represent the decision of choosing to do something or not.
Symmetrically, we allow use of the choose-with statement over categorical distributions
as outlined above in \ref{sugar:discrete}.
We can do this as for a decision with one choice $c$, 
the expression $\text{ExactlyOne}(c)=c$, and analogously, we can check that for
a categorical distribution \texttt{disc[x1 : p1 , ... , xn : pn]}, the one-hot encoding
will enforce the exactly-one constraint.

\subsubsection{Bounded loops.}\label{sugar:loops}

We allow bounded loops; that are, loops that terminate after a specified number of times.
This avoids the potential of infinite computation while maintaining exactness, 
which has been implemented in several existing PPLs. 
The syntax is $\texttt{loop } n \texttt{ \{} e \texttt{\}}$, on which an expression $e$
is run $n$ many times.
In the case that a decision is within the loop, as an optimization we can pull the decision out of the loop, at which point the expected utility becomes $n$ times that of $e$. This is proved sound in the following Lemma.

\begin{lemma}[Soundness of loops]
  For a $\dappl{}$ program $e$,
  \begin{equation}
  \infer
  {\EU(\texttt{loop } n \texttt{ \{} e \texttt{\}}) = nk}
  {\EU(e) = k & n >0}
\end{equation}
\end{lemma}

\begin{proof}
We prove this by induction.
As a base case we have that 
$\texttt{loop } n \texttt{ \{} e \texttt{\}} = e$,
and $\EU(e) =k $, so $\texttt{loop } n \texttt{ \{} e \texttt{\}} \Downarrow_{\EU} k$ 
as desired. 

In our inductive case, we observe that
$\texttt{loop } n \texttt{ \{} e \texttt{\}} = 
\bind x {\texttt{loop } (n-1) \texttt{ \{} e \texttt{\}}} e$. We see that $x$ will not occur free in $e$, 
or vice versa; thus we can add utilities via our semantics to get $\EU(\texttt{loop } n \texttt{ \{} e \texttt{\}}) = (n-1)k +k = nk$ as desired.  
\end{proof}
\subsection{Full Boolean compilation rules of \dappl{}}
\label{appendix:dappl bc}

See~\cref{fig:dappl full bc}.

\begin{figure}
  \begin{mdframed}
  {\footnotesize 
  \begin{align*}
    \infer[\texttt{bc/var}]
    {x \leadsto (x, T, \eset, \eset)}
    {}
  \end{align*}
  \begin{align*}
    \infer[\texttt{bc/true}]
      {\tt \leadsto (T, T,\eset, \eset)}{}
    \qquad
    \infer[\texttt{bc/false}]
      {\ff \leadsto (F,T, \eset, \eset)}{}
  \end{align*}
  \begin{align*}
    \infer[\texttt{bc/flip}]
      {\flip{\theta} \leadsto 
        (f_{\theta}, 
        T, 
        (f_{\theta} \mapsto (\theta,0), \overline{f_{\theta}} \mapsto (1-\theta, 0)),
        \eset)}
      {\text{fresh } f_{\theta}}
  \end{align*}
  \begin{align*}
    \infer[\texttt{bc/ret}]
      {\return P \leadsto 
        (\varphi, T, \eset, \eset)}
      {P \leadsto (\varphi, T, \eset, \eset)}
  \end{align*}
  \begin{align*}
    \infer[\texttt{bc/reward}]
    {
      \reward k  e \leadsto 
      (\varphi, \gamma, R \cup \{r_k\}, 
        w \cup \{r_k \mapsto (1,k), \overline{r_k} \mapsto (1,0)\})}
    {
      \text{fresh } r_k
      & e \leadsto (\varphi, \gamma, R,w)
    }
  \end{align*}
  \begin{align*}
    \infer[\texttt{bc/obs}]
    {\observe x e \leadsto (\varphi, \gamma \land x, w, R)}
    {x \leadsto (x, T, \eset, \eset)
    & e \leadsto (\varphi, \gamma, w, R)}
  \end{align*}
  \begin{align*}
    \infer[\texttt{bc/[]}]
    {[a_1, \cdots, a_n] \leadsto 
    (\exactlyone{(v_1,\cdots, v_n)}, T, \{v_i \mapsto (1,0), \overline{v_i} \mapsto (1,0)\}_{i \leq n}, \eset )}
    {\text{fresh }v_1, \cdots, v_n}
  \end{align*}
  \begin{align*}
    \infer[\texttt{bc/ite}]
    {
      \ite{x}{e_t}{e_e} \leadsto 
      \begin{gathered}
        \big(
          (x \land \varphi_t \land R_t \land \conjneg{R_e}) 
            \lor (\overline{x} \land \varphi_e \land R_e \land \conjneg{R_t}), \\
          (x \land \gamma_t) \lor (\overline{x} \land \gamma_e),
          w_t \cup w_e,
          \eset) 
      \end{gathered}
    } {
      x \leadsto 
      (x, T, \eset, \eset) 
      & 
      e_t \leadsto 
      (\varphi_t, \gamma_t, w_t, R_t) 
      & 
      e_e \leadsto 
      (\varphi_e, \gamma_e, w_e, R_e)
    }
  \end{align*}
  \begin{align*}
    \infer[\texttt{bc/choose}]
    {\choose x {a_i \implies e_i} 
    \leadsto 
    \begin{gathered}
      \Big(\varphi \land \bigvee (a_i \land e_i 
      \land \bigwedge_{j \neq i} \conjneg{R_j}), 
      x \land \bigvee (a_i \land \gamma_i), \\
      \bigcup w_i, \bigcup R_i)
    \end{gathered}
  }
    {e \leadsto 
    (\varphi, T, \eset, \eset) 
    & 
    \forall \ i.  \ e_i \leadsto (\varphi_i, \gamma_i, w_i, R_i)}
  \end{align*}
  \begin{align*}
    \infer[\texttt{bc/<-}]
    {\bind{x}{e}{e'} \leadsto 
    (\varphi'[x \mapsto \varphi], 
    \gamma \land \gamma'[x \mapsto \varphi],
    w \cup w', 
    R \cup R')}
    {
      e \leadsto (\varphi, \gamma, w, R) 
      & 
      e' \leadsto (\varphi', \gamma', w', R')
    }
  \end{align*}
  }
  \end{mdframed}
  \caption{Boolean compilation rules of \texttt{dappl}. Compilation rules for
  $\land, \lor, \neg$ are omitted as they are straightforward.}
  \label{fig:dappl full bc}
  \end{figure}
\subsection{Proof of Theorem~\ref{thm:compiler correctness}}\label{appendix:dappl correctness}

The architecture of the proof is as follows. First, we prove that~\cref{thm:compiler correctness}
reduces to the following:

\begin{theorem}\label{thm:util correspondence}
  Let $\Gamma \proves e : \Giry \Bool$ a $\util$ expression. Let 
  $e \leadsto (\varphi, \gamma, w, R)$ via~\cref{fig:dappl full bc}.
  Let $\NoWt{\varphi}$ be the variables in $\varphi$ without a defined 
  weight; that is, variables whose literals are not in the domain of $w$, 
  and $\mathcal W (\varphi)$ be the set of maps $lits(\NoWt{\varphi}) \to \mathcal S$.

  Then, there exists a function $f : \denote{\Gamma} \to\mathcal W (\varphi)$
  making the following diagram commute:

  \[\begin{tikzcd}
	\llbracket\Gamma\rrbracket & {\mathcal W(\varphi)} \\
	& {\mathbb R}
	\arrow["f", from=1-1, to=1-2]
	\arrow["{\text{EU}_{(\varphi,\gamma, R)}}", from=1-2, to=2-2]
	\arrow["{\EU \circ \llbracket e \rrbracket}"'{pos=0.2}, from=1-1, to=2-2]
\end{tikzcd}\]
where, for $\overline w \in \mathcal W (\varphi)$,
\begin{equation}\label{eq:util correspondence normalized}
  \text{EU}_{(\varphi,\gamma, R)}(\overline w) 
    = \frac{\AMC(\varphi \land \gamma \land R, w \cup \overline w)_{\EU}}
      {\AMC(\gamma, w \cup \overline w)_{\Pr}}.
\end{equation}
In particular, if there are no \texttt{observe}s (conditioning) 
in the program,~\cref{eq:util correspondence normalized} reduces to
\begin{equation}
  \text{EU}_{(\varphi,T, R)}(\overline w) 
    = \AMC(\varphi\land R, w \cup \overline w)_{\EU}.
\end{equation}
\end{theorem}

Then, we prove~\cref{thm:util correspondence} to complete the proof.

\subsubsection{Reduction of Theorem~\ref{thm:compiler correctness}
to Theorem~\ref{thm:util correspondence}}

The key observation is the following Lemma:

\begin{lemma}[Policy space correspondence]\label{lemma:policy space}
  Let $\cdot \proves e : \tau$ be a \dappl{} program and 
  let $\mathcal A = C_1 \times \cdots C_k$ 
  be the policy space of $e$. 
  Let $e \leadsto \target$. 
  Let $X$ be the set of Boolean variables
  representing choices in $\varphi$. 
  Then:
  \begin{enumerate}
    \item There is an bijective correspondence $\cup_i C_i \to X$,
    \item which lifts into a canonical injective map $\iota :\mathcal A \to inst(X)$,
    \item such that for which for all $\pi\in \mathcal A$, $\iota(\pi)$ satisfies all
    ExactlyOne clauses.
  \end{enumerate} 
\end{lemma}

\begin{proof}
The bijective correspondence is the map assigning $\alpha \in \cup_i C_i$
  to the Boolean variable generated by 
  Boolean compilation of $C_i = [\alpha, \cdots, \alpha_n]$.
  This is injective, as for $\alpha,\beta \in \cup_i C_i$ such that $\alpha \neq \beta$,
  the Boolean compilation rules will always introduce fresh variable names for $\alpha$
  and $\beta$ that cannot coincide. It is surjective, as variables in $X$ are only 
  introduced in the \texttt{bc/[]} rule, which also introduces choices. We see in the 
  \texttt{bc/[]} rule that for $n$ many alternatives in a choice $C$, $n$ many variables
  are generated.

  Call such a bijection $b$. Then the canonical injective map $\iota$ simply maps $b$
  on each factor of a policy $\pi \in \mathcal A$. The fact that 
  $\iota(\pi)$ satisfies all ExactlyOne clauses is an induction on the Boolean
  compilation rules.
\end{proof}

We can generalize~\cref{lemma:policy space} to general judgements $\Gamma \proves e : \tau$,
in particular get a map from policies $\pi_{\Gamma}$ 
in the context policy space $\mathcal{A}_{\Gamma}$ (recall~\cref{def:context policy space})
to variables in the compiled Boolean formula corresponding to the free variables.

\begin{lemma}\label{lemma:context policy space}
  Let $\Gamma \proves e : \tau$ be a $\dappl$ expression.
  Let $\mathcal{A}_{\Gamma}$ be the context policy space. 
  Let $e \leadsto \target$.
  Let $\mathsf{Var}_{\mathsf{Choice}}(\varphi)$
  be the variables in $\varphi$ that correspond to names of type 
  $\Choice S$ for some $S$ in $\Gamma$; that is,
  $$
  \mathsf{Var}_{\mathsf{Choice}}(\varphi) = \prod_{\Choice S \in \Gamma} S.
  $$
  Then for each $\pi_{\Gamma} \in \mathcal{A}_{\Gamma}$ there is a bijective map $\rho_{\pi_{\Gamma}}:  \pi_{\Gamma} \to \mathsf{Var}_{\mathsf{Choice}}(\varphi)$ 
  on which the inverse is a valid substitution of $\varphi$.
\end{lemma}

\begin{proof}
  The map simply maps each component of $\pi_{\Gamma}$ to its corresponding variable.
  This is a valid substitution as we can always generated fresh Boolean variable names corresponding to the component, which has already been assumed WLOG in the \texttt{bc/choose} rule.
\end{proof}

We can now state the following Lemma:

\begin{lemma}\label{lemma:util square}
  Let $\Gamma \proves e : \tau$ a well-typed $\dappl$ program and let $\iota$ 
  be the canonical injective map from the policy space 
  as described in~\cref{lemma:policy space} and $\rho$ the canonical map 
  from the context policy space to $\mathsf{Var}_{\mathsf{Choice}}(\varphi)$
  as described in ~\cref{lemma:context policy space}. 
  Let $\mathcal{A}_{\Gamma}$ be the context policy space and $\pi_{\Gamma} \in \mathcal{A}_{\Gamma}$.
  Let $\pi$ be a valid policy of $e$ and let 
  $e \leadsto \target$ and 
  $e|_{\pi} \leadsto (\varphi_{\pi}, \gamma_{\pi}, w_\pi, R_\pi)$.

  Let $\mathsf{subst}$ denote the operation that takes in a formula $\varphi$,
  and substitutes
  variables in $\varphi$
  corresponding to variables $x$ of type $\Choice S$ for some $S$ 
  with either $\iota^{-1}(x)$ or $\rho^{-1}(x)$.

  Then the following square commutes up to equisatisfiability of Boolean formulae:
  \[\begin{tikzcd}
	{e} & {(\varphi, \gamma)} \\
	{e|_{\pi \cup \pi_{\Gamma}}} & {(\varphi_{\pi}, \gamma_{\pi})}
	\arrow["\text{Reduction, see~\cref{appendix:reduction soundness}}"', from=1-1, to=2-1]
	\arrow[squiggly, from=1-1, to=1-2]
	\arrow[squiggly, from=2-1, to=2-2]
	\arrow["{(\mathsf{subst}(-), \mathsf{subst}(-))}", from=1-2, to=2-2]
  \end{tikzcd}\]
  in which the $w,R$ are elided as $w \supseteq w_{\pi}$ and $R \supseteq R_\pi$. 
\end{lemma}

\begin{proof}
  The proof follows from an induction on the syntax of $e$. All cases are straightforward except for three cases:
  \begin{itemize}
    \item If $e = x$ and $x$ is of type $\Choice S$ for some $S$ in $\Gamma$, then 
    the reduction yields the empty program so the square is trivially satisfied.
    \item If $e = [\alpha_1, \cdots, \alpha_n]$, then $e|_{\pi \cup \pi_{\Gamma}}$ is 
    $\return \tt$, which compiles to $\top$.
    There are no free variables in $e$; so the substitution must come from the 
    policy space. By~\cref{lemma:policy space} we know this satisfies the ExactlyOne
    clause of $\varphi$; so it is $\top$ as well. $\gamma$ and $\gamma|_{\pi}$ are both
    $\top$. So we are done.
    \item If $e = \choose x {\alpha_i \implies e_i}$, then WLOG assume that $x$ is substituted for $\alpha_1$. Then $\varphi$ will simplify to $a_i \land e_i|_{\pi \cup \pi_{\Gamma}} \land \bigwedge_{j \neq 1} \conjneg{R_j}$. This is equisatisfiable to $e_i|_{\pi \cup \pi_{\Gamma}} $, at which point the IH kicks in and we are done.
  \end{itemize}
\end{proof}

To prove Theorem 4, consider a $\dappl$ program $e$ (that is, $\cdot \proves e : \tau$) and let $\pi$ be an arbitrary policy. We need not consider context policy spaces as the context is empty. 
Then by~\cref{lemma:util square} we can reduce to a valid $\util$ program. 
Onto this $\util$ program we can apply~\cref{thm:util correspondence}
to know that this is the correct expected utility. Then, by knowing that this is true in 
particular for the optimal policy, and knowing that $\text{bb}$ (Algorithm~\ref{algorithm:bb}) finds this optimal policy
via~\cref{thm:soundness of bb}, we are done.

\subsubsection{Proof of Theorem~\ref{thm:util correspondence}}

We state helpful lemmata, some of which are 
applications of Theorem~\ref{thm:amc invariant} to Propositions
proven in~\citet{holtzen2020scaling}.

\begin{lemma}[Independent conjunction of probabilities]
\label{lemma:ind conj prob}
  For $\varphi, \psi$ Boolean formulas that share no variables
  and any weight function $w : lits(\varphi) \cup lits(\psi) \to \mathcal S$,
  $\AMC(\varphi,w)_{\Pr}\times\AMC(\psi,w)_{\Pr} = \AMC(\varphi \land \psi,w)_{\Pr}$
\end{lemma}

\begin{lemma}[Inclusion-exclusion of probabilities]
\label{lemma:inc exc prob}
  For $\varphi, \psi$ Boolean formulas
  and any weight function $w : lits(\varphi) \cup lits(\psi) \to \mathcal S$,
  $\AMC(\varphi,w)_{\Pr}+ \AMC(\psi,w)_{\Pr} -
  \AMC(\varphi \land \psi, w)_{\Pr} = \AMC(\varphi \lor \psi,w)_{\Pr}$.
\end{lemma}

The following additional Lemma extends Lemma~\ref{lemma:ind conj prob}
to expected utilities.

\begin{lemma}[Independent conjunction of expected utilities]
\label{lemma:ind conj eu}
  For $\varphi, \psi$ Boolean formulas that share no variables
  and any weight function $w : lits(\varphi) \cup lits(\psi) \to \mathcal S$,
  \begin{align*}
      \AMC(\varphi,w)\times\AMC(\psi,w) &= \AMC(\varphi \land \psi,w)\\
    &=\AMC(\varphi,w)\cdot\AMC(\psi,w)_{\Pr}+\AMC(\psi,w)\cdot\AMC(\varphi,w)_{\Pr}  
  \end{align*}
  where $\times$ is multiplication in the expectation semiring $\mathcal S$ and $\cdot$
  is scalar multiplication distributing over $\mathcal S$. In particular, if 
  $\AMC(\varphi, w)_{\EU} = 0$,
  \begin{align*}
    [\AMC(\varphi,w)\times\AMC(\psi,w)]_{\EU} &= [\AMC(\varphi \land \psi,w)]_{\EU}\\
    &=\AMC(\varphi,w)_{\EU}\cdot\AMC(\psi,w)_{\Pr}.
  \end{align*}
\end{lemma}

\begin{proof}
  We observe that if $\varphi, \psi$ are disjoint, then 
  the models $m \models \varphi \land \psi$ are exactly the set
  $\{m_{\varphi} \cup m_{\psi}: 
  m_{\varphi} \models \varphi \texttt{ and } m_\psi \models \psi\}$;
  the proof follows.
\end{proof}

This Lemma extends Lemma~\ref{lemma:inc exc prob} for
expected utilities, but specifically for compiled formulas.

\begin{lemma}[Additive expected utility.]\label{lemma:inc exc eu}
  Let $\varphi, \psi$ be two programs such that
  the variables of each formula can be partitioned into 
  disjoint sets of probabilistic and reward variables
  $vars(\varphi) = P_X \cup R_X$ and $vars(\psi) = P_Y \cup R_Y$.

  Let $w$ be a weight function such that for literals
  $p \in lits(P_X) \cup lits(P_Y)$, $w(p)_{\EU} = 0$, 
  and for $r \in lits(R_X) \cup lits(R_Y)$,
  $w(r)_{\Pr} = 1$, identifying that probabilistic
  variables carry no utility and reward variables carry
  probability 1. 

  If $R_X, R_Y$
  are disjoint, then
  \begin{align*}
    [\AMC((\varphi \land \conjneg{R_Y} )
    \lor (\psi \land \conjneg{R_X}), w)]_{\EU}
    &= [\AMC(\varphi,w)]_{\EU} + [\AMC(\psi, w)]_{\EU}.
  \end{align*}
\end{lemma}

\begin{proof}
  Consider the models $m$ such that
    $m \models (\varphi \land \conjneg{R_Y} )
    \lor (\psi \land \conjneg{R_X}).$
  The models will will either:
  \begin{enumerate}
    \item model $\varphi \land \conjneg{R_Y}$
    but not $\psi \land \conjneg{R_X}$,
    \item model $\psi \land \conjneg{R_X}$
    but not $\varphi \land \conjneg{R_Y}$, or
    \item model both $\varphi \land \conjneg{R_Y}$
    and $\psi \land \conjneg{R_X}$.
  \end{enumerate}
  In Cases (1) and (2), $\psi \land \conjneg{R_X}$
  and $\varphi \land \conjneg{R_Y}$ respecitvely will not
  contribute any expected utility as they are not modeled.
  In Case (3), as any model will make all reward variables in
  $R_X$ and $R_Y$ false, it will contribute no expected utility.
  Thus in summary

  \begin{align*}
    [\AMC((\varphi \land \conjneg{R_Y})
    \lor (\psi \land \conjneg{R_X}), w)]_{\EU}
    &= \left[\sum_{m \models \varphi \land \conjneg{R_Y},
        m \not\models \psi \land \conjneg{R_X}} w(m)\right]_{\EU} \\
    &\quad + \left[\sum_{m \not\models \varphi \land \conjneg{R_Y},
        m \models \psi \land \conjneg{R_X}} w(m)\right]_{\EU} \\
    &\quad + \left[\sum_{m \models \varphi \land \conjneg{R_Y},
        m \models \psi \land \conjneg{R_X}} w(m)\right]_{\EU} \\
    &= \left[\sum_{m \models \varphi \land \conjneg{R_Y}} w(m)\right]_{\EU}
      + \left[\sum_{
        m \models \psi \land \conjneg{R_X}} w(m)\right]_{\EU} &(\star)\\
    &= \left[\sum_{m \models \varphi} w(m)\right]_{\EU}
      + \left[\sum_{
        m \models \psi} w(m)\right]_{\EU} &(\dagger)\\
    &= [\AMC(\varphi,w)]_{\EU} + [\AMC(\psi, w)]_{\EU}, 
  \end{align*}
  where $w(m)$ denotes the weight of a model defined as the product of its literals. $(\star)$ is the usage of the fact that if 
  $m \models \varphi \land \conjneg{R_Y}$,
  then either $m \models \psi \land \conjneg{R_X}$ or it does not.
  If it does, then $w(m)$ is zero as all reward variables are negated.
  If not, then $m \not\models \psi \land \conjneg{R_X}$ so the formula 
  $\psi \land \conjneg{R_X}$ contributes nothing. The analogous is true 
  for when $m \models \psi \land \conjneg{R_X}$.
  $(\dagger)$ uses the fact that the weight of a negated
  reward literal is $(1,0)$, the multiplicative unit, so
  it can be factored out when calculating $w(m)$.
\end{proof}

It is worthwhile to note that $[\AMC(\varphi,w)]_{\EU} = \EU[\varphi]$ as mentioned in Lemma~\ref{thm:amc invariant}. 
Since expected utility is indeed an expectation, we can use 
techniques such as taking conditional expectations $\EU[\varphi | \gamma]$. We reap the benefits of this observation in the proof of
Theorem~\ref{thm:compiler correctness}.

Now we prove intermediate results about the distribution of a $\util$ program.

\begin{definition}
  Let $\Gamma \proves e : \Giry \Bool$ a $\util$ program. 
  Then we can define a probability distribution $\Pr : \{\tt, \ff, \bot\} \to [0,1]$
  by:
  \begin{align*}
    \Pr(\tt) = \sum_{r \in \R} \denote{e} \denote{\Gamma}((\tt,r))
    &&
    \Pr(\ff) = \sum_{r \in \R} \denote{e} \denote{\Gamma}((\ff,r)),
  \end{align*}
  identically we can write
  \begin{align*}
    \Pr(\tt) = \sum_{v = (\tt,r) \in \R} \denote{e} \denote{\Gamma}(v)
    &&
    \Pr(\ff) = \sum_{v = (\ff,r) \in \R} \denote{e} \denote{\Gamma}(v).
  \end{align*}
  With an abuse of notation we write $\Pr[\denote{e} \denote{\Gamma}]$ for this.
\end{definition}

\begin{theorem}\label{thm:util pr correspondence}
  Let $\Gamma \proves e : \Giry \Bool$ a $\util$ expression. 
  Let 
  $e \leadsto (\varphi, \gamma, w, R)$ via~\cref{fig:dappl full bc}.
  Let $\NoWt{\varphi}$ be the variables in $\varphi$ (hence, in $\gamma$ as well) 
  without a defined 
  weight; that is, variables whose literals are not in the domain of $w$, 
  and $\mathcal W (\varphi)$ be the set of maps $lits(\NoWt{\varphi}) \to \mathcal S$.

  Then, there exists a function $f : \denote{\Gamma} \to\mathcal W (\varphi)$
  making the following diagrams commute:

  \[\begin{tikzcd}
	\llbracket\Gamma\rrbracket & {\mathcal W(\varphi)} \\
	& {\mathbb R}
	\arrow["f", from=1-1, to=1-2]
	\arrow["{\mathsf{Prob_{\varphi,\gamma, R}}}", from=1-2, to=2-2]
	\arrow["{\Pr \circ \llbracket e \rrbracket (-)(\tt)}"'{pos=0.2}, from=1-1, to=2-2]
  \end{tikzcd}
  \qquad 
  \begin{tikzcd}
    \llbracket\Gamma\rrbracket & {\mathcal W(\varphi)} \\
    & {\mathbb R}
    \arrow["f", from=1-1, to=1-2]
    \arrow["{\mathsf{Prob_{\overline{\varphi},\gamma, R}}}", from=1-2, to=2-2]
    \arrow["{\Pr \circ \llbracket e \rrbracket (-)(\ff)}"'{pos=0.2}, from=1-1, to=2-2]
  \end{tikzcd}\]
  where, for $\overline w \in \mathcal W(\varphi)$,
  \begin{equation}
    \mathsf{Prob}_{\varphi, \gamma, R} (\overline w)
       = \AMC(\varphi \land \gamma \land R, w \cup \overline w)_{\Pr}.
  \end{equation}
  and
  \begin{equation}
    \mathsf{Prob}_{\overline\varphi, R} (\overline w)
       = \AMC(\overline \varphi \land \gamma \land R, w \cup \overline w)_{\Pr}.
  \end{equation}
  That is, computes the \emph{unnormalized probabilities of $e$ returning $\tt$ or $\ff$.}
\end{theorem}

\begin{proof}
  Recall that $\Gamma$ must only hold variables of 
  type $\Bool$ by~\cref{lemma:util context}. \
  So $\denote{\Gamma}$ will hold maps of variables to distributions $\mathbf{TT}$ or $\mathbf{FF}$. 
  Also note that 
  variables in $\NoWt{\varphi}$ are precisely the free variables of $e$; 
  these variables must be defined in $\Gamma$. 
  Thus we can define $f$ to be the map mapping $\{x \mapsto \mathbf{TT}\}$ to $\{x \mapsto (1,0), \overline x \mapsto (0,0)\}$ and $\{x \mapsto \mathbf{FF}\}$ to $\{x \mapsto (0,0), \overline x \mapsto (1,0)\}$.

  We prove that this is exactly what we need. This is done by simultaneous induction on syntax. 
  \begin{itemize}[leftmargin=*]
    \item If $e = \tt, \ff, \flip \theta$, then we are done after a simple evaluation.
    
    \item If $e = x$, then $x \leadsto (x, \top, \eset, \eset)$. 
    If $x \mapsto \mathbf{TT} \in \denote{\Gamma}$ then $\Pr \circ \denote{x} (\Gamma) (\tt,0) = 1$. Identically 
    \begin{equation}
      \AMC(x \land \top, w \cup (x \mapsto (1,0), \overline x \mapsto (0,0)))_{\Pr} 
      \AMC(x, w \cup (x \mapsto (1,0), \overline x \mapsto (0,0)))_{\Pr} 
      = 1.
    \end{equation}
    If $x \mapsto \mathbf{FF} \in \denote{\Gamma}$ then $\Pr \circ \denote{x} (\Gamma) (\ff,0) = 1$. Identically 
    \begin{equation}
      \AMC(\overline x, w \cup (x \mapsto (0,0), \overline x \mapsto (1,0)))_{\Pr} = 1.
    \end{equation}

    \item For $e = \return P$, for $g \in \denote{\Gamma}$,
    $\denote{e}g = \denote{P}g$. 
    Identically $e$ and $P$ compiles to the same Boolean formula. By the IH we are done.

    \item For $e = \reward k {e'}$, for $g \in \denote{\Gamma}$, we observe that
      $$\Pr \circ (\denote{e} g) (\tt) 
        = \sum_{r \in R} (\denote{e}g)(r)
        = \sum_{r \in R} (\denote{e'}g)(r) = \Pr \circ (\denote{e'} g)(\tt).$$
    Identically, we observe that
    $$
      \AMC(\varphi \land \gamma \land R \land r_k, w)_{\Pr}
      = \AMC(\varphi \land \gamma \land R, w)_{\Pr}
    $$
    as a straightforward application of~\cref{lemma:ind conj prob}. By the IH we are done. The case is identical for $\Pr \circ \denote e \denote \gamma  (\ff)$.

    \item For $e = \ite x {e'} {e''}$, for $g \in \denote{\Gamma}$, we case on $g(x)$. Assume it is $\mathbf{TT}$; the other case is symmetrical. Then 
    $\denote{e} g = \denote{e'} g$. 

    On the other hand this implies that $f(x \mapsto \mathbf{TT}) = \{x \mapsto (1,0), \overline x \mapsto (0,0)\}$. Consider that, using notation from \texttt{bc/ite}
    and writing 
    $$\varphi = (x \land \varphi_t \land R_t \land \conjneg{R_e}) 
            \lor (\overline{x} \land \varphi_e \land R_e \land \conjneg{R_t})
            \land (x \land \gamma_t) \lor (\overline x \land \gamma_e),$$
    we get, after simplification,
    \begin{align}
      &\AMC(\varphi, w_t \cup w_e \cup f(x \mapsto \mathbf{TT}))_{\Pr} \\
      &=\AMC(x \land \varphi_t \land \gamma_t \land R_t \land \conjneg{R_e})_{\Pr} \\
      &=\AMC(\varphi_t \land R_t \land \gamma_t \land \conjneg{R_e} ) \\
      &=\AMC(\varphi_t \land R_t \land \gamma_t)
    \end{align}
    where (22) is due to~\cref{lemma:inc exc prob} and (23), (24) is due to~\cref{lemma:ind conj prob}. At this point the IH works and we are done. The case is identical for $\Pr \circ (\denote e g)  (\ff)$.

  \item For $e = \observe x {e'}$, for $g \in \denote{\Gamma}$, we case on $g(x)$.
 
  If $g(x) = \mathbf{TT}$, then $\denote{e} g = \denote{e'} g$. Also, by following \texttt{bc/obs}, we get that 
      $$\AMC(\varphi \land \gamma \land x, w_t \cup w_e \cup f(x \mapsto \mathbf{TT}))_{\Pr}
      =\AMC(\varphi \land \gamma, w_t \cup w_e)_{\Pr}$$
      by~\cref{lemma:ind conj prob} and we are done by IH.

  If $g(x) = \mathbf{FF}$, then $\denote{e} g = \pmb{\bot}$. Identically we get
    $$\AMC(\varphi \land \gamma \land x, w_t \cup w_e \cup f(x \mapsto \mathbf{FF}))_{\Pr}
      =0$$
  concluding the case.
  \item For $e = \bind x {b} {e'}$, for $g \in \denote{\Gamma}$, we define some custom notation:
  \begin{itemize}
    \item $\mathcal D = \denote{e} g$,
    \item $\mathcal B = \denote{b} g$,
    \item $\mathcal{E}_{\mathbf{TT}} = (\denote{e'} g) \cup (x \mapsto \mathbf{TT})$,
    \item $\mathcal{E}_{\mathbf{FF}} =( \denote{e'} g )\cup (x \mapsto \mathbf{FF})$,
    \item $b \leadsto (\varphi_b, T, R_b)$, and
    \item $e' \leadsto (\varphi_{e'}, T, R_{e'})$.
  \end{itemize}
  Then, we can identify, via our denotational semantics (refer to~\cref{appendix:util semantics}), 
  expand 
  $\Pr \circ \mathcal D$:
  \begin{align*}
    \Pr \circ \mathcal D (\tt)
      &= \sum_{r \in \R} \mathcal D (\tt,r) 
      &\text{[Definition]}\\
      &=\sum_{r \in \R} 
        \left[
          \sum_{s \in \R} \mathcal B [(\tt,s)] \mathcal{E}_{\mathbf{TT}} (\tt, s-r)
          + \sum_{s \in \R} \mathcal B [(\ff,s)] \mathcal{E}_{\mathbf{FF}} (\tt, s-r)
        \right]
      &\text{[Unfolding]}\\
      &=\sum_{r \in \R} 
        \left[
          \sum_{s \in \R} \mathcal B [(\tt,s)] \mathcal{E}_{\mathbf{TT}} (\tt, r)
          + \sum_{s \in \R} \mathcal B [(\ff,s)] \mathcal{E}_{\mathbf{FF}} (\tt, r)
        \right]
      &\text{[Invariance]}\\
      &= 
          \sum_{s \in \R} \mathcal B [(\tt,s)] \sum_{r \in \R}\mathcal{E}_ {\mathbf{TT}} (\tt, r)
          + \sum_{s \in \R} \mathcal B [(\ff,s)] \sum_{r \in \R}\mathcal{E}_{\mathbf{FF}} (\tt, r)
      &\text{[Distributivity]} \\
      &= 
        \Pr \circ \mathcal B (\tt)
        \Pr \circ \mathcal{E}_ {\mathbf{TT}} (\tt)
        + \Pr \circ \mathcal B (\ff)
        \Pr \circ \mathcal{E}_ {\mathbf{FF}} (\tt)
      &\text{[Definition]}\\
      &= 
        \AMC(\varphi_b \land \gamma_b)_{\Pr}
        \AMC(\varphi_{e'}[x / T] \land \gamma_{e'}[x/T])_{\Pr}
      &\\
      &\quad
        + \AMC(\overline{\varphi_b} \land \gamma_b)_{\Pr}
        \AMC(\varphi_{e'}[x / F] \land \gamma_{e'}[x/F])_{\Pr}
      &\text{[IH]}\\
      &= 
        \AMC((\varphi_b \land \gamma_b \land \varphi_{e'}[x / T] \land \gamma_{e'}[x/T])
      &\\
      &\quad 
          \lor 
          (\overline{\varphi_b}\land \varphi_{e'}[x / F]\land \varphi_{e'}[x / F] \land \gamma_{e'}[x/F]))_{\Pr}
      &\text{[Lemmas]}\\
      &= 
        \AMC(\varphi_{e'}[x/\varphi_b] \land \gamma_b \land \gamma_e'[x/\varphi_b])_{\Pr}
      &\text{[Equisatisfiability]}
  \end{align*}
    which concludes the proof.
  \end{itemize}
\end{proof}

We can also additionally prove a similar Lemma.

\begin{lemma}\label{lemma:util pr bot}
  Let $\Gamma \proves e : \Giry \Bool$ a $\util$ expression. Let 
  $e \leadsto (\varphi, \gamma, w, R)$ via~\cref{fig:dappl full bc}.
  Let $f$ be the function defined in~\cref{thm:util pr correspondence}. 
  Then for $g \in \denote{\Gamma}$,
  \begin{equation}
    \Pr \circ (\denote{e}g)[\tt \texttt{ or } \ff] = \AMC(\gamma, w \cup \overline w)_{\Pr}.
  \end{equation}
\end{lemma}

\begin{proof}
  We observe that $\Pr \circ (\denote{e}g)[\tt]$ 
  and $\Pr \circ (\denote{e}g)[\ff]$
  are disjoint events. Then by~\cref{thm:util pr correspondence}
  we get that 
  \begin{align*}
    &\Pr \circ (\denote{e}g)[\tt] +  \Pr \circ (\denote{e}g)[\tt]&\\
    &= \AMC(\varphi \land \gamma \land R)_{\Pr} + \AMC(\overline{\varphi} \land \gamma \land R)_{\Pr}&\\
    &\AMC(\varphi \land \gamma \lor \overline{\varphi} \land \gamma)_{\Pr}
    &\text{[Lemmas]}\\
    &\AMC(\gamma)_{\Pr}
    &[\Pr[\varphi \lor \overline{\varphi}] = 1]
  \end{align*}
  which concludes the proof.
\end{proof}

With this we prove, automatically, a corollary:

\begin{corollary}\label{cor:util pr correspondence}
  Let $\Gamma \proves e : \Giry \Bool$ a $\util$ expression. Let 
  $e \leadsto (\varphi, \gamma, w, R)$ via~\cref{fig:dappl full bc}.
  Let $f$ be the function defined in~\cref{thm:util pr correspondence}. 
  Let $g \in \denote{\Gamma}$.
  Then
  \begin{equation}
    \Pr \circ (\denote{e}g)[\tt \mid \texttt{not }\bot]
    = \frac
      {\Pr \circ (\denote{e}g)[\tt]}
      { \Pr \circ (\denote{e}g)[\tt \texttt{ or } \ff]}
     = \frac{\AMC(\varphi \land \gamma, w \cup \overline w)_{\Pr}}{\AMC(\gamma, w \cup \overline w)_{\Pr}}
  \end{equation}
\end{corollary}

This Corollary is essentially a denotational version of the main theorem proven in~\citet{holtzen2020scaling}.
Now onto the good part.

\begin{proof}[Proof of~\cref{thm:util correspondence}]

  We claim that the same $f$ used for~\cref{thm:util pr correspondence} suffices. 
  Let $g \in \denote{\Gamma}$.
  By an application of~\cref{lemma:util pr bot}, it suffices to prove the unnormalized
  case. That is, let $\EU_{\mathsf{unn}}$ be the unnormalized expected utility, where,
  in contrast to~\cref{def:eu util},
  \begin{equation}
    \EU_{\mathsf{unn}}(\denote{e} g) (b)
      = \sum_{r \in \R} r \times (\denote{e}g)(b,r).
  \end{equation}
  It suffices to prove 
  \begin{equation}
    \EU_{\mathsf{unn}}(\denote{e} g)(\tt)
      = \AMC(\varphi \land \gamma \land R)_{\EU},
    \quad
    \EU_{\mathsf{unn}}(\denote{e} g)(\ff)
      = \AMC(\overline \varphi \land \gamma \land R)_{\EU}.    
  \end{equation}
  We induct on syntax once more.

  \begin{itemize}[leftmargin=*]
    \item If $e = \tt, \ff, \flip \theta, x,$ and $\return P$, 
    the expected utility is always zero,
    which proves the theorem.

    \item For $e = \reward k {e'}$, we observe that
    $$\denote{e}g = \lambda v. \begin{cases}
      (\denote{e'}g)(b,s-k) & v = (b,s)\\
      (\denote{e'}g)(v)  & \text{else}
    \end{cases} $$
    so in particular, writing $\mathcal D = \denote{e'}g$,
    \begin{align*}
      \EU_{\mathsf{unn}} \circ( \denote{e} g) (\tt)
        &= \sum_{r \in \R} r \times \mathcal D [(\tt, r - k)] &\\
        &= \sum_{r \in \R} (r + k) \times \mathcal D [(\tt, r)] &\text{[Rewriting]} \\
        &= \sum_{r \in R} r \times \mathcal D [(\tt, r)] 
          + k \sum_{r \in R}\mathcal D [(\tt, r)]. &\text{[Arithmetic]}
    \end{align*}
    Let $e \leadsto (\varphi, \gamma, w, R \cup r_k)$ as per \texttt{bc/reward}. Let $\overline w = f \denote{\Gamma}$. We see that
    \begin{align*}
      EU_{(\varphi, \gamma, R \cup \{r_k\})} (\overline w) 
        &=\AMC(\varphi \land \gamma \land R \land r_k, w \cup \overline w)_{\EU} &\\ 
        &=\AMC(\varphi \land \gamma \land R)_{\EU} \times \AMC(r_k)_{\Pr}
        + \AMC(\varphi \land \gamma \land R)_{\Pr} \times \AMC(r_k)_{\EU}
          &\text{[\cref{lemma:ind conj eu}]}\\
        &=\AMC(\varphi \land \gamma \land R)_{\EU}
          + \AMC(\varphi \land \gamma \land R)_{\Pr} \times k  
          &\text{[Evaluation]}\\
        &=\sum_{r \in R} r \times \mathcal D [(\tt, r)] + \AMC(\varphi \land \gamma \land R)_{\Pr} \times k
          &\text{[IH]}\\
        &=\sum_{r \in R} r \times \mathcal D [(\tt, r)] + k\sum_{r \in R} \mathcal D [(\tt, r)]  &\text{[\cref{thm:util pr correspondence}]}
    \end{align*}
    and the case is identical for $\ff$. 
    
    \item For $e = \ite x {e'} {e''}$, we case on $g(x)$. 
    WLOG assume $g(x) = \mathbf{TT}$ as the other case is symmetric. 
    Then 
    $$
    \denote{e} g = \denote{e'}g.
    $$
    Furthermore 
    $$e \leadsto ((x \land \varphi_t \land R_t \land \conjneg{R_e}) 
    \lor (\overline{x} \land \varphi_e \land R_e \land \conjneg{R_t}),
    (x \land \gamma_t) \lor (\overline x \land \gamma_e),
    w_t \cup w_e,
          \eset);$$
    writing $\mathcal D = \denote{e'} g$, $\overline w = f g$, and $\varphi, \gamma$ for the unnormalized and normalizing Boolean formulae we see that
    \begin{align*}
      EU_{(\varphi, \gamma, \eset)}(\overline w) 
        &=\AMC(\varphi \land \gamma)_{\EU}&\\
        &=\AMC((x \land \varphi_t \land \gamma_t \land R_t \land \conjneg{R_e}))
          &\text{[Evaluation, Lemmas]}\\
        &=\AMC((\varphi_t \land R_t \land \gamma_t))_{\EU}
        &\text{[\cref{lemma:ind conj eu},\cref{lemma:ind conj prob}]}\\
        &=\EU_{\mathsf{unn}}(\denote{e'} g) (\tt)
    \end{align*}
    which concludes the case via~\cref{appendix:util semantics}. 
    The case for $\ff$ is identical.

  \item For $e = \observe x {e'}$, we again case on $g(x).$ For the remainder of this case let $\mathcal D = \denote{e'} g$, $\overline w = f g$.
  
  \begin{itemize}
    \item If $g(x) = \mathbf{TT}$, then $\denote{e} g = \denote{e'} g$. Also, by following \texttt{bc/obs}, we get that 
    $$EU_{(\varphi, \gamma, R)}(\overline w) 
      = \AMC(\varphi \land R \land \gamma \land x)_{\EU} = \AMC(\varphi \land \gamma\land R)_{\EU}
      $$
    by~\cref{lemma:ind conj prob} and we are done by IH.
    \item If $g(x) = \mathbf{FF}$, then $g \denote{\Gamma} = \pmb{\bot}$. So $\EU \circ (\denote{e} g) (\tt) = 0$. Also, by following \texttt{bc/obs}, we get that 
    $$\AMC(\varphi \land x \land R \land \gamma )_{\Pr} 
    = 0$$
    by~\cref{lemma:ind conj prob} and we are done. 
  \end{itemize}
  the proof for $\ff$ is identical.
  \item For $e = \bind x {b} {e'}$, assume $\mathcal D, \mathcal B, \mathcal{E}_{\mathbf{TT}}, \mathcal{E}_{\mathbf{FF}}$ from the proof of~\cref{thm:util pr correspondence}.
  Then we can derive 
  \begin{align*}
    \EU_{\mathsf{unn}} \circ \mathcal D (\tt)
      &= \sum_{r \in \R} r \times \mathcal D (\tt,r) 
      &\text{[Definition]}\\
      &=\sum_{r \in \R} r \times
        \left[
          \sum_{s \in \R} \mathcal B [(\tt,s)] \mathcal{E}_{\mathbf{TT}} (\tt, s-r)
          + \sum_{s \in \R} \mathcal B [(\ff,s)] \mathcal{E}_{\mathbf{FF}} (\tt, s-r)
        \right]
      &\text{[Unfolding]}\\
      &=\sum_{r \in \R} \sum_{s \in \R} r \times
         \mathcal B [(\tt,s)] \mathcal{E}_{\mathbf{TT}} (\tt, s-r)
      &\\
      &\quad + 
        \sum_{r \in \R }\sum_{s \in \R} r \times \mathcal B [(\ff,s)] \mathcal{E}_{\mathbf{FF}} (\tt, s-r)
      &\text{[Rewriting]}\\
      &=\sum_{r \in \R} \sum_{s \in \R} (r + k) \times
         \mathcal B [(\tt,r)] \mathcal{E}_{\mathbf{TT}} (\tt, k)
      &\\
      &\quad + 
        \sum_{r \in \R }\sum_{s \in \R} (r + k) \times \mathcal B [(\ff,r)] \mathcal{E}_{\mathbf{FF}} (\tt, k)
      &\text{[Rewriting]}\\
      &=\sum_{r \in \R} \sum_{s \in \R} r \times 
         \mathcal B [(\tt,r)] \mathcal{E}_{\mathbf{TT}} (\tt, k)
      + \sum_{r \in \R} \sum_{s \in \R} k \times 
         \mathcal B [(\tt,r)] \mathcal{E}_{\mathbf{TT}} (\tt, k)
      &\\
      &\quad + 
        \sum_{r \in \R} \sum_{s \in \R} r \times 
          \mathcal B [(\ff,r)] \mathcal{E}_{\mathbf{FF}} (\tt, k)
        + \sum_{r \in \R} \sum_{s \in \R} k \times 
          \mathcal B [(\ff,r)] \mathcal{E}_{\mathbf{FF}} (\tt, k)
      &\text{[Rewriting]}\\
      &=\AMC(\varphi_b \land R_b \land \gamma_b)_{\EU} 
        \times \AMC(\varphi_{e'}[x/T] \land \gamma_{e'}[x/T] \land R_{e'})_{\Pr}
      &\\
      & \quad + \AMC(\varphi_b \land R_b \land \gamma_b)_{\Pr} 
        \times \AMC(\varphi_{e'}[x/T] \land \gamma_{e'}[x/T] \land R_{e'})_{\EU}
      &\\
      &\quad + \AMC(\overline{\varphi_b} \land R_b \land \gamma_b)_{\EU} 
        \times \AMC(\varphi_{e'}[x/F] \land \gamma_{e'}[x/F] \land R_{e'})_{\Pr}
      &\\
      & \quad + \AMC(\overline{\varphi_b} \land R_b \land \gamma_b)_{\Pr} 
        \times \AMC(\varphi_{e'}[x/F] \land \gamma_{e'}[x/F] \land R_{e'})_{\EU}
      &\text{[IH]}
  \end{align*}
  at which point routine applications of~\cref{lemma:inc exc eu,lemma:ind conj eu}
  complete the proof. The case for $\ff$ is identical.
  \end{itemize}

\end{proof}

\newpage
\section{Supplementary material for Section \ref{sec:pineappl}}

\subsection{$\pineappl$ Ergonomics and Syntactic Sugar}
\label{appendix:pineappl-sugar}

We endow the core of $\pineappl$ with a few pieces of syntactic sugar to aid in
modeling.

\subsubsection{Support for discrete distributions} 
Similar to $\dappl$, we extend $\pineappl$ with support for discrete
distributions, rather than a simple flip, a program can sample from a discrete
set of events, with either a uniform prior, or custom priors on each event
(provided that those priors sum to 1). This is accomplished via a one-hot
encoding, similar to \citep{holtzen2020scaling}. This also introduces a
predicate $is$ which tests for the presence of a categorical-variable.  

\subsubsection{Support for multiple queries} 
As a result of $\pineappl$ compiling the statements of a program to a boolean
formula, it is trivial to extend the program with the ability to make multiple
queries over the same set of statements. Full $\pineappl$ programs can end with 
any number of query expressions, and the results are returned as a list. 

\subsubsection{Support for MMAP as a terminal query}
In addition to using MMAP as a first-class primitive, it can also be useful to 
obtain the map state of some variables at the end of the program. This can easily
done using BBIR directly at the end of the program.

\subsubsection{Support for bounded loops}
We implement bounded loops in $\pineappl$ as a hygenic macro expansion of the
code inside the loop. Loops relax $\pineappl$'s demand for entirely fresh names
at the source-level; after expansion, the compiler will enforce the freshness
constraint with hygiene, potential introduction of join points for loops that
occurr within the branches of an if-statement, and a global renaming pass to
ensure that all names bound in the loop are referenced appropriately in later
code. \cref{fig:loop-pineappl-sugar} is a simple $\pineappl$ program that uses
a loop and \cref{fig:loop-pineappl-desugar} shows its expansion. Note, that
\pineapplcode{pr(a)} is rewritten to \pineapplcode{pr(a2)} in the renaming pass
to ensure that the query refers to the ``latest'' value of \pineapplcode{a}.
\cref{fig:loop-pineappl-ite-sugar} is a $\pineappl$ program containing an
if-statement where each branch has a loop. Since the loop expansion binds fresh
names, variables bound in the loop must be explicitly joined at the end of the
if-statement, and then global renaming pass utilizes the joined name for
subsequent uses of the variable.  Clearly, loops expand to syntactically valid
$\pineappl$ programs. Determining the last binding introduced by a loop for
join-points and rewriting can be done as a lightweight analysis at expansion time.

\begin{figure}
  \begin{subfigure}[t]{0.45\textwidth}
    \begin{pineapplcodeblock}
      a = flip 0.5;
      loop 3 {
        tmp = flip 0.1;
        a = a || tmp; 
      }
      pr(a)
    \end{pineapplcodeblock}
    \caption{A simple $\pineappl$ program with a loop}
    \label{fig:loop-pineappl-sugar}
  \end{subfigure}
  \begin{subfigure}[t]{0.45\textwidth}
    \begin{pineapplcodeblock}
      a = flip 0.5;
      tmp0 = flip 0.1;
      a0 = a || tmp0; 
      tmp1 = flip 0.1;
      a1 = a0 || tmp1;
      tmp2 = flip 0.1;
      a2 = a1 || tmp2;
      pr(a2)
    \end{pineapplcodeblock}
    \caption{An expansion and renaming of the program from (a).}
    \label{fig:loop-pineappl-desugar}
  \end{subfigure}
  \begin{subfigure}[t]{0.45\textwidth}
    \begin{pineapplcodeblock}
      x = flip 0.5;
      y = flip 0.5;
      if x {
        loop 2 {
          tmp = flip 0.3;
          y = y && tmp;
        }
      }
      else {
        loop 3 {
          tmp = flip 0.7;
          y = y || tmp; 
        }
      }
      pr(y)
    \end{pineapplcodeblock}
    \caption{A $\pineappl$ program with loops in both branches of an if statement.}
    \label{fig:loop-pineappl-ite-sugar}
  \end{subfigure}
  \begin{subfigure}[t]{0.45\textwidth}
    \begin{pineapplcodeblock}
      x = flip 0.5;
      y = flip 0.5;
      if x {
        tmp0 = flip 0.2;
        y0 = y && tmp0;
        tmp1 = flip 0.2;
        y1 = y0 && tmp1;
      }
      else {
        tmp2 = flip 0.7;
        y2 = y || tmp2;
        tmp3 = flip 0.7;
        y3 = y2 || tmp3;
        tmp4 = flip 0.7;
        y4 = y3 || tmp4;
      }
      tmp_j = (x && tmp1) || (!x && tmp4);
      y_j = (x && y1) || (!x && y4);
      pr(y_j)
    \end{pineapplcodeblock}
    \caption{An expansion, introduction of join-points, and renaming of the program
    from (c).}
    \label{fig:loop-pineappl-ite-desugar}
\end{subfigure}
\caption{Examples of loop expansion in $\pineappl$}
\label{fig:loop-pineappl}
\end{figure}

%%% Local Variables:
%%% mode: LaTeX
%%% TeX-master: "../../oopsla-appendix.tex"
%%% End:

% \input{appendices/pineappl/e_relation.tex}
\subsection{The $\leadsto_E$ relation for $\pineappl{}$}
\label{appendix:pineappl bc}

See~\cref{fig:pineappl-e}.

\begin{figure}
\begin{mdframed}
  {\footnotesize
  \begin{align*}
    \infer[\texttt{e/x}]
    {
      \texttt{x} \leadsto_E x
    }
    {\mathrm{fresh} \; x}
    \qquad
    \infer[\texttt{e/tt}]
    {
      \texttt{tt} \leadsto_E \top
    }
    {}
    \qquad
    \infer[\texttt{e/ff}]
    {
      \texttt{ff} \leadsto_E \bot
    }
    {}
  \end{align*}
  \begin{align*}
    \infer[\texttt{e/}\land]
    {
      \texttt{e1} \land \texttt{e2}
    \leadsto_E e_1 \land e_2
    }
    { \texttt{e1} \leadsto_E e_1 & \texttt{e2} \leadsto_E e_2}
    \qquad
    \infer[\texttt{e/}\lor]
    {
      \texttt{e1} \lor \texttt{e2} \leadsto_E e_1 \lor e_2
    }
    { \texttt{e1} \leadsto_E e_1 & \texttt{e2} \leadsto_E e_2}
    \qquad
    \infer[\texttt{e/}\neg]
    {
      \neg \texttt{e} \leadsto_E \neg e
    }
    { \texttt{e} \leadsto_E e}
  \end{align*}
  }
\end{mdframed}
\caption{The $\leadsto_E$ relation for $\pineappl{}$ expressions.}
\label{fig:pineappl-e}
\end{figure}

\subsection{Full Boolean compilation rules for $\pineappl{}$}
\label{appendix:pineappl-compl-complete}

See~\cref{fig:pineappl-compl-complete}.

\begin{figure}
\begin{mdframed}
  {\footnotesize
  \begin{align*}
    \infer[\texttt{bc/flip}]
    {
        ((\texttt{x = flip} \ \theta, \mathcal{F}, w) \leadsto
        (\{(x, f)\} \cup \mathcal{F},
        w \cup \{x \mapsto (1, 1), f \mapsto (\theta, 1-\theta)\}))
    }
    {\mathrm{fresh} \; f}
  \end{align*}
  \begin{align*}
    \infer[\texttt{bc/assn}]
    {
      (\texttt{x = e}, \mathcal{F}, w) \leadsto
      (\{(x,  \varphi)\} \cup \mathcal{F}, w \cup \{x \mapsto (1, 1)\})
    }
    {
      \texttt e \leadsto_E \varphi
    }
  \end{align*}
  \begin{align*}
    \infer[\texttt{bc/seq}]
    {
        (\texttt{$s_1$; $s_2$}, w) \leadsto (\mathcal{F}'', w'')
    }
    {
        (s_1, \mathcal{F}, w) \leadsto (\mathcal{F}', w')
        & (s_2, \mathcal{F}', w') \leadsto (\mathcal{F}'', w'')
    }
  \end{align*}
  \begin{align*}
    \infer[\texttt{bc/if}]
    {
      (\texttt{if e \{$s_1$\} else \{$s_2$\}}, \mathcal{F}, w) \leadsto
    (\{(x_i, (\chi \land \varphi_i \lor \neg\chi \land \psi_i))\} \cup \mathcal{F}, w_1 \cup w_2)
    % could include extra parens around and clauses, but and binds tighter than or
    }
    {
        \texttt e \leadsto_E \chi
        & (s_1, \mathcal{F}, w) \leadsto (\{(x_i, \varphi_i)\} \cup \mathcal{F}, w_1)
        & (s_2, \mathcal{F}, w) \leadsto (\{(x_i, \psi_i)\} \cup \mathcal{F}, w_2)
    }
  \end{align*}
  \begin{align*}
    \infer[\texttt{bc/mmap}]
    {
      (\vec{\texttt{m}}\texttt{ = mmap }\vec{\texttt{x}}, \mathcal{F}, w) \leadsto
      (\mathcal{F} \cup \{(m_i, k_i) \}, w \cup w_{M})
    }
    {
      \mathrm{fresh} \; k_i
      & \vec A = MMAP(\{\bigwedge_{(x, \varphi) \in \mathcal{F}} x \leftrightarrow \varphi, \eset\}, \vec{\texttt{x}}, w)
      &
      w_{M} = \{m_i \mapsto (1,1), k_i \mapsto A_i \}
    }
  \end{align*}
  \begin{align*}
    \infer[\texttt{bc/mmap/with}]
    {
      (\vec{\texttt{m}}\texttt{ = mmap }\vec{\texttt{x}} \ \pineapplcode{ with \{e\}}, \mathcal{F}, w) \leadsto
      (\mathcal{F} \cup \{(m_i, k_i) \}, w \cup w_{M})
    }
    {
      \mathrm{fresh} \; k_i
      & \texttt e \leadsto_E \psi
      & \vec A = MMAP(\{\bigwedge_{(x, \varphi) \in \mathcal{F}} x \leftrightarrow \varphi, \psi\}, \vec{\texttt{x}}, w)
      &
      w_{M} = \{m_i \mapsto (1,1), k_i \mapsto A_i \}
    }
  \end{align*}
  \begin{align*}
    \infer[\texttt{bc/pr}]
    {
      \texttt{s; Pr(e)} \leadsto_P (\varphi \land \left( \bigwedge_{(x, \varphi) \in \mathcal{F}} x \leftrightarrow \varphi \right), \top,w)
    }
    {
      (\texttt{s}, \eset, \eset) \leadsto (\mathcal F, w)
      &
      \texttt{e} \leadsto_E \varphi
    }
  \end{align*}
  \begin{align*}
    \infer[\texttt{bc/pr/with}]
    {
      \pineapplcode{s; Pr(e1) with \{e2\}} \leadsto_P (\varphi \land \left(\bigwedge_{(x, \varphi) \in \mathcal{F}} x \leftrightarrow \varphi \right), \psi, w)
    }
    {
      (s, \eset, \eset) \leadsto (\mathcal F, w)
      &
      e \leadsto_E \varphi
    }
  \end{align*}
  }
\end{mdframed}
\caption{Full Boolean compilation rules for $\pineappl{}$ statements and programs.}
\label{fig:pineappl-compl-complete}
\end{figure}

\subsection{Proof of Theorem~\ref{thm:pineappl correctness}}
\label{appendix:proof-pineappl-correctness}

The proof is by way of simulation.

\begin{definition}[$\sim$]
Let $\mathcal D$ be a distribution over assignments to variables,
$\mathcal F$ a set of formulae of shape $x \leftrightarrow \varphi$,
and $w$ a weight function
of literals in $\mathcal F$ to the reals.
Let the variables in $\mathcal F$ be a superset of those in $\mathcal D$.
Then we define $D \sim (F,w)$ if and only if
for all $\sigma \in \dom(D)$,
\begin{equation}
  D(\sigma) = \left[ \prod_{\ell \in \sigma} w(\ell) \right]
    \AMC_{\R} \paren{
      \bigwedge_{(x, \varphi) \in \mathcal F} (x \leftrightarrow \varphi)|_{\sigma} , w
    }
  .
\end{equation}
\end{definition}

On this relation we can define a helpful Lemma:

\begin{lemma}\label{lemma:pineappl simulate expressions}
  Let $D \sim (F,w)$.
  Let $e$ be a Boolean expression in $\pineappl$ on which $\Pr_{\mathcal D} (e)$ is well-defined.
  Let $e \leadsto_E \varphi$. Then the following holds:
  \begin{equation}
    \Pr_{\mathcal D} (e) = \AMC\left(\varphi \land \paren{\bigwedge_{(x, \varphi) \in \mathcal F} x \leftrightarrow \varphi}, w\right).
  \end{equation}
\end{lemma}

\begin{proof}
  The proof is a straightforward induction on the syntax of $e$.
\end{proof}

Now, we have the necessary machinery to prove the theorem.
Let $s;q$ be a $\pineappl{}$ program. We first prove the following.

\begin{theorem}\label{thm:pineappl statements invariant}
  Let $D, \mathcal F, w$ such that $D \sim (\mathcal F, w)$.
  Let $(s, D) \Downarrow D'$ and $(s, \mathcal F, w) \leadsto (\mathcal F', w')$.
  Then $D' \sim (\mathcal F', w')$.
\end{theorem}

\begin{proof}
  We induct on syntax.
  \begin{itemize}[leftmargin=*]
    \item If $s = \pineapplcode{x = flip }\ \theta$, then observe that $\mathcal F' =\mathcal F \cup \{x \leftrightarrow f_\theta\}$.
    For any trace where $x \mapsto \top$, we get
    \begin{align}
      D'(\sigma \cup \{x \mapsto \top\})
        &= \theta \times D(\sigma)\\
        &= w(f_{\theta}) \times
        \left[ \prod_{\ell \in \sigma} w(\ell) \right]
        \AMC_{\R} \paren{
          \bigwedge_{(x, \varphi) \in \mathcal F} (x \leftrightarrow \varphi)|_{\sigma} , w
        }
        \\
        &= \AMC_{\R}((x \leftrightarrow f_{\theta})\mid_{x = \top}, w) \times
        \left[ \prod_{\ell \in \sigma \cup \{x \mapsto \top\} } w(\ell) \right]
        \AMC_{\R} \paren{
          \bigwedge_{(x, \varphi) \in \mathcal F} (x \leftrightarrow \varphi)|_{\sigma} , w
        }
        \\
        &=
        \left[ \prod_{\ell \in \sigma \cup \{x \mapsto \top\} } w(\ell) \right]
        \AMC_{\R} \paren{
          (x \leftrightarrow f_{\theta}) \land \bigwedge_{(x, \varphi) \in \mathcal F} (x \leftrightarrow \varphi)|_{\sigma \cup \{x \mapsto \top\}} , w
        }
    \end{align}
    where (49) is the inductive hypothesis and (51) used~\cref{lemma:ind conj prob}.
    The case when $x \mapsto \bot$ is symmetrical.

    \item If $s = \pineapplcode{x = e}$, let $e \leadsto_E \chi$.
    Then observe that $\mathcal F' =\mathcal F \cup \{x \leftrightarrow \varphi\}$.
    Then for any trace where $e[\sigma] = \top$, we get
    \begin{align}
      D'(\sigma \cup \{x \mapsto \top\})
        &= \Pr_{\mathcal D}[e] \times D(\sigma)\\
        &= \AMC(\chi \land \bigwedge_{(x, \varphi) \in \mathcal F} (x \leftrightarrow \varphi), w) \nonumber\\
        &\quad \times
        \left[ \prod_{\ell \in \sigma} w(\ell) \right]
        \AMC_{\R} \paren{
          \bigwedge_{(x, \varphi) \in \mathcal F} (x \leftrightarrow \varphi)|_{\sigma} , w
        }
        \\
        &= \AMC_{\R}((y \leftrightarrow \chi)\mid_{y = \top} \land \bigwedge_{(x, \varphi) \in \mathcal F} (x \leftrightarrow \varphi), w) \nonumber\\
        &\quad \times
        \left[ \prod_{\ell \in \sigma \cup \{x \mapsto \top\} } w(\ell) \right]
        \AMC_{\R} \paren{
          \bigwedge_{(x, \varphi) \in \mathcal F} (x \leftrightarrow \varphi)|_{\sigma} , w
        }
        \\
        &=
        \left[ \prod_{\ell \in \sigma \cup \{x \mapsto \top\} } w(\ell) \right]
        \AMC_{\R} \paren{
         \left[(x \leftrightarrow \varphi)| \land \bigwedge_{(x, \varphi) \in \mathcal F} (x \leftrightarrow \varphi)\right]\mid_{\sigma \cup \{x \mapsto \top\}} , w
        }
    \end{align}
    where (53) follows from~\cref{lemma:pineappl simulate expressions}.
    (55) is valid because $\varphi$ consists exclusively of variables occuring in $\sigma$,
    so $\bigwedge_{(x, \varphi) \in \mathcal F} (x \leftrightarrow \varphi)|_{\sigma}$ has no variables in common.
    Furthermore the restriction of $\bigwedge_{(x, \varphi) \in \mathcal F} (x \leftrightarrow \varphi)$ to only those
    that satisfy $\sigma \cup \{x \cup \top\}$ eliminates the larger.x
    The case when $x\mapsto \bot$ is symmetrical.
    \item If $s = \pineapplcode{s1 ; s2}$, the proof is straightforward and is omitted.
    \item If $s = \pineapplcode{if e \{s1\} else  \{s2\}}$, let $e \leadsto_E \varphi$, then let
    \begin{itemize}
      \item $e \leadsto_E \chi$,
      \item $(s_1, \mathcal D) \Downarrow D_1$,
      \item $(s_2, \mathcal D) \Downarrow D_2$,
      \item $(s_1, \mathcal F, w) \leadsto (\mathcal F \cup \{x_i \leftrightarrow \varphi_i\}, w_1)$, and
      \item $(s_2, \mathcal F, w) \leadsto (\mathcal F \cup \{x_i \leftrightarrow \psi_i\}, w_2)$.
    \end{itemize}
    Without loss of generality assume that $D_1$ and $D_2$ are over the same domain.
    This is possible because if there exists some variable $v$ such that $v \in \sigma$ in which $D_1(\sigma)$ is defined but $D_2(\sigma)$ is not,
    then we can extend all $\tau \in \dom(D_2)$ with $v$, and vice versa. Similarly without loss of generality assume that the $w_1$ and $w_2$ have the same domain.

    Let $I = \bigwedge_{(x, \varphi) \in \mathcal F} (x \leftrightarrow \varphi)$.
    Then we can deduce, for some $\sigma$,
    \begin{align}
      \mathcal D' (\sigma)
        &= p \times \mathcal D_1(\sigma) + (1-p) \times \mathcal D_2(\sigma)
        \\
        &= \AMC(\chi \land I, w) \times \mathcal D_1(\sigma) + \AMC(\overline \chi \land I, w) \times \mathcal D_2(\sigma)
        \\
        &= \AMC(\chi \land I, w) \times
          \left[ \prod_{\ell \in \sigma } w(\ell) \right]
            \AMC_{\R} \paren{
              \{x_i \leftrightarrow \varphi_i\} \mid_{\sigma} \land I|_{\sigma} , w
          }
        \\
        &\quad+ \AMC(\overline \chi \land I, w) \times
          \left[ \prod_{\ell \in \sigma} w(\ell) \right]
            \AMC_{\R} \paren{
              \{x_i \leftrightarrow \varphi_i\} \mid_{\sigma} \land I|_{\sigma} , w
          }
        \\
        &= \left[ \prod_{\ell \in \sigma} w(\ell) \right] \nonumber \\
        &\quad \times
            \AMC_{\R} \paren{
              \chi \land I \land \paren{\{x_i \leftrightarrow \varphi_i\} \mid_{\sigma} \land I|_{\sigma}}
              \lor
              \overline \chi \land I \land \paren{\{x_i \leftrightarrow \varphi_i\} \mid_{\sigma} \land I|_{\sigma}} , w
            }.
    \end{align}
    Consider the formula within the AMC in (60). We observe that
    \begin{align}
      &\chi \land I \land \paren{\{x_i \leftrightarrow \varphi_i\} \mid_{\sigma} \land I|_{\sigma}}
        \lor
        \overline \chi \land I \land \paren{\{x_i \leftrightarrow \varphi_i\} \mid_{\sigma} \land I|_{\sigma}}\\
      &\quad =
      \bigwedge_{(x, \varphi) \in \mathcal F} (x \leftrightarrow \varphi)|_{\sigma} \land \paren{\chi \land \paren{\{x_i \leftrightarrow \varphi_i\} \mid_{\sigma}}
        \lor
        \overline \chi \land \paren{\{x_i \leftrightarrow \varphi_i\} \mid_{\sigma}}}\\
      &\quad =
      \bigwedge_{(x, \varphi) \in \mathcal F} (x \leftrightarrow \varphi)|_{\sigma} \land
        \paren{\paren{\{x_i \leftrightarrow \chi \land \varphi_i\} \mid_{\sigma}}
        \lor
        \paren{\{x_i \leftrightarrow \overline \chi \land  \varphi_i\} \mid_{\sigma}}}\\
      &\quad =
      \bigwedge_{(x, \varphi) \in \mathcal F} (x \leftrightarrow \varphi)|_{\sigma} \land
        \paren{\paren{\{x_i \leftrightarrow \chi \land \varphi_i\lor \overline \chi \land  \varphi_i \} \mid_{\sigma}}}
    \end{align}
    as desired by repeated usage of~\cref{lemma:inc exc prob,lemma:ind conj prob}.
    \item If $s = \vec{\texttt{m}}\texttt{ = mmap }\vec{\texttt{x}}$,
    we defer the proof to the next case, with the specialization that $e = \tt$.
    \item If $s = \vec{\texttt{m}}\texttt{ = mmap }\vec{\texttt{x}} \ \pineapplcode{ with \{e\}}$,
    then it suffices to show that, for $e \leadsto_E \psi$,
    \begin{equation*}
      MMAP_{\mathcal D} (\vec x \mid e) = MMAP(\{\bigwedge_{\mathcal F} x_i \leftrightarrow \varphi_i, \psi\}, \vec x, w).
    \end{equation*}
    We observe that
    \begin{align}
      MMAP_{\mathcal D} (\vec x \mid e)
      & = \argmax_{\sigma \in inst(\vec x)} \mathcal D (\sigma  \mid e)\\
      & = \argmax_{\sigma \in inst(\vec x)} \frac{\mathcal D (\sigma \land e)}{\Pr_{\mathcal D}[e]}
      \\
      & = \argmax_{\sigma \in inst(\vec x)} \frac{\AMC(\bigwedge \mathcal F |_\sigma \land \psi, w)}{\AMC_{\mathbb R}(\psi, w)}\\
      & = MMAP(\{\bigwedge_{\mathcal F} x_i \leftrightarrow \varphi_i, \psi\},\vec x, w)
    \end{align}
    as desired.
  \end{itemize}
\end{proof}

Now, we can finally prove~\cref{thm:pineappl correctness}.

\begin{proof}[Proof of~\cref{thm:pineappl correctness}]
  Let $s;q$ be a $\pineappl{}$ program.
  Let $(s, \eset) \Downarrow \mathcal D$ and $(s, \eset,\eset) \leadsto (\mathcal F, w)$.
  By~\cref{thm:pineappl statements invariant} we know that $D \sim (\mathcal F, w)$.
  It suffices to prove correctness for $q = \pineapplcode{Pr(e1) with \{e2\}}$ as the other case is identical with $\texttt{e2}  = \tt$.
  We observe that, as an application of~\cref{lemma:pineappl simulate expressions}
  \begin{align}
    \frac{\Pr_{\mathcal D}[e_1 \land e_2]}{\Pr_{\mathcal D}[e_2]}
    &=\frac{\AMC_{\mathbb R}(\varphi \land (\land_{(x, \varphi) \in \mathcal F} x \leftrightarrow \varphi) \land \psi,w)}
      {\AMC_{\mathbb R}(\psi \land (\land_{(x, \varphi) \in \mathcal F} x \leftrightarrow \varphi),w)}
  \end{align}
  which completes the proof.
\end{proof}

% \section{Supplementary material for Section \ref{sec:eval}}
% \input{appendices/pineappl/scaling_example.tex}

\end{document}